\documentclass[twoside,11pt]{article}

\usepackage{blindtext}

 \usepackage[abbrvbib, preprint]{jmlr2e}

\usepackage{jmlr2e}
\usepackage{array}
\usepackage{amsmath}
\usepackage{graphicx}
\usepackage{subcaption}
\usepackage{bm}
\usepackage{float} 
\usepackage{xcolor}

\usepackage{lastpage}
\jmlrheading{23}{2025}{1-\pageref{LastPage}}{1/21; Revised 5/22}{9/22}{21-0000}{Yihan Zhao, Wenqing Su, Ying Yang}

\ShortHeadings{High-dimensional Continual Ridge Regression}{Zhao, Su and Yang}
\firstpageno{1}

\begin{document}

\title{High-dimensional Asymptotics of Generalization Performance in Continual Ridge Regression}

\author{\name Yihan Zhao \email zhao-yh23@mails.tsinghua.edu.cn \\
       \addr Department of Mathematical Sciences\\
       Tsinghua University\\
       Beijing, 100084, China
       \AND
       \name Wenqing Su \email suwenqing@snnu.edu.cn \\
        \addr School of Mathematics and Statistics \\
       Shaanxi Normal University\\
       Xi'an, 710119, China
        \AND
        \name Ying Yang \email yangying@tsinghua.edu.cn \\
     \addr Department of Statistics and Data Science\\
       Tsinghua University\\
       Beijing, 100084, China\\~\\
       }
\editor{My editor}

\maketitle
\begin{abstract}
Continual learning is motivated by the need to adapt to real-world dynamics in tasks and data distribution while mitigating catastrophic forgetting. Despite significant advances in continual learning techniques, the theoretical understanding of their generalization performance lags behind. This paper examines the theoretical properties of continual ridge regression in high-dimensional linear models, where the dimension is proportional to the sample size in each task. Using random matrix theory, we derive exact expressions of the asymptotic prediction risk, thereby enabling the characterization of three evaluation metrics of generalization performance in continual learning: average risk, backward transfer, and forward transfer. Furthermore, we present the theoretical risk curves to illustrate the trends in these evaluation metrics throughout the continual learning process. Our analysis reveals several intriguing phenomena in the risk curves, demonstrating how model specifications influence the generalization performance. Simulation studies are conducted to validate our theoretical findings.
\end{abstract}


\section{Introduction}
Continual learning, also termed incremental learning or lifelong learning, trains models on sequential tasks with evolving data distributions. The ideal learner is obtained using both previous data and current data to ensure that the current model is adapted to the current task while maintaining the performance on previous tasks. However, due to memory limitations, the previous data may not be available when the new task arrives, which results in performance reduction on previous tasks. This phenomenon is referred to as catastrophic forgetting \citep{MCCLOSKEY1989109, WOS:A1995RL52000001}. To mitigate the influence of catastrophic forgetting, substantial progress has been made in continual learning by designing various techniques, leading to impressive success in practical applications \citep{9349197, 10444954}. Despite their methodological innovations, theoretical understandings of continual learning methods remain scarce even in relatively simple models. When a provably effective continual learning method is applied to a statistical model, how can we characterize its generalization performance (the ability to maintain stable predictions across previous tasks while adapting to new prediction tasks)? Specifically, how do model specifications, such as model complexity and task similarity, influence the generalization performance under different patterns of task dynamics?  Addressing this question comprehensively relies on an insightful theory to capture the stepwise performance in continual learning, which is underdeveloped in most existing work.

To bridge this gap, we begin with the linear regression framework, since it provides a principled starting point that balances analytical tractability and practical utility. We focus on the theoretical properties of continual ridge regression \citep{pmlr-v232-li23b} in high-dimensional random-designed linear models. As a standard continual learning method in regression models, continual ridge regression adds $l_2$-regularization terms to the loss functions to constrain changes in the regression coefficients. Motivated by the question mentioned above, our primary concern is to establish a rigorous characterization of the stepwise generalization performance in continual ridge regression. Building upon this characterization, we aim to interpret how the generalization performance is influenced by key model specifications. Besides, the severity of catastrophic forgetting can be modulated to some extent by the choice of regularization parameters. This allows us to quantitatively investigate how catastrophic forgetting affects the generalization performance in continual ridge regression. Our contributions are summarized as follows. 

\begin{itemize}
    \item Using asymptotic random matrix theory, we derive the exact expression of asymptotic prediction risk on a single prediction task in a high-dimensional regime where the parameter dimension grows proportionally with the training sample size in each task. 
    This result explicitly characterizes how the asymptotic prediction risk depends on model complexity and task similarity, where the model complexity is quantified by the ratio of parameter dimension to sample size, and task similarity is characterized by the joint empirical spectral distribution of task-specific covariance matrices.

    \item 
    We establish the asymptotic behavior of three evaluation metrics for continual ridge regression---average risk, backward transfer, and forward transfer---and characterize their precise dependence on model complexity and task similarity in the high-dimensional setting. These metrics capture the key capabilities of interest in continual learning: average risk measures the overall performance across tasks, backward transfer evaluates how learning new tasks influences performance on previously learned tasks, and forward transfer quantifies how learning a current task improves performance on future tasks.

    \item We demonstrate our theoretical findings through three representative examples with different structures of task-specific covariance matrices.
    In each example, the asymptotic risk curves are derived to visualize the stepwise generalization performance in the procedure of continual ridge estimation. We find that the asymptotic risk curves behave differently according to the dynamics of task-specific covariance matrices and the choice of regularization parameters. In addition, simulations are also conducted to verify our theoretical risk curves.
\end{itemize}

\subsection{Related Works}
Over the past few years, researchers have developed a massive number of continual learning methods. Roughly speaking, these methods can be categorized into three groups, including 
regularization-based approach \citep{pnas.1611835114, 8107520, NEURIPS2018_f31b2046, 10.5555/3495724.3496032}, replay-based approach \citep{10.5555/3295222.3295393, 8100070, DBLP:conf/iclr/RiemerCALRTT19, 9412614, 10.5555/3540261.3542497} and architecture-based approach \citep{ramesh2021model, pmlr-v162-gurbuz22a, WOS:000870759102067}. Regularization-based approaches introduce regularization terms into loss functions to balance the information from old tasks and new tasks. \cite{pnas.1611835114} used weighted regularization methods to overcome catastrophic forgetting in neural networks.
\cite{8107520} applied function regularization methods to enable learning without forgetting in convolutional neural networks. Replay-based approaches store a few previous data in a memory buffer and replay them in future tasks. \cite{10.5555/3295222.3295393} proposed the Gradient Episodic Memory (GEM) model, which alleviated forgetting while allowing beneficial transfer of knowledge to previous tasks. \cite{10.5555/3540261.3542497} proposed the Gradient based Memory EDiting (GMED) framework, which edited the examples stored in the replay memory. Architecture-based approaches construct task-specific parameters to prevent catastrophic forgetting.
\cite{ramesh2021model} grouped the associated tasks and split the learning capacity across sets of synergistic tasks. \cite{pmlr-v162-gurbuz22a} employed  connection rewriting in sparse neural networks to create new plastic paths that reused existing knowledge on novel tasks. 

In recent years, several theoretical studies have been conducted on continual learning. For example, \cite{pmlr-v139-lee21e} studied continual learning in the teacher-student setup to find out the theoretical reasons for interference between tasks. \cite{li2022} established sample complexity and generalization error bounds for new tasks in continual representation learning problems. \cite{yang2023optimizing} proposed a theoretical analysis of a SPCA-based continual learning algorithm using high-dimensional statistics. \cite{pmlr-v235-wen24f} analyzed  the theoretical properties of contrastive continual learning methods. The results in the above studies are not as explicit as ours since they focus on model-free prediction problems. 
By contrast, some theoretical works have focused on results within concrete models, especially regression models. \cite{pmlr-v178-evron22a} theoretically characterized the worst-case catastrophic forgetting in over-parameterized linear regression models. Different from our work, they assumed that the collection of tasks has cyclic task orderings, so that the mechanism of forgetting is quite different. \cite{pmlr-v232-li23b} presented a fixed-design analysis of the $l_2$-regularized continual learning algorithm for linear regression models. However, they only consider a two-task setting and provide risk bounds, which cannot adequately explain the generalization performance of continual ridge regression. 
Very recently, \cite{pmlr-v235-zhao24n} initiated a statistical analysis of a family of generalized $l_2$-regularized continual learning algorithms for linear regression models. Although they derived precise results of estimation error, their analysis did not provide a comprehensive evaluation of the generalization ability of estimators in continual learning. \cite{goldfarb2025} demonstrated that overparameterization can mitigate forgetting by considering a two-task latent space regression model. While their work revealed the influence of model complexity on generalization performance, it did not account for the impact of task similarities.

\subsection{Organizations and Notations}
The rest of this paper is organized as follows. In Section \ref{60}, we introduce the settings of continual risk regression, formally define the continual ridge estimator and compute its prediction risk. In Section \ref{61}, we present our main result of asymptotic prediction risk and provide three concrete examples of covariance structures to illustrate our main theorem. In Section \ref{62}, we derive asymptotic risk curves to evaluate the performance of continual ridge regression, and conduct experiments to verify our results. We also discuss the procedure of choosing regularization parameters in continual learning frameworks. In Section \ref{63}, we conclude the paper and provide possible extensions. Technical proofs are included in the Appendix.

Throughout this article, $\mathbb{C}^+=\{z\in \mathbb{C},\text{Im}(z)>0\}$, and $\lfloor x \rfloor$ denotes the greatest integer less than or equal to $x$. Let $\Vert \cdot \Vert:=\Vert \cdot \Vert_2$ denote the Euclidean norm of a real vector, and $\langle \cdot,\cdot \rangle$ denote the inner product induced by the Euclidean norm. For any symmetric real matrix $A \in \mathbb{R}^{p \times p}$, $\lambda_{max}(A)$ and $\lambda_{min}(A)$ denote the largest and smallest eigenvalues of $A$ respectively.  Let $\{\lambda_t\}_{t\ge 1 } \in \mathbb{R}^{p}$ be a sequence of positive numbers, and $\{A_t\}_{t\ge 1 } \in \mathbb{R}^{p \times p}$ be a sequence of positive definite matrices.
To handle boundary cases in recursive expressions uniformly (e.g., in Lemma \ref{46} and Theorem \ref{10}), we define the empty scalar product as 
 $\prod_{t=k}^{l}\lambda_t=1$ if $k>l$, and the empty matrix product as $A_lA_{l-1}\cdots A_k=I_p$ if $k>l$. 

\section{Continual Ridge Regression}
\label{60}
Consider a continual learning problem with a sequence of $T$ tasks. In the $t$-th task ($t=1,\cdots,T$), we observe a dataset $\mathcal{D}_t=\{(x_{t,i}, y_{t,i})\in \mathbb{R}^p\times \mathbb{R} \}_{i=1}^{n_t}$, where $x_{t,i}$ is a random feature vector, $y_{t,i}$ is the corresponding response, and $n_t$ denotes the sample size of $\mathcal{D}_t$. We assume that the feature vectors in different tasks are independent, and a shared linear regression model governs all tasks, such that
\begin{align*}
y_{i,t} = x_{i, t}^\top \beta + \epsilon_{i, t}, \quad t=1,\cdots,T,\quad  i =1,\cdots, n_t,
\end{align*}
where the feature vectors in the $t$-th task $\{x_{t,i}\}_{i=1}^{n_t}$ are i.i.d. with $\mathbb{E}(x_{t, i}) = 0, \text{Cov}(x_{t, i}) = \Sigma_t$, and the noise terms
$\{\epsilon_{t,i}\}_{t=1,\cdots,T, i =1,\cdots, n_t,} $ across all tasks are i.i.d. with $\mathbb{E}(\epsilon_{t, i}) = 0, \text{Var}(\epsilon_{t, i}) = \sigma^2$.  For analytical convenience, we rewrite the linear model for each task $t$ as the matrix form
\begin{align}
y_t=X_t\beta + \epsilon_t, \quad t=1,\cdots,T,
\label{1}
\end{align}
where $X_t=(x_{1,t}, \cdots, x_{n_t,t})^\top \in \mathbb{R}^{n_t \times p}, y_t = (y_{1,t}, \cdots, y_{n_t,t})^\top \in \mathbb{R}^{n_t}$ and $\epsilon_t 
= (\epsilon_{1,t}, \cdots, \epsilon_{n_t, t})^\top 
\\ \in \mathbb{R}^{n_t}$.

In continual learning frameworks, the estimation of $\beta$ is updated upon the arrival of each new task. Let $\hat{\beta}_t$ denote the estimator after the arrival of task $t$, then $\hat{\beta}_t$ is computed using the current dataset $\mathcal{D}_{t}$ while retaining no direct access to previous datasets $\mathcal{D}_1, \cdots, \mathcal{D}_{t-1}$ due to memory constraints. Instead of previous datasets, historical information is incorporated through the previous estimator $\hat{\beta}_{t-1}$. Continual ridge regression updates the estimator by fitting the current dataset with ridge regularization to constrain the change from previous estimator $\hat{\beta}_{t-1}$ to the new estimator $\hat{\beta}_{t}$. Specifically, the updating rule is
\begin{align}
\hat{\beta}_t= \mathop{\arg\min} \limits_\beta \{\frac{1}{n_t}\Vert X_t\beta - y_t\Vert^2 + \lambda_t\Vert\beta - \hat{\beta}_{t-1}\Vert^2\},
\quad t=1,\cdots,T,
\label{59}
\end{align}
where $\lambda_t>0$ is the ridge tuning parameter at step $t$. The solution of continual ridge estimator is 
\begin{align*}
\hat{\beta}_t=  \hat{\beta}^{\rm ridge}_{t} + \lambda_t(\hat{\Sigma}_t+\lambda_tI_p)^{-1}\hat{\beta}_{t-1},\quad t=1,\cdots,T,
\end{align*}
where $\hat{\beta}^{\rm ridge}_{t}=(\hat{\Sigma}_t+\lambda_t I_p)^{-1}(\frac{1}{n_t}X_t^\top y_t)$ is the ridge estimator using dataset $\mathcal{D}_t$, and $\hat{\Sigma}_t=\frac{1}{n_t}X_t^\top X_t$ denotes the sample covariance matrix of $X_t$. If we initialize $\hat{\beta}_0=0$, then the explicit form of the continual ridge estimator is given by
\begin{equation}
\hat{\beta}_t=  \hat{\beta}^{\rm ridge}_{t} + A_t\hat{\beta}^{\rm ridge}_{t-1} + A_tA_{t-1}\hat{\beta}^{\rm ridge}_{t-2} + \cdots + A_tA_{t-1}\cdots A_2\hat{\beta}^{\rm ridge}_{1}, \quad t=1,\cdots,T,
\label{2}
\end{equation}
where $A_t=\lambda_t(\hat{\Sigma}_t+\lambda_tI_p)^{-1}$. 

To evaluate the generalization performance of continual learning methods in linear regression model, we first define the out-of-sample excess risk conditional on training data $X=(X_1^\top, \cdots, X_T^\top)^\top$. If the input test data $x_0$ satisfies $\mathbb{E}(x_{0}) = 0, \text{Cov}(x_{0}) = \Sigma_0$, and the response variable is denoted as $y_0$, then the corresponding prediction risk of an estimator $\hat{\beta}$ is
\begin{align*}
R_X(\hat{\beta}; \beta,\Sigma_0)=\mathbb{E}[(x_0^\top \hat{\beta}-y_0)^2|X]-\mathbb{E}[(x_0^\top \beta-y_0)^2|X]=\mathbb{E}[(x_0^\top \hat{\beta}-x_0^\top \beta)^2|X].
\end{align*}

Our primary interest is to evaluate the generalization performance of continual ridge estimator given by (\ref{2}). Advised by \cite{10444954}, we consider the following three aspects of evaluation rules:

\begin{itemize}
    \item \textbf{Overall performance}, which is evaluated by (weighted) average risk
    \begin{align*}
\bar{R}_X(\hat{\beta}_T; \beta)=\sum_{t=1}^T\omega_tR_X(\hat{\beta}_T; \beta,\Sigma_t),
\end{align*}
where $\omega_t>0$ for all $t$ and $\sum_{t=1}^T \omega_t=1$.
The weights are chosen according to the prior knowledge of the prediction task. By default, the test data at step $T$ comes from one of the $T$ training tasks with probability proportional to the sample size of each training dataset. Under this assumption, we may choose $\omega_t=n_t/\sum_{k=1}^T n_k$.
    \item \textbf{Memory stability}, which is evaluated by backward transfer
     \begin{align*}
{BWT}_X(\hat{\beta}_T; \beta)=\sum_{t=1}^{T-1}\tilde{\omega}_t\big(R_X(\hat{\beta}_T; \beta,\Sigma_t)-R_X(\hat{\beta}_t; \beta,\Sigma_t)\big),
\end{align*}
where $\tilde{\omega}_t=n_t/\sum_{k=1}^{T-1} n_k$. Backward transfer on a sequence of tasks quantifies the average influence of learning new tasks on old tasks. If the backward transfer is less than 0, learning new tasks are beneficial for predicting on old tasks. Otherwise, learning new tasks in this task sequence produce performance reduction on old tasks.
    \item \textbf{Learning plasticity}, which is evaluated by forward transfer
     \begin{align*}
{FWT}_X(\hat{\beta}_T; \beta)=\sum_{t=2}^{T}\bar{\omega}_t\big(R_X(\hat{\beta}_t; \beta,\Sigma_t)-R_X(\hat{\beta}^{\rm ridge}_{t}; \beta,\Sigma_t)\big),
\end{align*}
where $\bar{\omega}_t=n_t/\sum_{k=2}^{T} n_k$. Forward transfer on a sequence of tasks quantifies the average influence of historical information on learning current tasks. Note that the risk of continual ridge estimator measures the performance when learning with historical information, while the risk of ridge estimator at current step measures the performance when learning without historical information. If the forward transfer is less than 0, using historical information is beneficial for learning current tasks. Otherwise, a continual learning method is worse than simply learning current tasks.
\end{itemize}

From the preceding definitions, we observe that these evaluation metrics can be expressed as linear functions of the prediction risk associated with specific continual ridge estimators (the standard ridge estimator corresponds to the special case where the task number $T=1$). Therefore, it remains to analyze the prediction risk of continual ridge estimators.

Note that the prediction risk has a bias-variance decomposition \citep{10.1214/21-AOS2133}, 
\begin{align*}
R_X(\hat{\beta}; \beta,\Sigma_0)=B_X(\hat{\beta}; \beta,\Sigma_0)+V_X(\hat{\beta}; \beta,\Sigma_0),
\end{align*}
where
\begin{equation}
B_X(\hat{\beta}; \beta,\Sigma_0)=\big(\mathbb{E}(\hat{\beta}|X)-\beta\big)^\top\Sigma_0\big(\mathbb{E}(\hat{\beta}|X)-\beta\big),
\label{3}
\end{equation}
\begin{equation}
V_X(\hat{\beta}; \beta,\Sigma_0)={\rm Tr}\big(\text{Cov}(\hat{\beta}|X)\Sigma_0\big).
\label{4}
\end{equation}
To end this section, we derive the expressions of bias and variance terms for continual ridge estimators.

\begin{lemma}
Under model (\ref{1}), the bias and variance terms of continual ridge estimator (\ref{2}) are respectively
\begin{align*}
B_X(\hat{\beta}_T; \beta,\Sigma_0)&=\beta^\top A_1A_2\cdots A_T\Sigma_0 A_T\cdots A_2A_1\beta,\\
V_X(\hat{\beta}_T; \beta,\Sigma_0)&=\sigma^2\sum_{t=1}^{T} \frac{1}{\lambda_tn_t}{\rm Tr}\big[A_TA_{T-1}\cdots A_{t+1}(A_t-A_t^2)A_{t+1}\cdots A_{T-1}A_T\Sigma_0\big].
\end{align*}
\label{46}
\end{lemma} 
The details of calculations above can be found in Appendix \ref{58}. 

\section{Asymptotics of Prediction Risk}
\label{61}
In this section, we investigate the asymptotic behaviors of the prediction risk for continual ridge estimators. The computation of high-dimensional asymptotic risk is intimately connected to the limiting spectral properties of sample covariance matrices. Random matrix theory (RMT) is a powerful tool to characterize these properties \citep{book1, Yao_Zheng_Bai_2015}. To establish our theory, we begin by reviewing key concepts from RMT.

Let $A \in \mathbb{R}^{p \times p}$ be a symmetric matrix, then the empirical spectral distribution (ESD) of $A$ is defined as $F_A(x)=p^{-1}\sum_{i=1}^p \textbf{1}\{\lambda_i(A) \le x\}$. For any distribution $F$ supported on $[0, \infty)$, its Stieltjes transform is defined as 
\begin{align*}
    m_F(z)=\int \frac{1}{x-z}dF(x) , z \in \mathbb{C} \backslash [0, \infty).
\end{align*}
If $A$ is positive semi-definite, the Stieltjes transform of $F_A$ is given by 
\begin{align*}
     m_A(z)=\int \frac{1}{x-z}dF_A(x)=\frac{1}{p}{\rm Tr}(A-zI_p)^{-1}, z \in \mathbb{C} \backslash [0, \infty).
\end{align*}
Assume now that the data matrix is generated as $X=Z\Sigma^{1/2}$, where $\Sigma\in \mathbb{R}^{p \times p}$ is a deterministic positive semidefinite matrix, and $Z \in \mathbb{R}^{n \times p}$ has i.i.d. entries with zero mean and unit variance. The well-known Mar{\v{c}}enko-Pastur theorem \citep{Marenko1967DISTRIBUTIONOE, SILVERSTEIN1995175} states that, if $p,n \rightarrow \infty$ such that $p/n \rightarrow \gamma \in (0, \infty)$, and $F_\Sigma$ converges weakly to some limit distribution $H$, then the ESD of sample covariance matrix $\hat{\Sigma}=\frac{1}{n}X^\top X$ converges weakly to a limiting distribution $F_{H,\gamma}$. The  characterization of 
$F_{H,\gamma}$ relies on its Stieltjes transform.
Consider the corresponding matrix $\tilde{\Sigma}=\frac{1}{n}XX^\top $. By the definition of Stieltjes transform, $m_{\hat\Sigma}(z)$ and $m_{\tilde\Sigma}(z)$ are linked by
\begin{align*}
  m_{\hat\Sigma}(z)=\frac{1}{\gamma} {m}_{\tilde\Sigma}(z)+\frac{1-\gamma}{\gamma z}.
\end{align*}
Note that, the Mar{\v{c}}enko-Pastur theorem also ensures that $F_{\tilde{\Sigma}}$ has a limiting distribution, denoted as $\tilde{F}_{H,\gamma}$. Let $m_{H,\gamma}(z)$ and $\tilde{m}_{H,\gamma}(z)$ be the Stieltjes transform of $F_{H,\gamma}$ and $\tilde{F}_{H,\gamma}$ respectively, then they are liked by
\begin{align}
  {m}_{H,\gamma}(z)=\frac{1}{\gamma} \tilde{m}_{H,\gamma}(z)+\frac{1-\gamma}{\gamma z}. 
  \label{39}
\end{align}
From the Mar\v{c}enko-Pastur Theorem, $\tilde{m}_{H,\gamma}(z)$ is given by the unique solution of  
\begin{align}
\tilde{m}_{H,\gamma}(z)=\Big( -z+\gamma\int \frac{t}{1+\tilde{m}_{H,\gamma}(z)t}dH(t) \Big)^{-1}, (z,\tilde{m}_{H,\gamma}(z))\in \mathbb{C}^+\times\mathbb{C}^+.
\label{38}
\end{align}
In the special case where $\Sigma = I_p$ for all $n,p$, the  Stieltjes transform of limiting distribution $m_\gamma(z)$ has an explicit form
\begin{equation}
     m_\gamma(z)=\frac{-(1-\gamma-z)+\sqrt{(1-\gamma-z)^2-4\gamma z}}{-2\gamma z}.
\label{6}
\end{equation}

\begin{remark}
The Mar\v{c}enko-Pastur Theorem only provides the expression of Stieltjes transform $m_{H,\gamma}(z)$ on $z \in \mathbb{C}^+$. To complete the definition of Stieltjes transform $m_{H,\gamma}(z)$ on $\mathbb{C} \backslash [0, \infty)$, we need to notice that (\ref{39}) and (\ref{38}) can be extended to $z \in \mathbb{C} \backslash [0, \infty)$ by continuity. The restriction $z \in \mathbb{C}^+$ makes sure that the equation (\ref{38}) has unique solution (see, e.g., \cite{Yao_Zheng_Bai_2015}).
\end{remark}

\subsection{The Main Theorem}
To ensure the existence of the limiting risk in general settings, some technical assumptions are needed.
\begin{assumption}
In each dataset $\mathcal{D}_t$, the sample matrix $X_t=Z_t\Sigma_t^{1/2}$, where $Z_t$ has i.i.d. entries with zero mean, unit variance, and finite $16$-th moment. Besides, the test data $x_0=\Sigma_0^{1/2}z_0$, and $z_0$ has zero mean and unit variance.  
\label{56}
\end{assumption}

\begin{assumption}
The task number $T$ is a constant. In each dataset $\mathcal{D}_t$, $n_t,p \rightarrow \infty$ and $p/n_t \rightarrow \gamma_t \in (0, \infty)$.
\label{57}
\end{assumption}

\begin{assumption}
The signal $\Vert \beta \Vert^2 \rightarrow r^2 \in (0, \infty)$,  as $p \rightarrow \infty$. 
\label{48}
\end{assumption}
Assumption \ref{56} describes the data generation setting,  Assumption \ref{57} defines the asymptotic regime in each dataset, and Assumption \ref{48} fixes the asymptotic scale of regression coefficients. These assumptions are classical in analysis of high-dimensional regression models. The last two assumptions establish the relationship between different tasks.

\begin{assumption}
The covariance matrices $\{\Sigma_t,t=1,\cdots,T\}$ and $\Sigma_0$ are commutable, and there exists two constants $c<C$ such that $0<c\le \lambda_{min}(\Sigma_t) \le \lambda_{max}(\Sigma_t) \le C < \infty$ for $t=1,\cdots,T$ and $t=0$.
\label{51}
\end{assumption}
The commutability of covariance matrices in Assumption \ref{51} ensures that, there exists an orthogonal matrix $U \in \mathbb{R}^{p\times p}$ such that $\Sigma_t=UD_t U^\top$ for $t=1,\cdots,T$ and $\Sigma_0=UD_0 U^\top$, where $\{D_t\}_{t=1,\cdots,T}$ and $D_0$ are diagonal matrices. Let $\Sigma_t=\sum_{i=1}^p d_{t,i}u_iu_i^\top$ be the simultaneous eigenvalue decomposition of $\Sigma_t$ for $t=1,\cdots,T$ and $t=0$, then we define the joint ESD as 
\begin{align}
 H_n(x_1, \cdots, x_T, x_0)=\frac{1}{p}\sum_{i=1}^p 
 \textbf{1}\{d_{1,i}\le x_1,\cdots,d_{T,i}\le x_T,d_{0,i}\le x_0\},
 \label{31}
\end{align}
and a weighted joint ESD as
\begin{align}
 G_n(x_1, \cdots, x_T, x_0)=\frac{1}{\Vert\beta\Vert^2}\sum_{i=1}^p \langle\beta,u_i\rangle^2\textbf{1}\{d_{1,i}\le x_1,\cdots,d_{T,i}\le x_T,d_{0,i}\le x_0\}.
 \label{36}
\end{align}
Here we note that, although $\{\Sigma_t,t=1,\cdots,T\}$ and $\Sigma_0$ can be diagonalized simultaneously by different orthogonal matrix, the values of joint ESD (\ref{31}) and weighted joint ESD (\ref{36}) are unique regardless of the choice of orthogonal matrix $U$. 

The next assumption describes the relationship between covariance matrices and the true coefficients in the limiting form.
\begin{assumption}
The joint distributions $H_n$ and $G_n$ converge weakly to  limit joint distributions $H$ and $G$, respectively.
\label{52}
\end{assumption}
By Assumption \ref{52}, the marginal ESDs $F_{\Sigma_t}$ and $F_{\Sigma_0}$ converges weakly to $H_t$ and $H_0$, respectively, where $H_t$ and $H_0$ are the marginal distributions of $H$ with respect to $x_t$ and $x_0$.

For $t=1,\cdots,T$, define $\tilde{m}_t=\tilde{m}_{H_t,\gamma_t}(-\lambda_t)$, and 
\begin{align*}
\mu_t&=\Big[\Big(1+\int \frac{\gamma_t s}{\lambda_t(1+\tilde{m}_ts)}dH_t(s)\Big)^2-\int \frac{\gamma_t s^2}{\lambda_t^2(1+\tilde{m}_ts)^2}dH_t(s)\Big]^{-1},\\
a_t&=\int \frac{s_ts_0}{\prod_{j=t}^T\lambda_j^2(1+\tilde{m}_js_j)^2}dH(\textbf{s}), \quad 
b_{t}=\int \frac{\lambda_t(1+\tilde{m}_ts_t)\cdot s_0}{\prod_{j=t}^T\lambda_j^2(1+\tilde{m}_js_j)^2}dH(\textbf{s}),\\
c_{t}&=\int \frac{s_0}{\prod_{j=t}^T\lambda_j^2(1+\tilde{m}_js_j)^2}dH(\textbf{s}), \quad
g_{t}=\int \frac{s_0}{\prod_{j=t}^T\lambda_j^2(1+\tilde{m}_js_j)^2}dG(\textbf{s}).
\end{align*}
where $\textbf{s}=(s_1,\cdots,s_T,s_0)$. For $1 \le t \le l \le T$, define
\begin{align*}
a_{t,l}=\int \frac{s_ts_l}{\prod_{j=t}^l\lambda_j^2(1+\tilde{m}_js_j)^2}dH(\textbf{s}),\quad
b_{t,l}=\int \frac{\lambda_t(1+\tilde{m}_ts_t)\cdot s_l}{\prod_{j=t}^l\lambda_j^2(1+\tilde{m}_js_j)^2}dH(\textbf{s}),\\
c_{t,l}=\int \frac{s_l}{\prod_{j=t}^l\lambda_j^2(1+\tilde{m}_js_j)^2}dH(\textbf{s}),\quad
g_{t,l}=\int \frac{s_l}{\prod_{j=t}^l\lambda_j^2(1+\tilde{m}_js_j)^2}dG(\textbf{s}).
\end{align*}
Next, we define $\{\rho_t,1\le t\le T\}$ recursively as $\rho_T=a_T$, and
\begin{align*}
\rho_{t}=a_t+\sum_{j=t+1}^T\gamma_j\mu_ja_{t,j}\rho_{j}, \quad 1\le t\le T-1.
\end{align*}
Define $\{\rho_{s,t}^{(1)},1\le s \le t\le T\}$ as $\rho_{s,s}^{(1)}=0$ and 
\begin{align*}
\rho_{s,t}^{(1)}=b_{s,t}+\sum_{j=s}^{t-1}\gamma_j\mu_ja_{j,t}\rho_{s,j}^{(1)}, \quad 1\le s< t\le T.
\end{align*}
Define $\{\rho_{s,t}^{(2)},1\le s  \le t\le T\}$ as $\rho_{s,s}^{(2)}=c_{s,s}$ and 
\begin{align*}
\rho_{s,t}^{(2)}=c_{s,t}+\sum_{j=s}^{t-1}\gamma_j\mu_ja_{j,t}\rho_{s,j}^{(2)}, \quad 1\le s< t\le T.
\end{align*}

We are now ready to state our main theorem on the asymptotic prediction risk of continual ridge estimators in general settings.
\begin{theorem}
Under Assumption \ref{56}-\ref{52}, it holds almost surely that
\begin{align*}
   B_X(\hat{\beta}_T; \beta,\Sigma_0) &\rightarrow  \tilde{B}_T(r, \bm{\gamma}, \bm{\lambda},G,H):=r^2 \Big(\prod_{j=1}^T\lambda_j^2\Big)\Big(g_1+ 
   \sum_{t=1}^T\gamma_t\mu_t\rho_tg_{1,t}\Big),\\
   V_X(\hat{\beta}_T; \beta,\Sigma_0) &\rightarrow \tilde{V}_T(\bm{\gamma}, \bm{\lambda},H):=\sigma^2
\sum_{t=1}^{T}\gamma_{t}\big(\prod_{s=t+1}^T\lambda_s^2\big)(L_{1,t}-\lambda_tL_{2,t}),
\end{align*}
and 
\begin{align*}
   R_X(\hat{\beta}_T; \beta,\Sigma_0) \rightarrow \tilde{R}_T(r, \bm{\gamma}, \bm{\lambda},G,H ):=\tilde{B}_T(r, \bm{\gamma}, \bm{\lambda},G,H)+\tilde{V}_T(\bm{\gamma}, \bm{\lambda},H), 
\end{align*}
where
\begin{align*}
L_{1,t}=b_t+\sum_{j=t}^T\gamma_j\mu_ja_{j}\rho_{t,j}^{(1)},
\quad L_{2,t}=c_t+\sum_{j=t}^T\gamma_j\mu_ja_{j}\rho_{t,j}^{(2)}.
\end{align*}
\label{10}
\end{theorem}
The proof of Theorem \ref{10} is provided in Appendix \ref{43}. In Theorem \ref{10}, the regularization parameters $\bm{\lambda}=(\lambda_1\cdots,\lambda_T)^\top$ are treated as constants independent on $n$ and $p$. Note that the asymptotic bias term is related to the coefficients $\beta$ and the asymptotic variance term is related to the variance of noise $\sigma^2$. As $n \rightarrow \infty$, the variance of noise is assumed to remain constant, while the dimension of regression coefficients is growing. The variation of $\beta$ in this limiting process is characterized by the signal strength $r^2$ and the joint distribution $G$ through $g_1$ and $g_{1,t}$. In addition, $\bm{\gamma}=(\gamma_1,\cdots,\gamma_T)^\top$ is intuitively interpreted as model complexity parameters, and the joint distribution $H$ characterizes the relationship of covariance matrices in different tasks. These factors influence both the asymptotic bias term and the asymptotic variance term.

\subsection{Examples}
\label{54}
To get a better understanding of our main theorem, we focus on some concrete examples by considering different structures of covariance matrices in each task. Note that Assumption \ref{51} guarantees that $\Sigma_1, \cdots, \Sigma_T$ and $\Sigma_0$ are simultaneously diagonalizable.
Suppose $\Sigma_t=UD_t U^\top$ for $t=1,\cdots,T$ and $\Sigma_0=UD_0 U^\top$, where $\{D_t\}_{t=1,\cdots,T}$ and $D_0$ are diagonal matrices, and $U$ is an orthogonal matrix. Let $\mathring{Z}_t=Z_tU$, by rotational invariance of $Z_t$, the entries of $\mathring{Z}_t$ and $Z_t$ have the same joint distribution, and $\mathring{Z}_1, \cdots, \mathring{Z}_T$ are independent. If we define $\mathring{X}_t=\mathring{Z}_tD_t^{1/2}, \mathring{\Sigma}_t=\frac{1}{n_t}\mathring{X}_t^\top \mathring{X}_t$, and $ \mathring{A}_t=\lambda_t(\mathring{\Sigma}_t+\lambda_tI_p)^{-1}$, we have
\begin{align*}
\hat{\Sigma}_t=\frac{1}{n_t}UD_t^{1/2}(Z_tU)^\top (Z_tU)D_t^{1/2}U^\top{=}U\mathring{\Sigma}_tU^\top,
\end{align*}
and $A_t{=}U^\top \mathring{A}_tU$. Thus by Lemma \ref{46},
\begin{align*}
B_X(\hat{\beta}; \beta,\Sigma_0)&=(U\beta)^\top \mathring{A}_1\mathring{A}_2\cdots \mathring{A}_TD_0 \mathring{A}_T\cdots \mathring{A}_2\mathring{A}_1(U\beta),\\
V_X(\hat{\beta}; \beta,\Sigma_0)&=\sigma^2\sum_{t=1}^{T} \frac{1}{\lambda_tn_t}{\rm Tr}\big[\mathring{A}_T\mathring{A}_{T-1}\cdots \mathring{A}_{t+1}(\mathring{A}_t-\mathring{A}_t^2)\mathring{A}_{t+1}\cdots \mathring{A}_{T-1}\mathring{A}_TD_0\big].
\end{align*}
From the above derivation, we observe that the prediction risk can be expressed in terms of
diagonal covariance matrices and the transformed coefficients $\mathring{\beta}=U\beta$. Therefore, in the following cases, we assume that all the covariance matrices are diagonal matrices.

\subsubsection{Identity Covariance Matrices}
We first consider a trivial case where  $\Sigma_1=\cdots=\Sigma_T=\Sigma_0 = I_p$. In scenarios where task distributions are identical, the continual learning framework effectively reduces to an online learning framework, as no task-specific adaptation or catastrophic forgetting mitigation is required. Nevertheless, it remains important to investigate the properties of continual ridge estimation in this setting, since this algorithm is closely related to online learning methods. Taking the derivative in the updating rule (\ref{59}), we have
\begin{align*}
\hat{\beta}_t=\hat{\beta}_{t-1}+\lambda^{-1}n_t^{-1}X_t^\top(X_t\hat\beta_t-y_t).
\end{align*}
To enhance the updating efficiency, one may substitute $\hat{\beta}_{t-1}$ for $\hat{\beta}_{t}$ in the right-hand side of the equation, yielding
\begin{align*}
\hat{\beta}_t=\hat{\beta}_{t-1}-\eta n_t^{-1}X_t^\top(y_t-X_t\hat\beta_{t-1}),
\end{align*}
which is the updating rule for online (batch) gradient descent method with the learning rate $\eta=\lambda^{-1}$ (see, e.g., \cite{HOI2021249}). Assuming computational efficiency is not a constraint, the updating rule of continual ridge regression seems to be a reasonable choice for online learning.

The following result shows the asymptotic risk in this setting, which has been largely simplified compared to the main theorem.
\begin{theorem}
If $\Sigma_1=\cdots=\Sigma_T=\Sigma_0=I_p$, under Assumption \ref{56}-\ref{48}, it holds almost surely that
\begin{align*}
   B_X(\hat{\beta}_T; \beta, \Sigma_0) &\rightarrow  \tilde{B}_T(r, \bm{\gamma}, \bm{\lambda}):=r^2 \prod_{t=1}^T\big[\lambda_t^2m_{\gamma_t}'(-\lambda_t)\big],\\
   V_X(\hat{\beta}_T; \beta,\Sigma_0) &\rightarrow \tilde{V}_T(\bm{\gamma}, \bm{\lambda}):=
\sigma^2   \sum_{t=1}^{T}\gamma_{t}\upsilon_{t}\prod_{s=t+1}^T\big[\lambda_s^2m_{\gamma_s}'(-\lambda_s)\big],
\end{align*}
and 
\begin{align*}
   R_X(\hat{\beta}_T; \beta,\Sigma_0) \rightarrow \tilde{R}_T( r,\bm{\gamma}, \bm{\lambda}):=\tilde{B}_T(r, \bm{\gamma}, \bm{\lambda})+\tilde{V}_T(\bm{\gamma}, \bm{\lambda}), 
\end{align*}
where
\begin{align*}
 \upsilon_t=m_{\gamma_t}(-\lambda_t)-\gamma_tm_{\gamma_t}'(-\lambda_t),
\end{align*}
and $m_{\gamma_t}$ has an explicit form as (\ref{6}).
\label{5}
\end{theorem}
The proof of Theorem \ref{5} is provided in Appendix \ref{24}. The main difference between Theorem \ref{5} and Theorem \ref{10} is that, the asymptotic prediction risk of continual ridge estimation depends on $\beta$ only through the limiting signal strength $r^2$ in the bias term. 

\begin{remark}
The asymptotic result of continual ridge regression coincides with that of classical high-dimensional ridge regression. When $T=1$, the continual learning framework reduces to the ridge regularization. To describe the impact of the true parameter on the asymptotic risk, we define the limiting signal-to-noise ratio by $\text{SNR}=r^2/\sigma^2$. According to our analysis, the asymptotic prediction risk is given by
\begin{align*}
  \tilde{R}(r, \gamma, \lambda)=\sigma^2\Big\{\lambda^2m_{\gamma}'(-\lambda) \text{SNR}+ \gamma[m_{\gamma}(-\lambda)-\gamma m_{\gamma}'(-\lambda)]\Big\},
\end{align*}
which is consistent with existing results. (see, e.g., \cite{10.1214/17-AOS1549}, \cite{10.1214/21-AOS2133}).
\end{remark}

\subsubsection{Isotropic Covariance Matrices with Different Scales}
If the covariance matrices are different in each task, it may be hard to find a simplified consequence of limiting risk. Moreover, even the numeral calculation of the results in Theorem \ref{10} is not a trivial task, unless the covariance matrices have simple structures. A natural idea is to assume that all the covariance matrices are isotropic, and the difference of tasks embody in the variation of scales. Specifically, we assume that $\Sigma_t=\delta_tI_p, t=1,\cdots,T$ and $\Sigma_0=\delta_0I_p$, where $\{\delta_t\}_{t=1,\cdots,T}$ and $\delta_0$ are different positive constants. In this setting, we have an explicit form of Stieltjes transform
\begin{align*}
     m_t=\frac{-[\delta_t(1-\gamma_t)+\lambda_t]+\sqrt{[\delta_t(1-\gamma_t)+\lambda_t]^2+4\delta_t\gamma_t \lambda_t}}{2\gamma_t \delta_t \lambda_t},
\end{align*}
and $ \tilde{m}_t=m_t\gamma_t+\frac{1-\gamma_t}{\lambda_t}$.
Furthermore, the elements of the expressions in Theorem \ref{10} can be simplified as 
\begin{align*}
 \mu_t&=\Big\{\big[1+\frac{\gamma_t\delta_t}{\lambda_t(1+\tilde{m}_t\delta_t)}\big]^2-\frac{\gamma_t\delta_t^2}{\lambda_t^2(1+\tilde{m}_t\delta_t)^2}\Big\}^{-1},\\
a_t&=\frac{\delta_t\delta_0}{\prod_{j=t}^T\lambda_j^2(1+\tilde{m}_j\delta_j)^2},\quad a_{tl}=\frac{\delta_t\delta_l}{\prod_{j=t}^l\lambda_j^2(1+\tilde{m}_j\delta_j)^2},\\
b_t&=\frac{\lambda_t(1+\tilde{m}_t\delta_t)\delta_0}{\prod_{j=t}^T\lambda_j^2(1+\tilde{m}_j\delta_j)^2},\quad b_{tl}=\frac{\lambda_t(1+\tilde{m}_t\delta_t)\delta_l}{\prod_{j=t}^l\lambda_j^2(1+\tilde{m}_j\delta_j)^2},\\
c_t&=\frac{\delta_0}{\prod_{j=t}^T\lambda_j^2(1+\tilde{m}_j\delta_j)^2},\quad c_{tl}=\frac{\delta_l}{\prod_{j=t}^l\lambda_j^2(1+\tilde{m}_j\delta_j)^2}.
\end{align*}
Since all covariance matrices are isotropic, any $\beta$ with $\Vert \beta \Vert^2=r^2$ satisfies Assumption \ref{48} and Assumption \ref{52}, and it holds that $g_1=c_1$ and $g_{1,t}=c_{1,t}$.

\subsubsection{Covariance Matrices with Different Block Sizes}
We finally focus on a fundamentally different setting where covariance matrices are not isotropic. For simplicity, we assume that each covariance matrix has two different eigenvalues, the first of which is $\delta$ and the second is 1. However, in different covariance matrices, the proportion of numbers of these two eigenvalues may be different. For $\Sigma_t \in \mathbb{R}^{p\times p}$, we suppose the first $\lfloor\pi_tp\rfloor$ eigenvalues are all $\delta$ $(\delta \not= 1)$ and the others are 1, where $\pi_t \in (0,1)$ is a constant. Namely, $\Sigma_t = \text{diag}\{\delta I_{\lfloor\pi_tp\rfloor},I_{p-\lfloor\pi_tp\rfloor}\}, t=1,\cdots,T$ and $t=0$. 

In this example, the joint ESD $H_n$ converges weakly to a limiting distribution $H$. Let $\{q_1,\cdots,q_{T+1}\} = \{1,\cdots, T+1\}$ such that $\pi_{q_t}=\pi_{(t)}$, where $\pi_{(1)}\le \pi_{(2)} \le \cdots \le \pi_{(T+1)}$ is the ordered sequence of $\pi_1, \cdots, \pi_T, \pi_0$. Then the limiting distribution $H$ satisfies
\begin{align}
  \int f(\textbf{s})dH(\textbf{s})=\sum_{t=0}^{T+1} (\pi_{(t+1)}-\pi_{(t)})f({s_t}),
  \label{49}
\end{align}
where $\pi_{(0)}=0, \pi_{(T+2)}=1$, and
\begin{align*}
  {s_t}=\delta(1,\cdots,1)^\top+(1-\delta)\sum_{k=1}^t{e}_{q_k} \in \mathbb{R}^{T+1}, t=0,1,\cdots,T+1.
\end{align*}
Here ${e}_{1}=(1,0,\cdots,0)^\top, \cdots, {e}_{T+1}=(0,0,\cdots,1)^\top$.

By (\ref{49}), we have
\begin{align*}
 \mu_t&=\Big\{\big[1+\frac{\pi_t\gamma_t\delta}{\lambda_t(1+\tilde{m}_t\delta)}+\frac{(1-\pi_t)\gamma_t}{\lambda_t(1+\tilde{m}_t)}\big]^2-\big[\frac{\pi_t\gamma_t\delta^2}{\lambda_t^2(1+\tilde{m}_t\delta)^2}+\frac{(1-\pi_t)\gamma_t}{\lambda_t^2(1+\tilde{m}_t)^2}\big]\Big\}^{-1},\\
a_t&=\sum_{k=0}^{T+1} (\pi_{(k+1)}-\pi_{(k)})\frac{s_{k,t}s_{k,T+1}}{\prod_{j=t}^T\lambda_j^2(1+\tilde{m}_js_{k,j})^2},
a_{t,l}=\sum_{k=0}^{T+1} (\pi_{(k+1)}-\pi_{(k)})\frac{s_{k,t}s_{k,l}}{\prod_{j=t}^l\lambda_j^2(1+\tilde{m}_js_{k,j})^2},\\
b_t&=\sum_{k=0}^{T+1} (\pi_{(k+1)}-\pi_{(k)})\frac{\lambda_t(1+\tilde{m}_ts_{k,t})s_{k,T+1}}{\prod_{j=t}^T\lambda_j^2(1+\tilde{m}_js_{k,j})^2},
b_{t,l}=\sum_{k=0}^{T+1} (\pi_{(k+1)}-\pi_{(k)})\frac{\lambda_t(1+\tilde{m}_ts_{k,t})s_{k,l}}{\prod_{j=t}^l\lambda_j^2(1+\tilde{m}_js_{k,j})^2},\\
c_t&=\sum_{k=0}^{T+1} (\pi_{(k+1)}-\pi_{(k)})\frac{s_{k,T+1}}{\prod_{j=t}^T\lambda_j^2(1+\tilde{m}_js_{k,j})^2},
c_{t,l}=\sum_{k=0}^{T+1} (\pi_{(k+1)}-\pi_{(k)})\frac{s_{k,l}}{\prod_{j=t}^l\lambda_j^2(1+\tilde{m}_js_{k,j})^2},
\end{align*}
where $s_{k,t}$ is the $t$-th elements of $\vec{s_k}$. By (\ref{38}), $\tilde{m}_t$ is determined by
\begin{align*}
\tilde{m}_t=\Big(\lambda_t+\frac{\pi_t\gamma_t\delta}{1+\tilde{m}_t\delta}+\frac{(1-\pi_t)\gamma_t}{1+\tilde{m}_t}\Big)^{-1}.
\end{align*}
Unlike the isotropic case, it is tedious to derive the explicit expression of $\tilde{m}_t$. We may calculate $\tilde{m}_t$ using iterative methods.

The main difference between isotropic settings and the present case is that, the limiting risk may not exist for all $\beta$ satisfying Assumption \ref{48}. The source of this problem is that $G_n$ may not converge weakly to $G$ for complex structures of covariance matrices. However, if the components of $\beta$ are equally distributed along the eigenvectors of $\{\Sigma_t\}_{t=1,\cdots, T}$ (which are the same for different covariance matrices), it always holds that $G=H$, and the limiting risk exists regardless of the structure of covariance matrices. If $\beta$ is fixed, a natural example is $\beta=\frac{r}{\sqrt{p}}(1,\cdots,1)^\top$, since we assume that the covariance matrices are all diagonal. In more cases, $\beta$ can be regarded as random regression coefficients \citep{10.1214/16-AOS1456, JMLR:v21:19-277}. If the prior distribution of $\beta$ satisfies $\mathbb{E}\beta=0$ and $\text{Cov}(\beta)=\frac{r^2}{p}I_p$, then $\beta$ is equally distributed and the limiting posterior risk exists. Anyhow, if $\beta$ are equally distributed, we have $g_1=c_1$ and $g_{1,t}=c_{1,t}$, and the limiting risk can be calculated by Theorem \ref{10}.

\section{Experiments}
\label{62}
In this section, we analyze the asymptotic risk curves for continual ridge estimators and verify our results by experiments. 
For convenience, we assume that each task has identical sample size $n$, so that $\gamma_1=\cdots=\gamma_T=p/n$.
 Thereby, the asymptotics of our evaluation metrics for generalization performance are simplified to
\begin{align*}
&\bar{R}_X(\hat{\beta}; \beta)\rightarrow \frac{1}{T}\sum_{t=1}^T \tilde{R}_{T,t},\quad a.s.,\\
&{BWT}_X(\hat{\beta}_T; \beta)\rightarrow \frac{1}{T-1}\sum_{t=1}^{T-1} (\tilde{R}_{T,t}-\tilde{R}_{t,t}),\quad a.s.,\\
&{FWT}_X(\hat{\beta}_T; \beta)\rightarrow \frac{1}{T-1}\sum_{t=2}^{T} (\tilde{R}_{t,t}-\tilde{R}^{\rm ridge}_{t}),\quad a.s.,
\end{align*}
where $\tilde{R}_{T,t}$ and $\tilde{R}^{\rm ridge}_{t}$ denote  the almost sure limits of $R_X(\hat{\beta}_T; \beta,\Sigma_t)$ and $R_X(\hat{\beta}^{\rm ridge}_{t}; \beta,\Sigma_t)$ from Theorem \ref{10}, respectively.

In each experiment, we plot the theoretical risk curves using Theorem \ref{10}. Meanwhile, we calculate the evaluation metrics above by simulated data and compare the results with the theoretical risk curves. We consider the data generated settings mentioned in section
\ref{54}. In each setting, we are interested in the relationship between the evaluation metrics and the task number $T$ under different model complexity (quantified by $\gamma=p/n$). Besides, we also consider the influence of the choice of regularization parameters $\bm{\lambda}$ on generalization performance.

\subsection{Tuning Parameters Selection}
Before turning to explicit cases, we focus on the choice of regularization parameters $\bm{\lambda}=(\lambda_1, \cdots, \lambda_T)^\top$. From (\ref{59}), we know that the choice of $\lambda_t$ is related to the balance between the knowledge from $\mathcal{D}_t$ and previous tasks. If $\lambda_t$ is too large, the estimation is not adapted to the current task. Alternatively, if $\lambda_t$ is too small, the catastrophic forgetting may happen on previous tasks. Intuitively, there exists an optimal choice, and the oracle optimal regularization parameter is the minima of the risk function. However, this is obviously not a reasonable choice in continual learning procedures, because the risk function contains information from all steps, which is unavailable when we choose regularization parameters at each step before future tasks arrive. In other words, a reasonable strategy of parameter selection may determine $\lambda_t$ at step $t$. If we write $\tilde{R}_{t}=t^{-1}\sum_{s=1}^t \tilde{R}_{t,s}$ as the average risk function at time $t$, we may choose $\bm{\lambda}=(\lambda_1, \cdots, \lambda_T)^\top$ entrywise by
\begin{align}
\lambda_t= \mathop{\arg\min} \limits_{\lambda>0} \tilde{R}_{t}(\lambda_1, \cdots, \lambda_{t-1}, \lambda),
\quad t=1,\cdots, T.
\label{55}
\end{align}
It is a greedy selection strategy, since we choose the (oracle) presently optimal parameter at each step. Although there is no evidence to show that the present choice is optimal for the final average risk function, this strategy is somehow acceptable if there is no prior information about future tasks at present time.

\begin{remark}
In practice, the regularization parameters may be chosen from a data-driven criterion, but standard techniques, such as cross-validation, may be invalid in continual learning frameworks. Similar as the determination of oracle parameters, it is important to notice that a reasonable criterion chooses $\lambda_t$ using the information at time $t$. Specifically, only the dataset $\mathcal{D}_t$ and the estimator $\hat{\beta}_{t-1}=\hat{\beta}_{t-1}(\lambda_1, \cdots, \lambda_{t-1})$ are available, where $\lambda_1, \cdots, \lambda_{t-1}$ are determined at previous steps. Designing such a criterion is beyond the scope of this article.
\end{remark}

\subsection{Experimental Results}
In each of the following experiments, we set the maximum task number $T=20$, sample size $n=100$ and the parameter dimension $p=\lfloor n\gamma \rfloor$, where $\gamma$ is the metric of model complexity. We consider three levels of model complexity: $\gamma=0.6,1.2,2.4$. In each experiment, we assume that the noise is Gaussian with variance $\sigma^2=1$. The signal strength $r^2$ is set to be 1, and the true parameter $\beta$ is generated as $\beta=\frac{r}{\sqrt{p}}(\beta_1,\cdots,\beta_p)^\top$, where $\beta_k$ is randomly sampled by $\textbf{P}(\beta_k=1)=\textbf{P}(\beta_k=-1)=1/2$. For the choice of regularization parameters, the standard setting $\bm{\lambda}_{st}$ follows from (\ref{55}), while a contrast setting is designed as $\bm{\lambda}_{st}/20$, which means there exists catastrophic forgetting in the learning procedure. Each experiment is repeated $B=100$ times and we present the average results as estimation curves. The asymptotic results are presented as theoretical curves in each plot.

\begin{figure*}[!htbp]
    \centering
    \begin{subfigure}{0.31\textwidth}
    	\includegraphics[width=\textwidth]{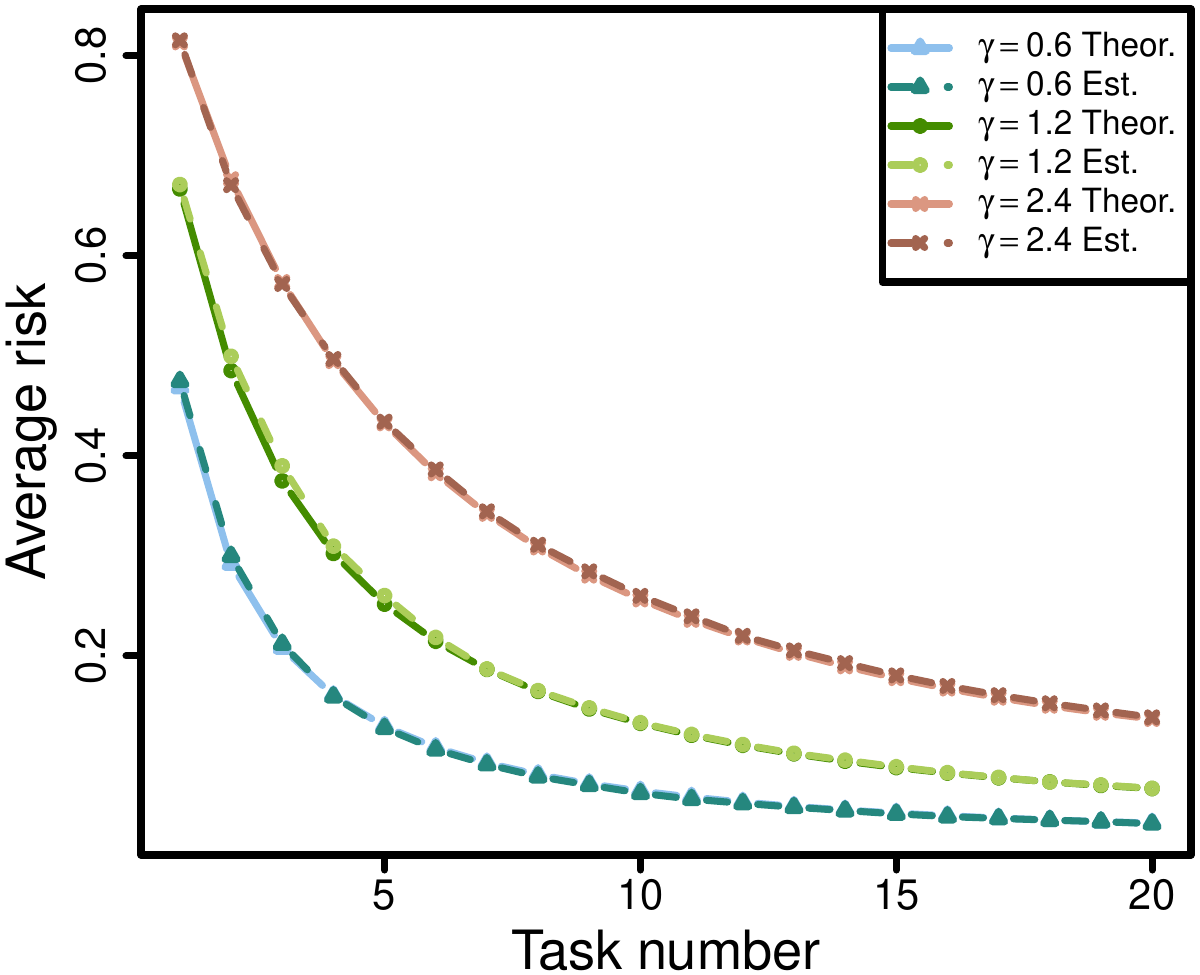}
    	\vspace{-0.13\textwidth}
        \caption{Average risk}
    \end{subfigure}
    \hspace{0.01\textwidth}
    \begin{subfigure}{0.31\textwidth}
        \includegraphics[width=\textwidth]{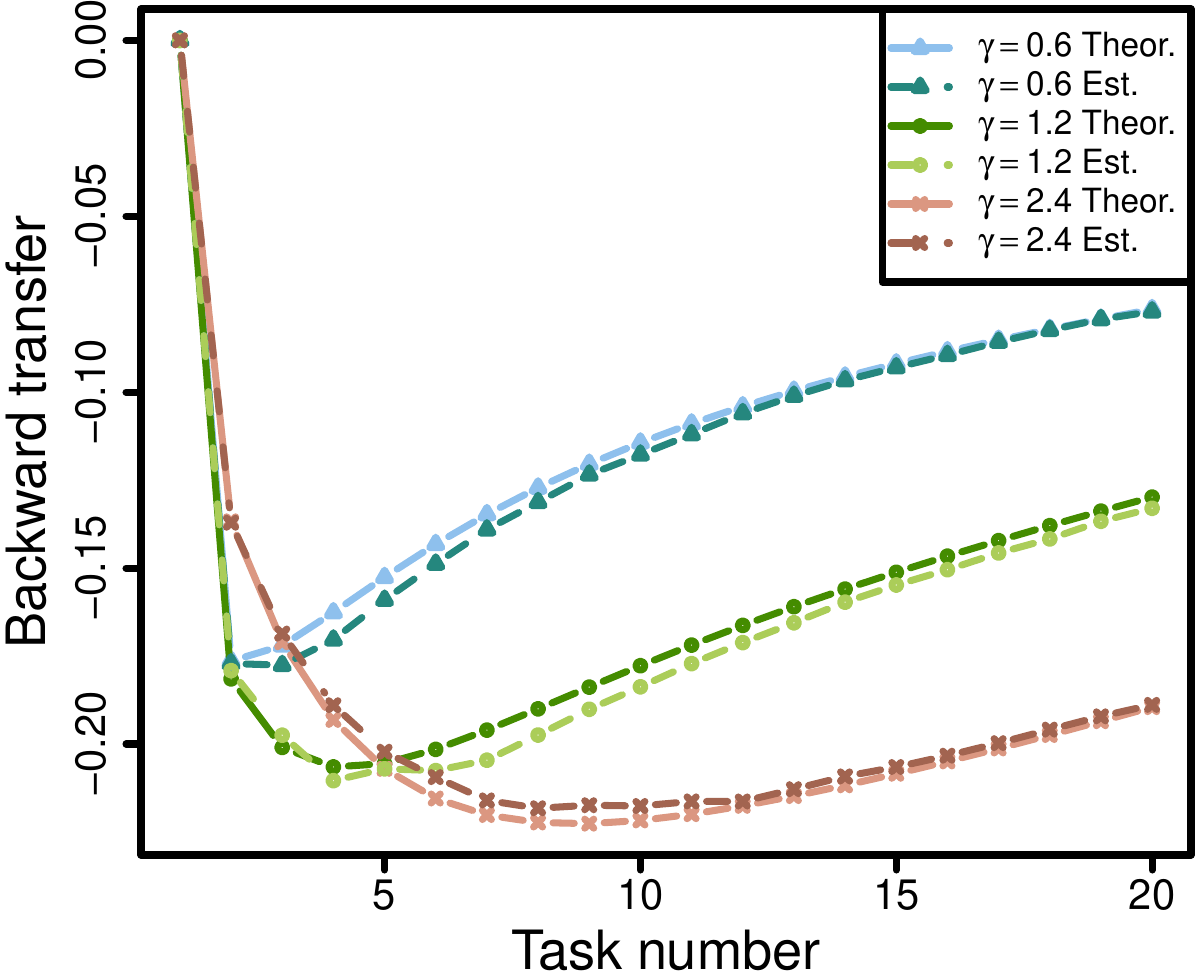}
    	\vspace{-0.13\textwidth}
        \caption{Backward transfer}
    \end{subfigure}
    \hspace{0.01\textwidth}
    \begin{subfigure}{0.31\textwidth}
        \includegraphics[width=\textwidth]{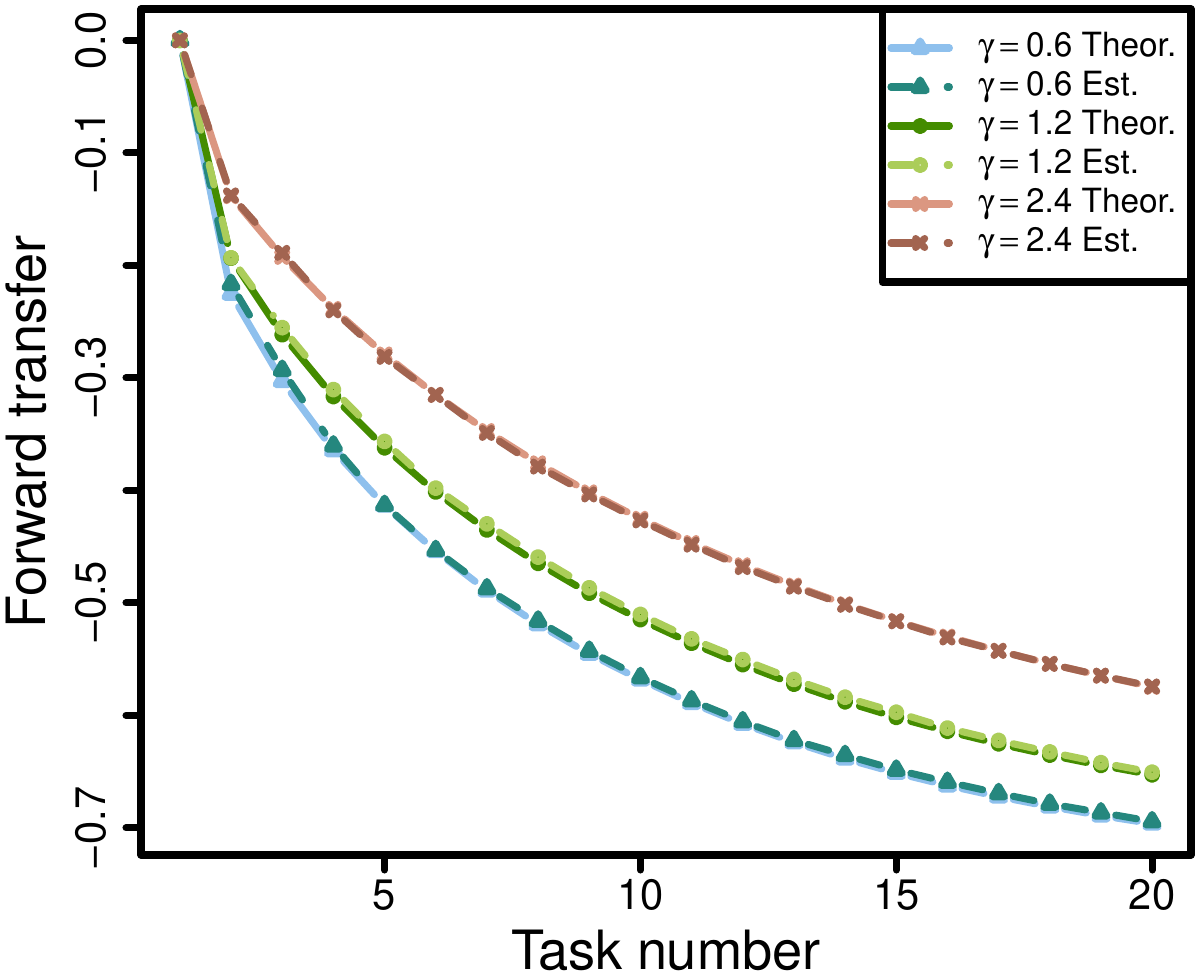}
    	\vspace{-0.13\textwidth}
        \caption{Forward transfer}
    \end{subfigure} \\
    \vspace{0.02\textwidth}
    \begin{subfigure}{0.31\textwidth}
    	\includegraphics[width=\textwidth]{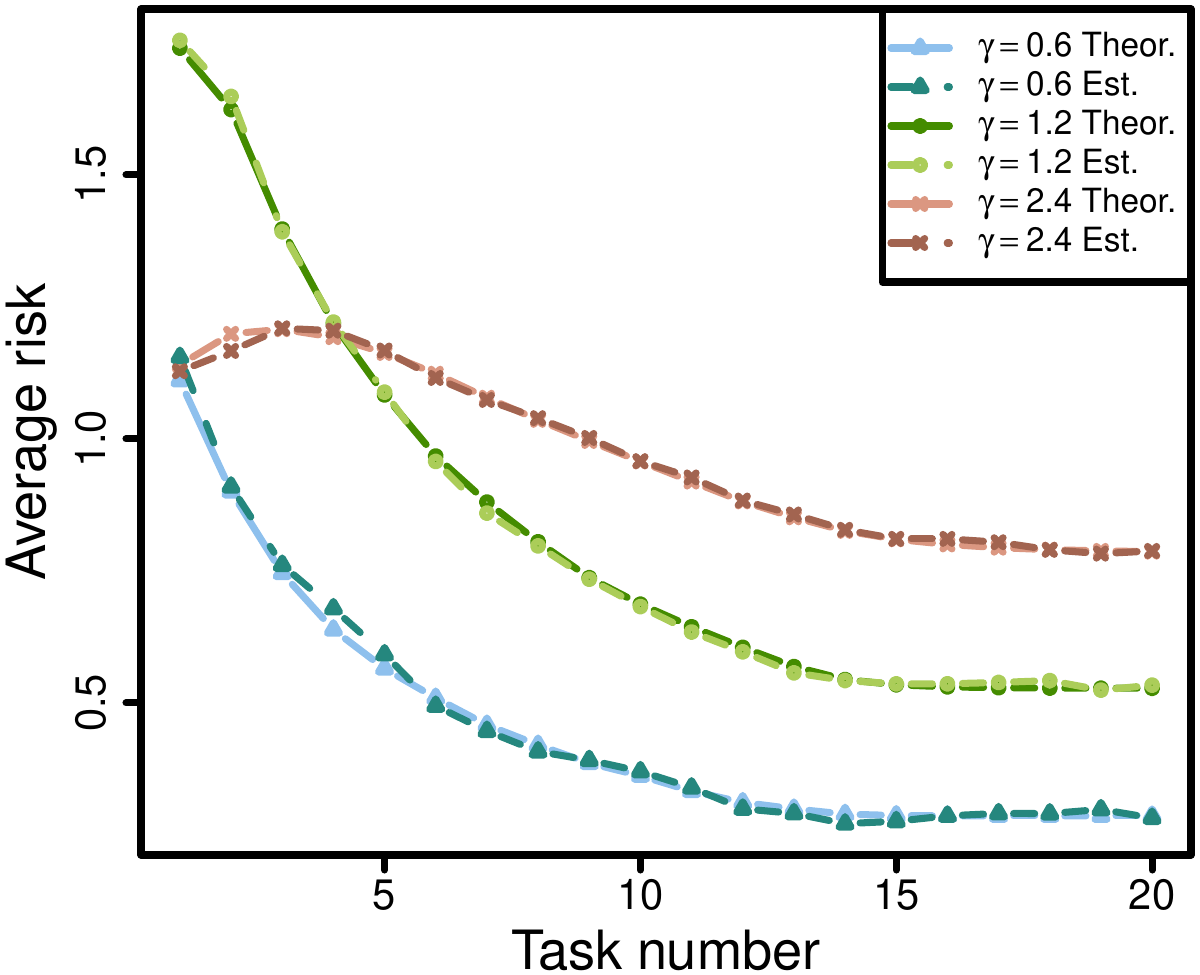}
    	\vspace{-0.13\textwidth}
        \caption{Average risk}
    \end{subfigure}
    \hspace{0.01\textwidth}
    \begin{subfigure}{0.31\textwidth}
        \includegraphics[width=\textwidth]{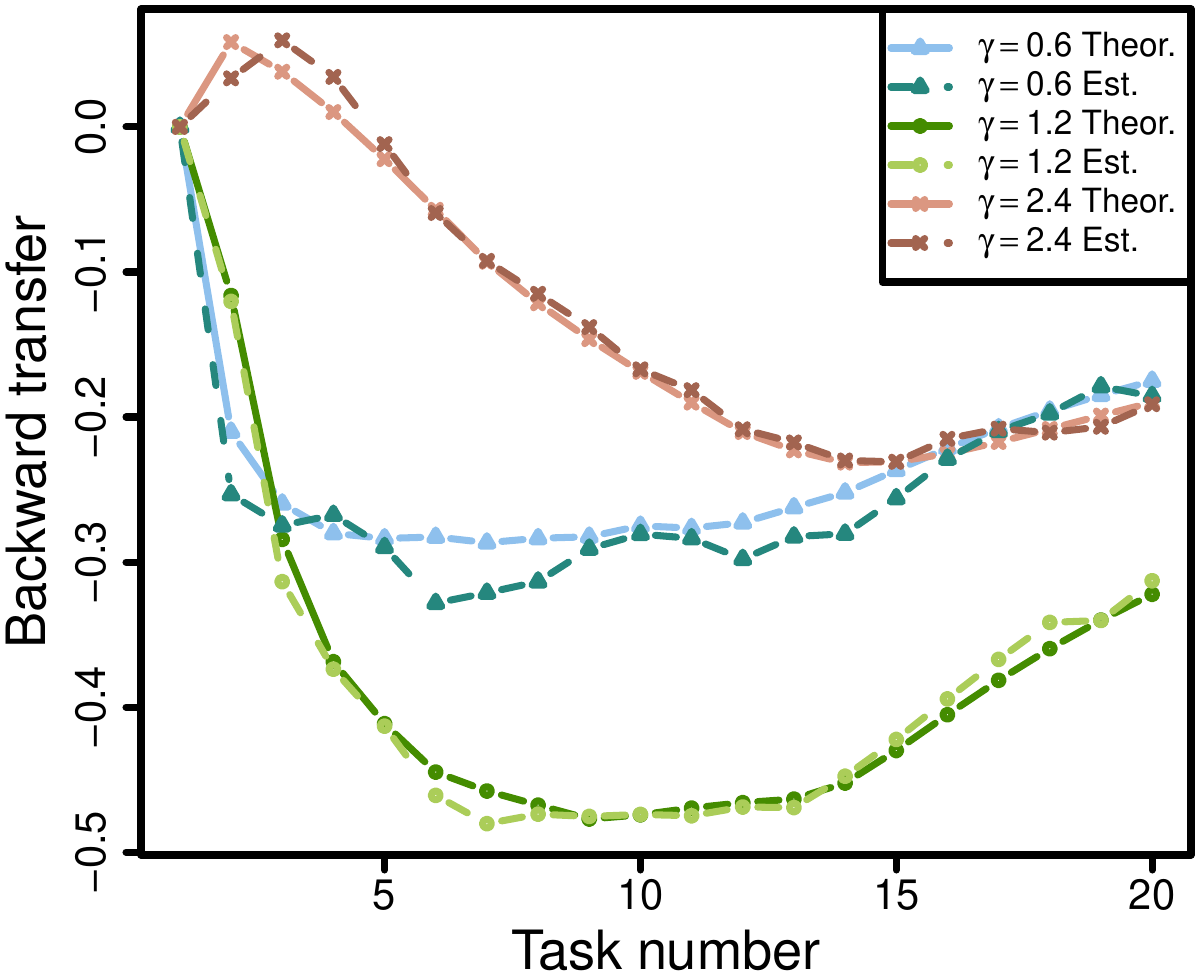}
    	\vspace{-0.13\textwidth}
        \caption{Backward transfer}
    \end{subfigure}
    \hspace{0.01\textwidth}
    \begin{subfigure}{0.31\textwidth}
        \includegraphics[width=\textwidth]{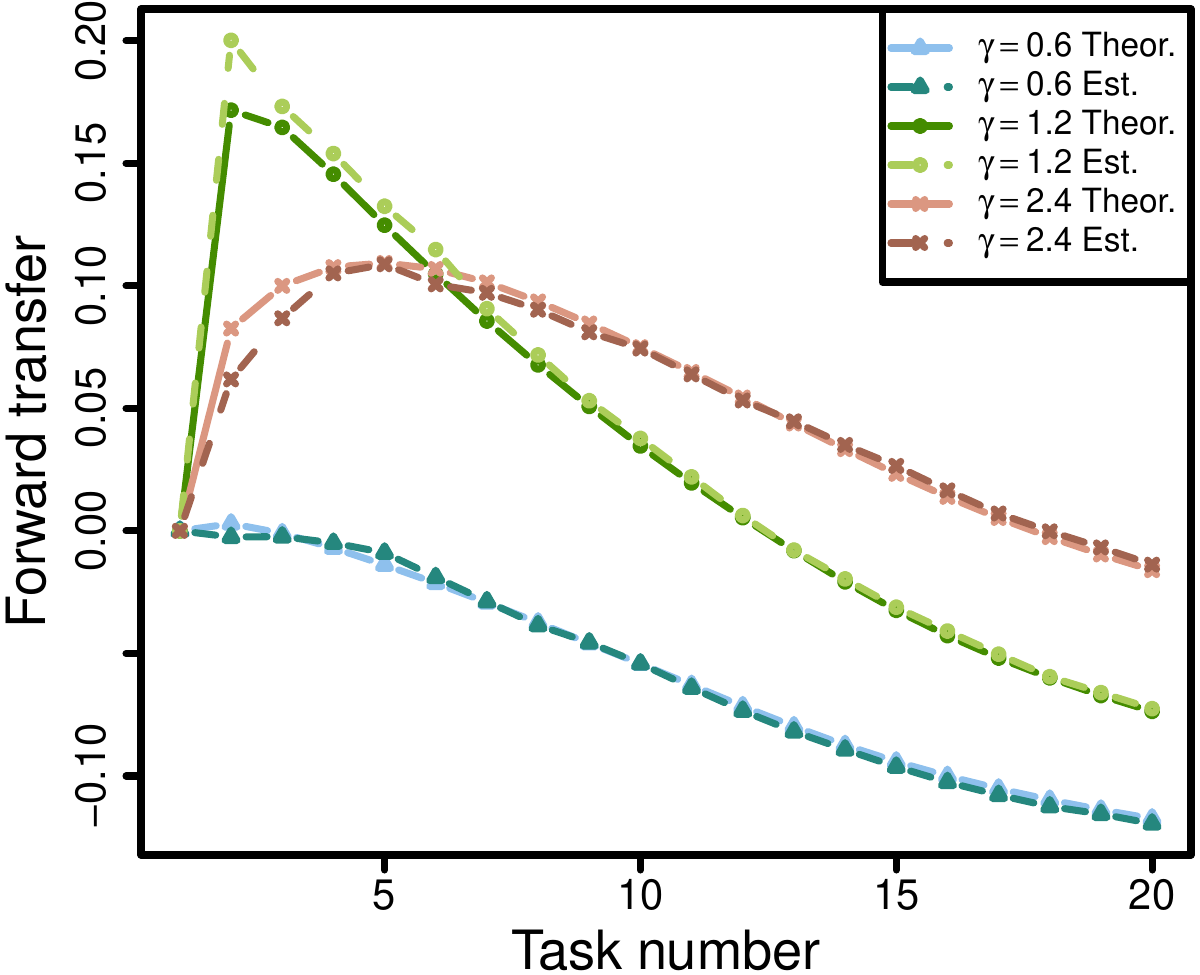}
    	\vspace{-0.13\textwidth}
        \caption{Forward transfer}
    \end{subfigure} 
    \caption{Risk curves with respect to the task number: The covariance matrices are all identity matrices. In the first row, the regularization parameters are $\bm{\lambda}=\bm{\lambda}_{st}$, while in the second row,  $\bm{\lambda}=\bm{\lambda}_{st}/20$. }
    \label{100}
\end{figure*}	
\paragraph{Identity covariance matrices} 
In this example, we consider a setting where the covariance matrices are identical across all tasks. Under this assumption, increasing the number of tasks would be equivalent to increasing the sample size if catastrophic forgetting is disregarded. Therefore, one would expect the average risk to decrease monotonically as more tasks are observed. The simulation results, presented in Figure \ref{100}, confirm this behavior when the regularization parameters are appropriately tuned. Moreover, the results demonstrate that continual ridge estimation exhibits both forward transfer and backward transfer capabilities, and the dominant transfer shifts from backward transfer to forward transfer as the task number increases. However, when regularization parameters are set too small, continual ridge estimation has significant performance reduction, manifested on larger average risk and delayed emergence of forward transfer.


\paragraph{Isotropic covariance matrices with different scales} 
In this example, we consider $\Sigma_t=\delta_tI_p, t=1,\cdots T$, where $\{\delta_t\}_{t=1,\cdots T}$ are different positive constants. The parameters $\bm{\delta}=(\delta_1,\cdots,\delta_T)^\top$ describe the patterns of covariance shift. We investigate two distinct mechanisms: (1) random shift: $\delta_1,\cdots,\delta_T \sim i.i.d. U(0.5,3.5)$; (2) increasing shift: $\delta_t=4t/(T+1)$. The simulation results are presented in Figure \ref{101}. When the regularization parameters are appropriately tuned, we observe that the average risk decreases monotonically under the random shift mechanism, while rising to a peak and then decreasing under the increasing shift mechanism. This behavior suggests that the continual ridge estimation identifies the trend of covariance shift in the early steps, then leveraging the acquired knowledge to enhance overall performance. Besides, the performance of transfer capacities are similar to the first example. When regularization parameters are set too small, we explore that the performance of continual estimator is significantly disturbed by randomness of covariance shift, and the forward transfer behaves poorly since the trends of covariance shift are not identified effectively.

\paragraph{Covariance matrices with different block sizes} 
In the last example, we consider $\Sigma_t = \text{diag}\{\delta I_{\lfloor \pi_tp\rfloor},I_{p-\lfloor\pi_tp\rfloor}\}, t=1,\cdots T$, where $\{\pi_t\}_{t=1,\cdots T} \in [0,1]$ and $\delta>1$ are positive constants. The parameters $\bm{\pi}=(\pi_1,\cdots,\pi_T)^\top$ describe the variation of block size, and $\delta$ describes the scale of the main block. In this experiment we set $\delta=5$. We also investigate two distinct mechanisms: (1) random block size: $\pi_1,\cdots,\pi_T \sim i.i.d. U(0,1)$; (2) increasing block size: $\pi_t=t/T$. The simulation results are presented in Figure \ref{102}. When the regularization parameters are appropriately tuned, the average risk decreases nearly monotonically for both covariance block mechanisms, and the performance of transfer capacities are similar to the first two examples. In the under-regularized setting, there is also 
performance reduction on average risk and forward transfer due to  catastrophic forgetting. \\

\begin{figure*}[!htbp]
    \centering
    \begin{subfigure}{0.31\textwidth}
    	\includegraphics[width=\textwidth]{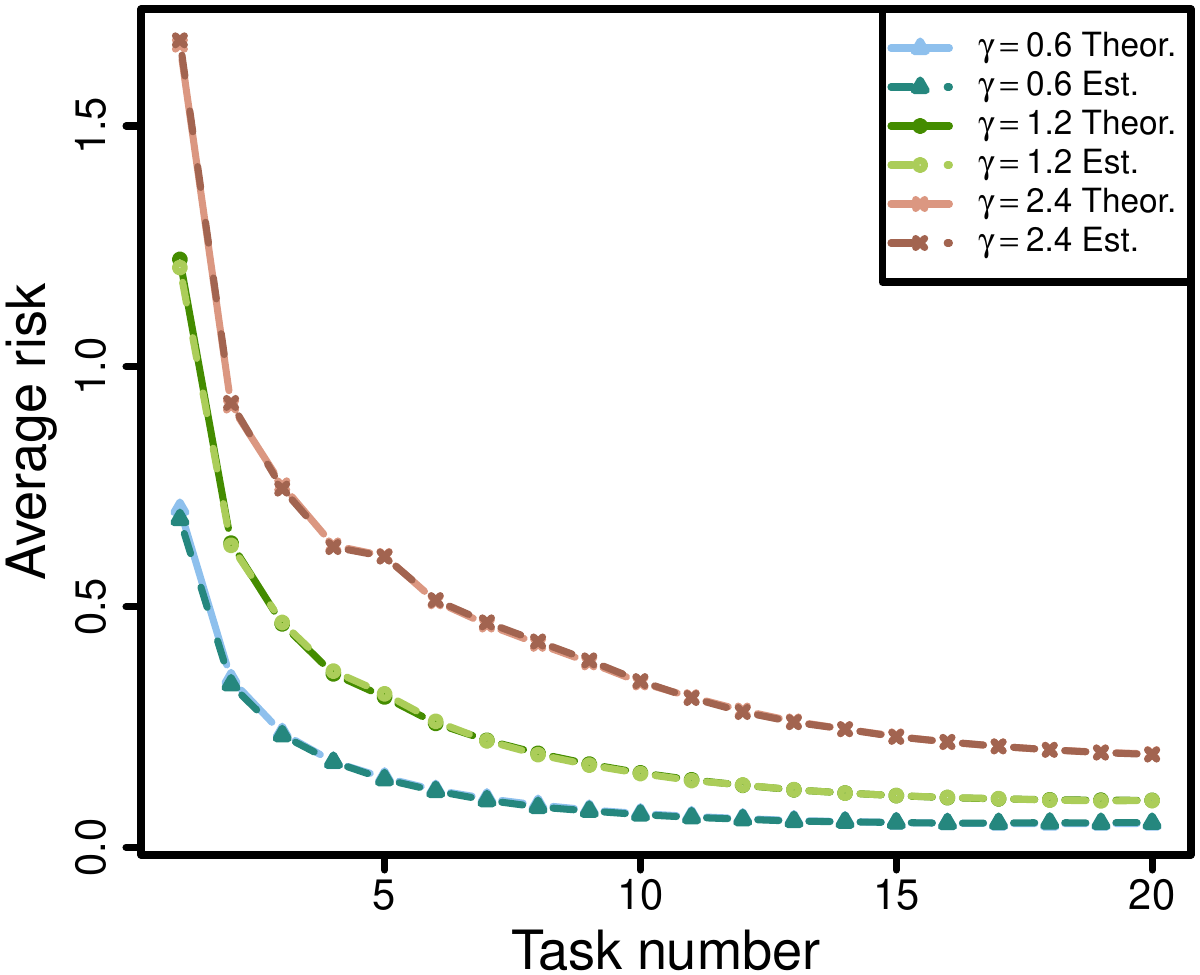}
    	\vspace{-0.13\textwidth}
        \caption{Average risk}
    \end{subfigure}
    \hspace{0.01\textwidth}
    \begin{subfigure}{0.31\textwidth}
        \includegraphics[width=\textwidth]{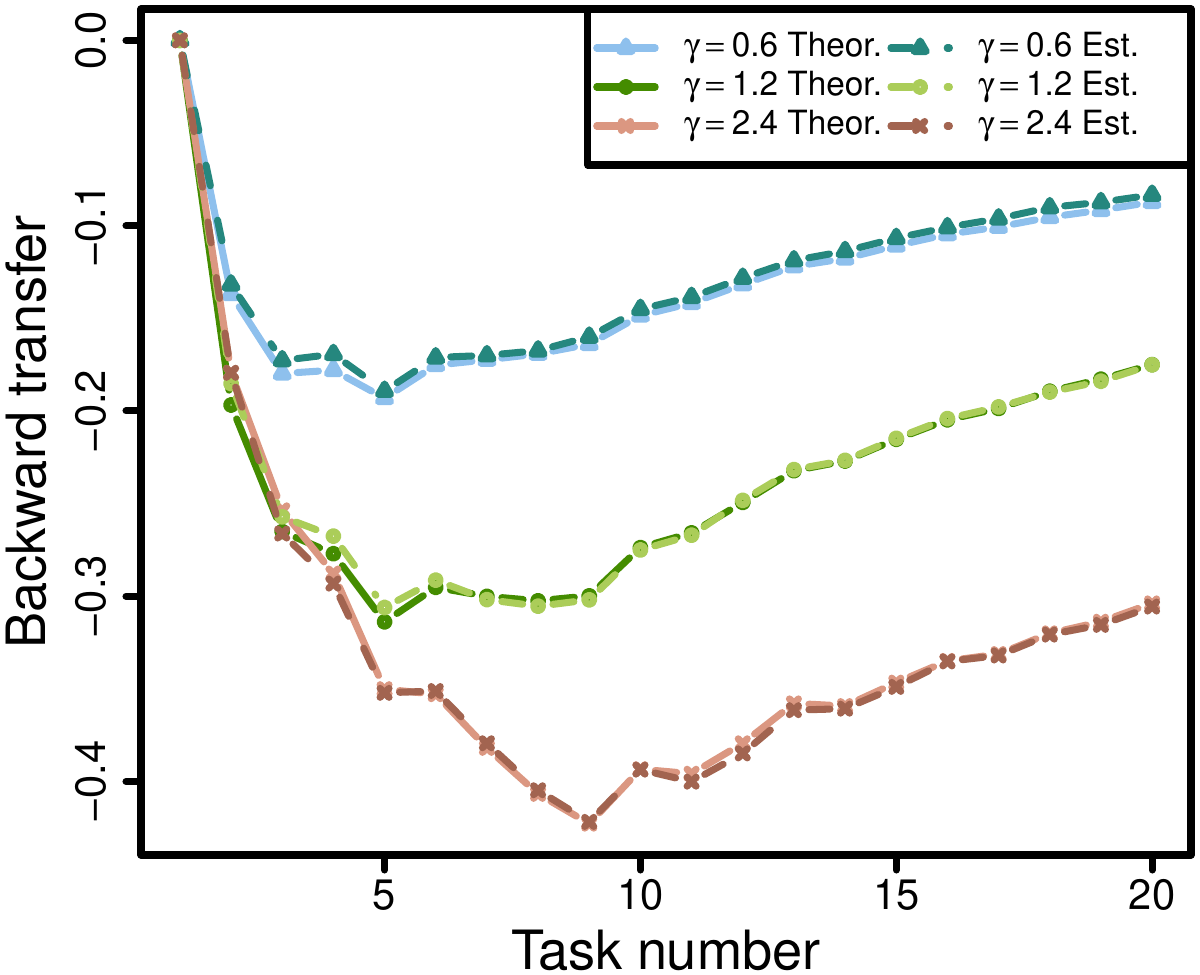}
    	\vspace{-0.13\textwidth}
        \caption{Backward transfer}
    \end{subfigure}
    \hspace{0.01\textwidth}
    \begin{subfigure}{0.31\textwidth}
        \includegraphics[width=\textwidth]{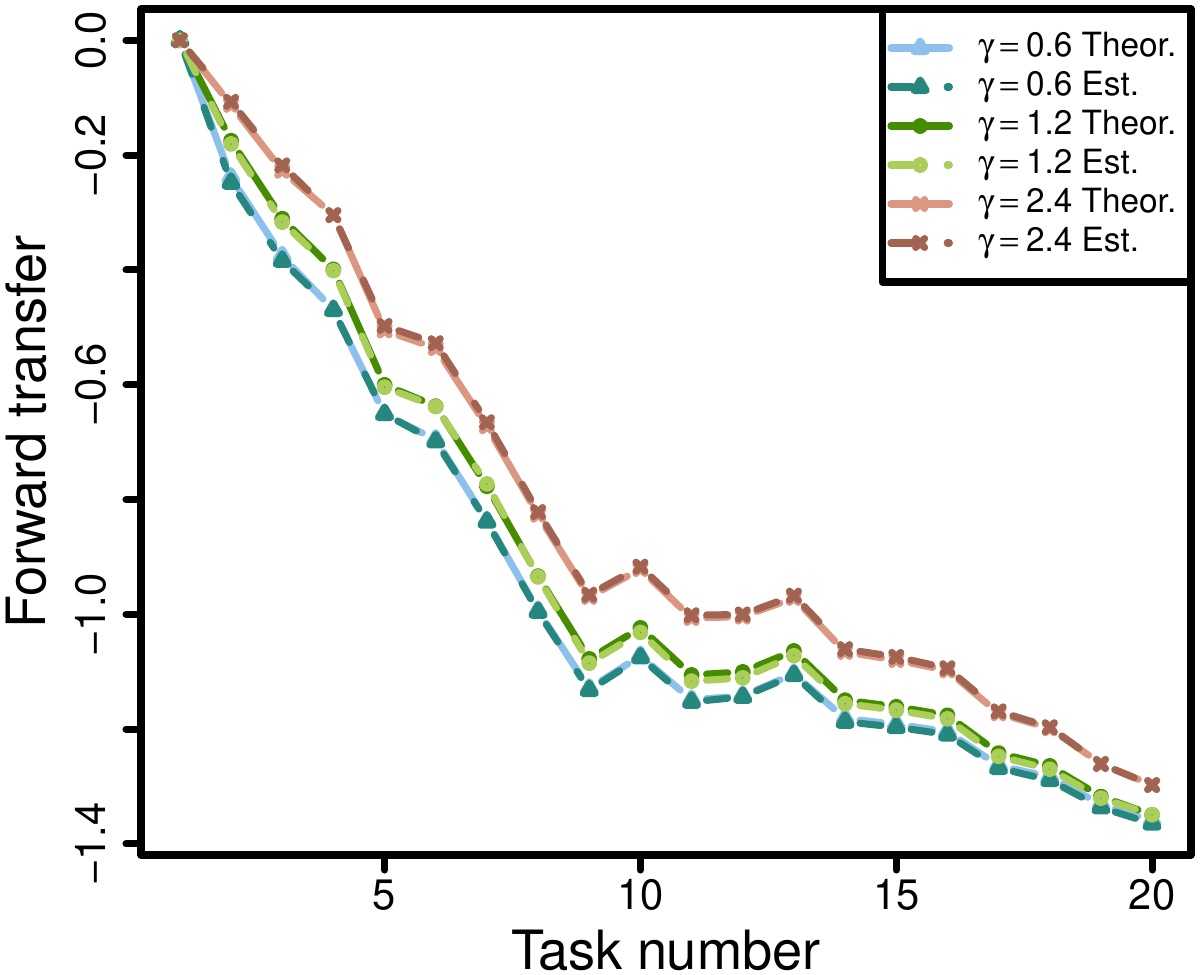}
    	\vspace{-0.13\textwidth}
        \caption{Forward transfer}
    \end{subfigure} \\
    \vspace{0.02\textwidth}
    \begin{subfigure}{0.31\textwidth}
    	\includegraphics[width=\textwidth]{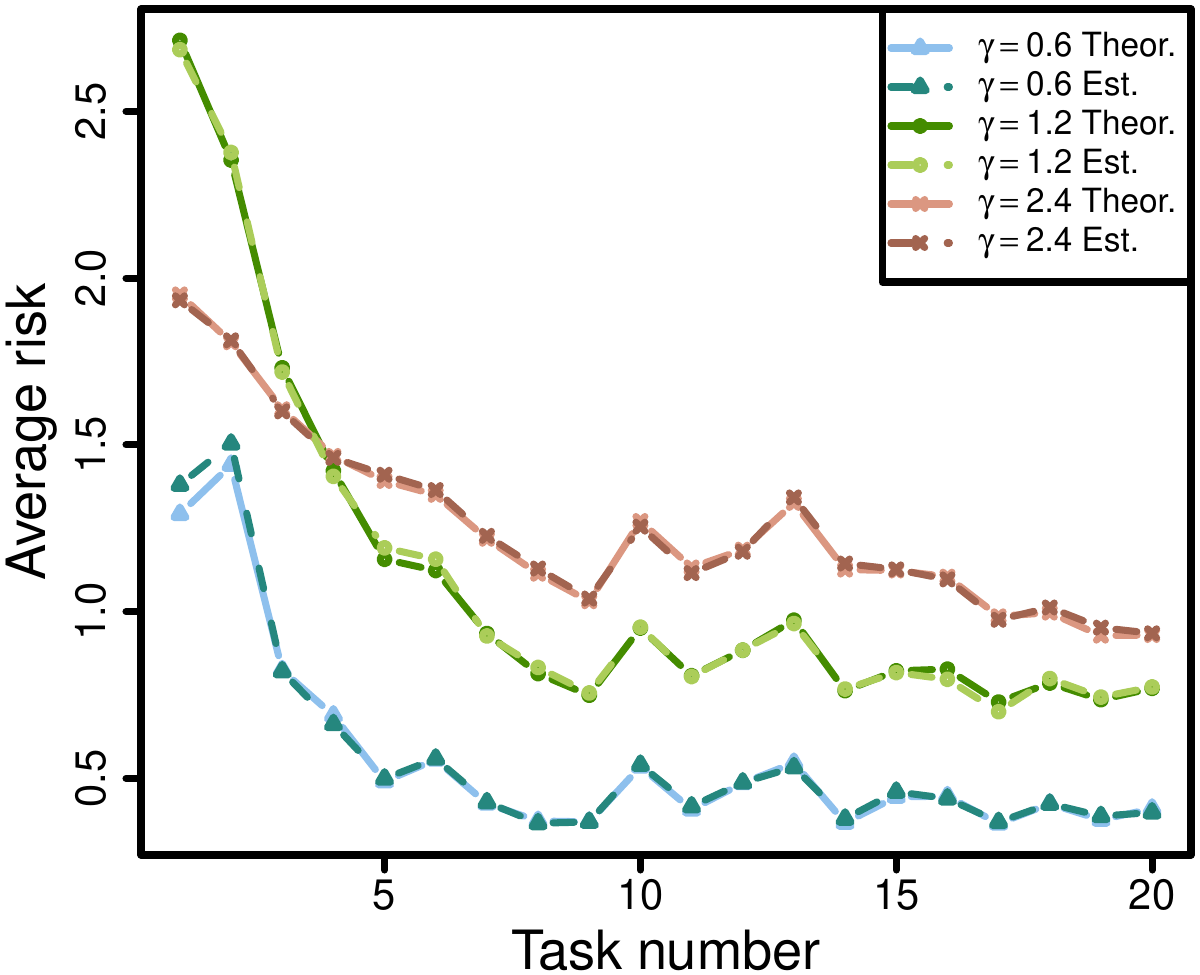}
    	\vspace{-0.13\textwidth}
        \caption{Average risk}
    \end{subfigure}
    \hspace{0.01\textwidth}
    \begin{subfigure}{0.31\textwidth}
        \includegraphics[width=\textwidth]{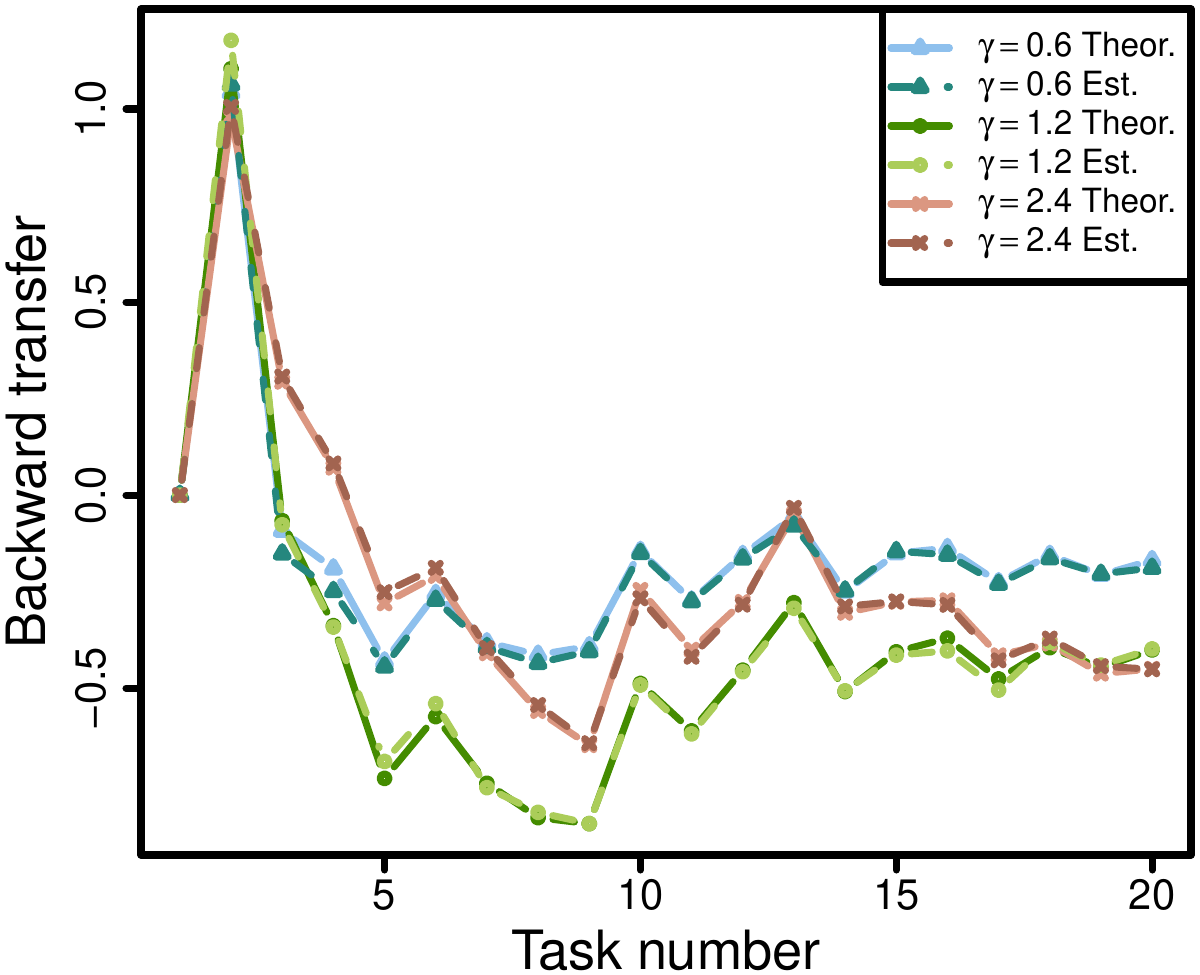}
    	\vspace{-0.13\textwidth}
        \caption{Backward transfer}
    \end{subfigure}
    \hspace{0.01\textwidth}
    \begin{subfigure}{0.31\textwidth}
        \includegraphics[width=\textwidth]{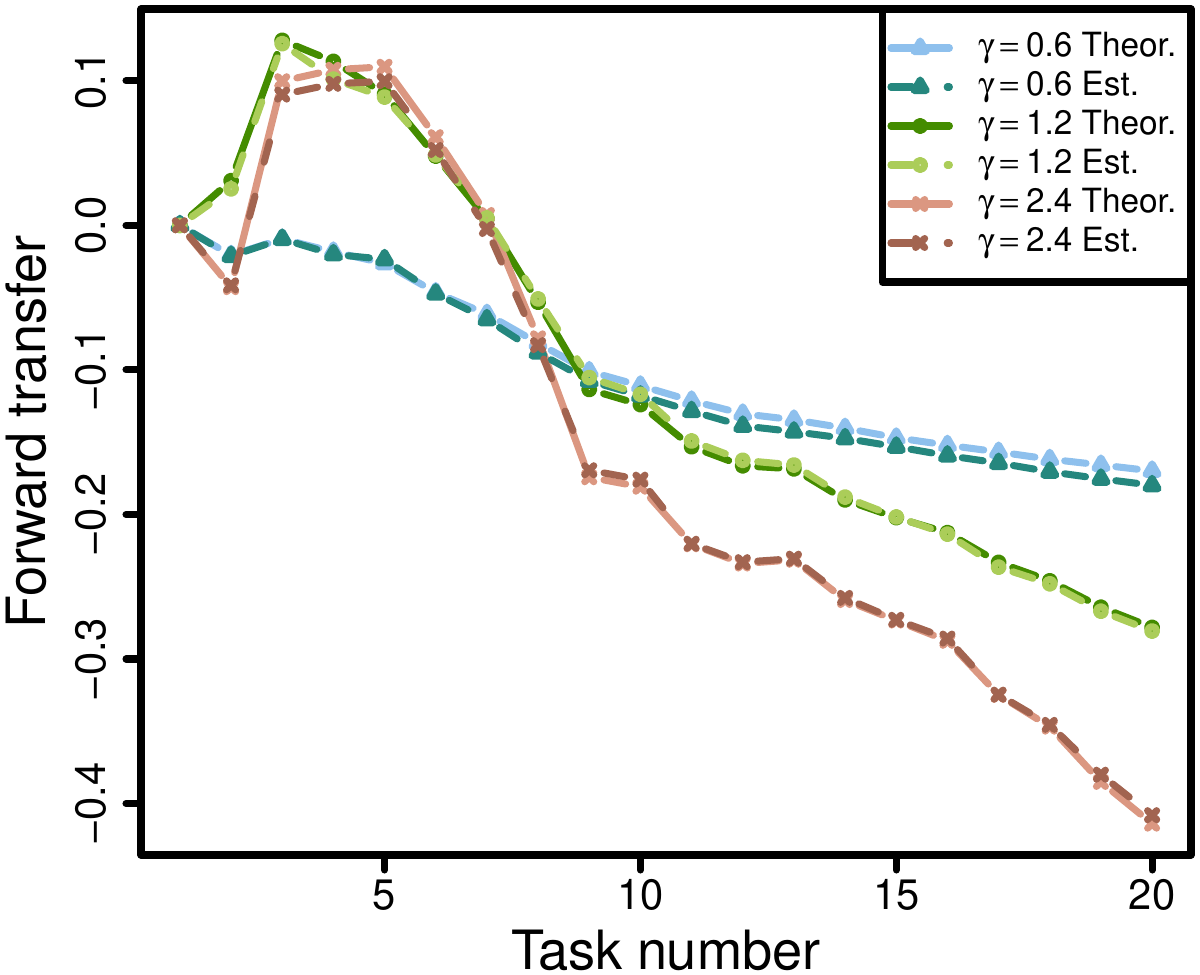}
    	\vspace{-0.13\textwidth}
        \caption{Forward transfer}
    \end{subfigure}\\
    \vspace{0.02\textwidth}
    \begin{subfigure}{0.31\textwidth}
    	\includegraphics[width=\textwidth]{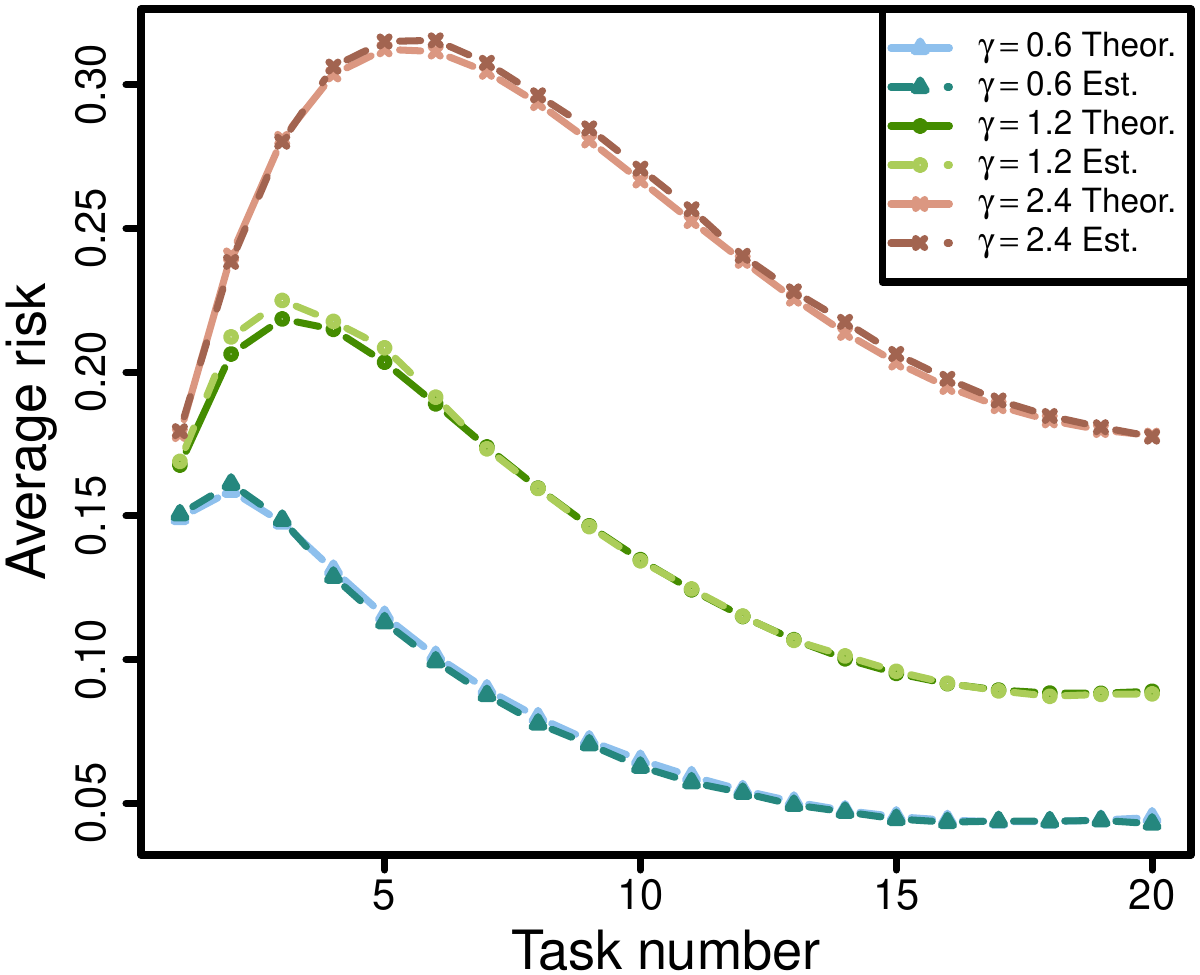}
    	\vspace{-0.13\textwidth}
        \caption{Average risk}
    \end{subfigure}
    \hspace{0.01\textwidth}
    \begin{subfigure}{0.31\textwidth}
        \includegraphics[width=\textwidth]{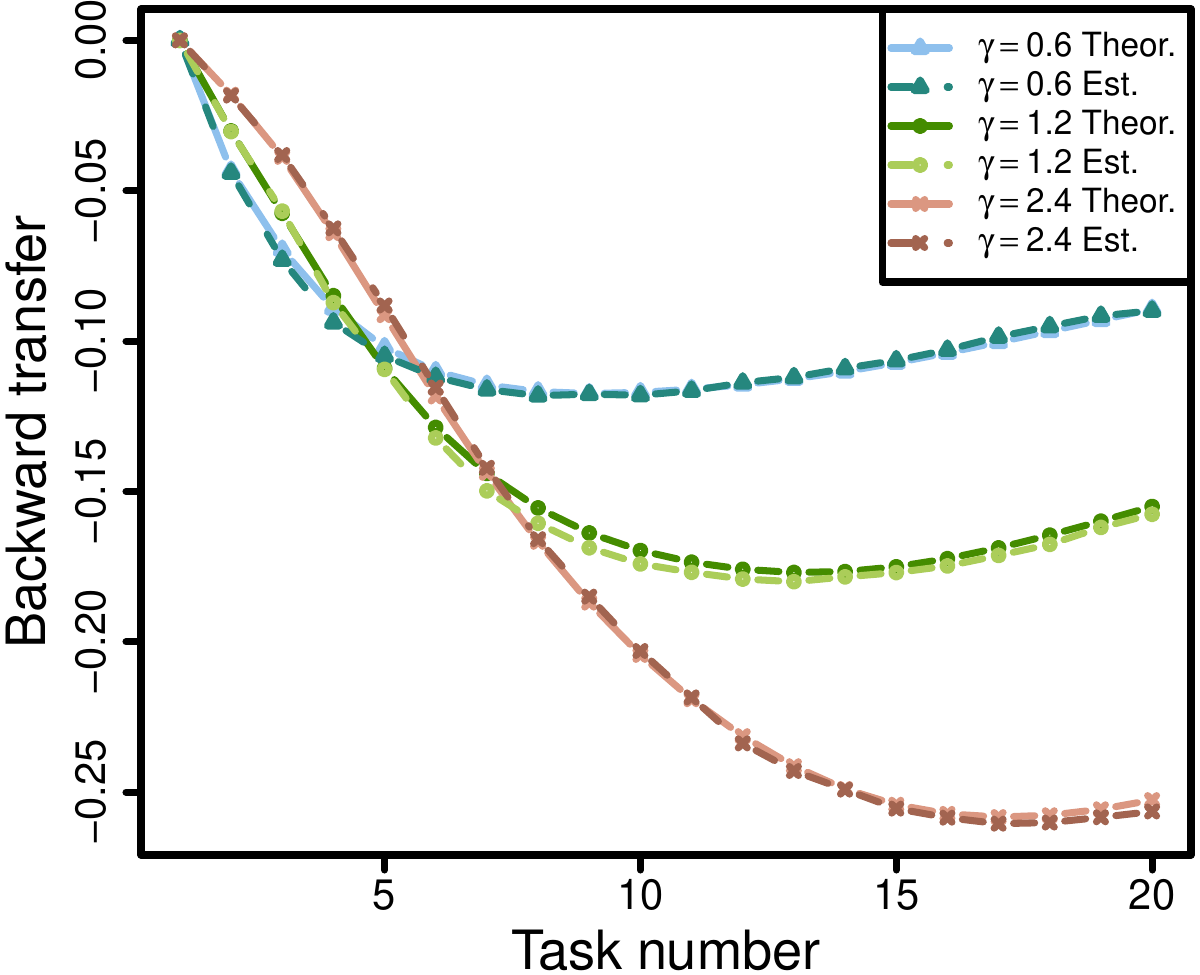}
    	\vspace{-0.13\textwidth}
        \caption{Backward transfer}
    \end{subfigure}
    \hspace{0.01\textwidth}
    \begin{subfigure}{0.31\textwidth}
        \includegraphics[width=\textwidth]{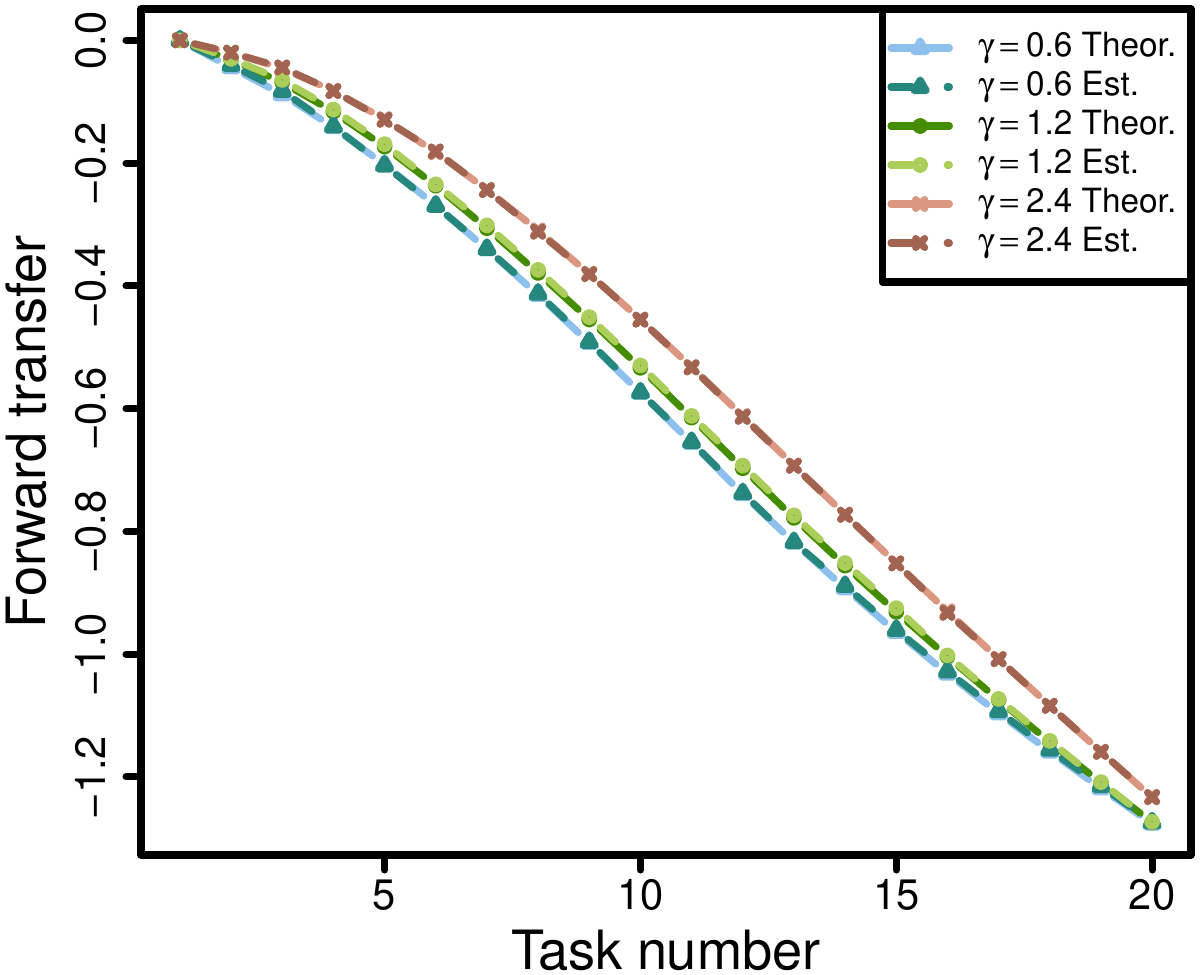}
    	\vspace{-0.13\textwidth}
        \caption{Forward transfer}
    \end{subfigure} \\
    \vspace{0.02\textwidth}
    \begin{subfigure}{0.31\textwidth}
    	\includegraphics[width=\textwidth]{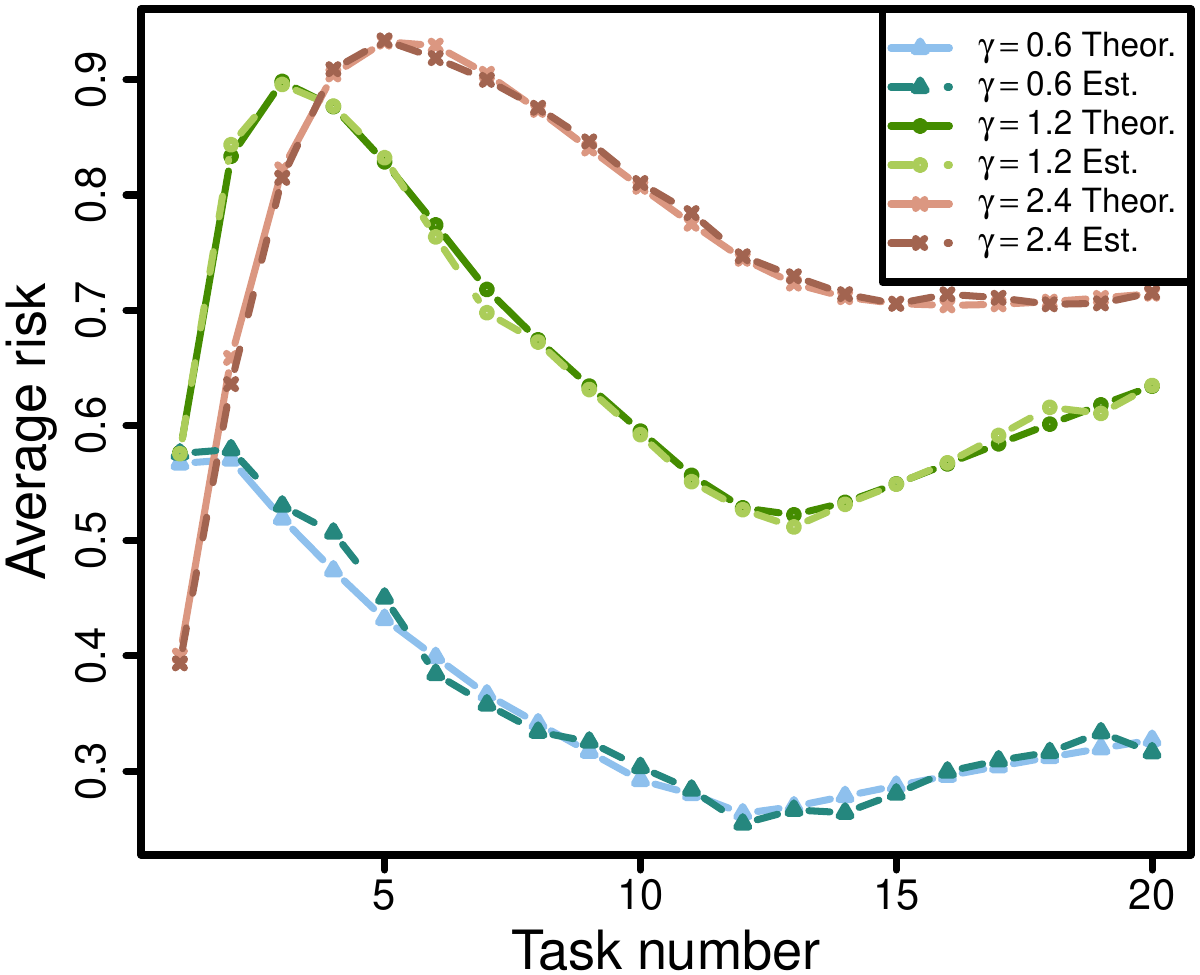}
    	\vspace{-0.13\textwidth}
        \caption{Average risk}
    \end{subfigure}
    \hspace{0.01\textwidth}
    \begin{subfigure}{0.31\textwidth}
        \includegraphics[width=\textwidth]{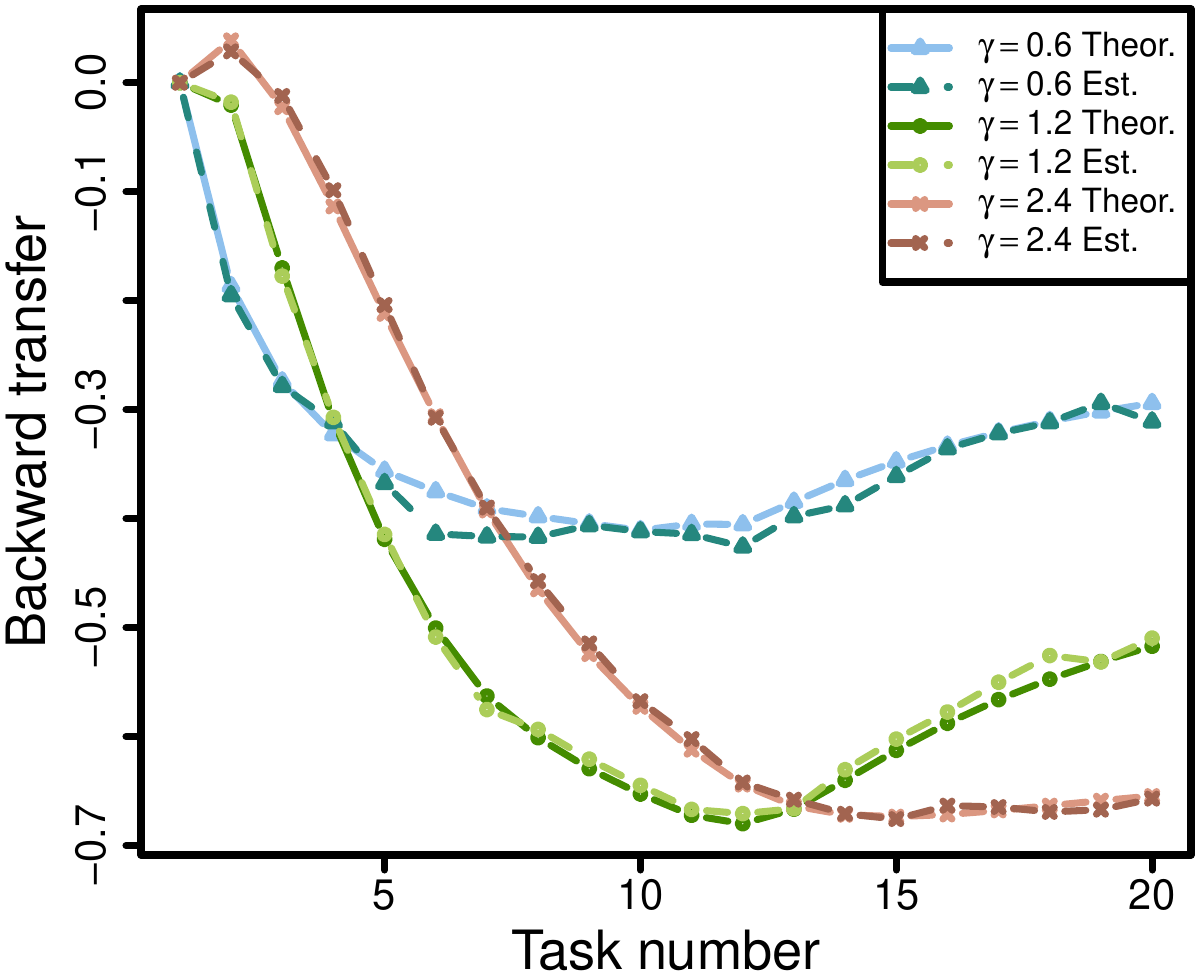}
    	\vspace{-0.13\textwidth}
        \caption{Backward transfer}
    \end{subfigure}
    \hspace{0.01\textwidth}
    \begin{subfigure}{0.31\textwidth}
        \includegraphics[width=\textwidth]{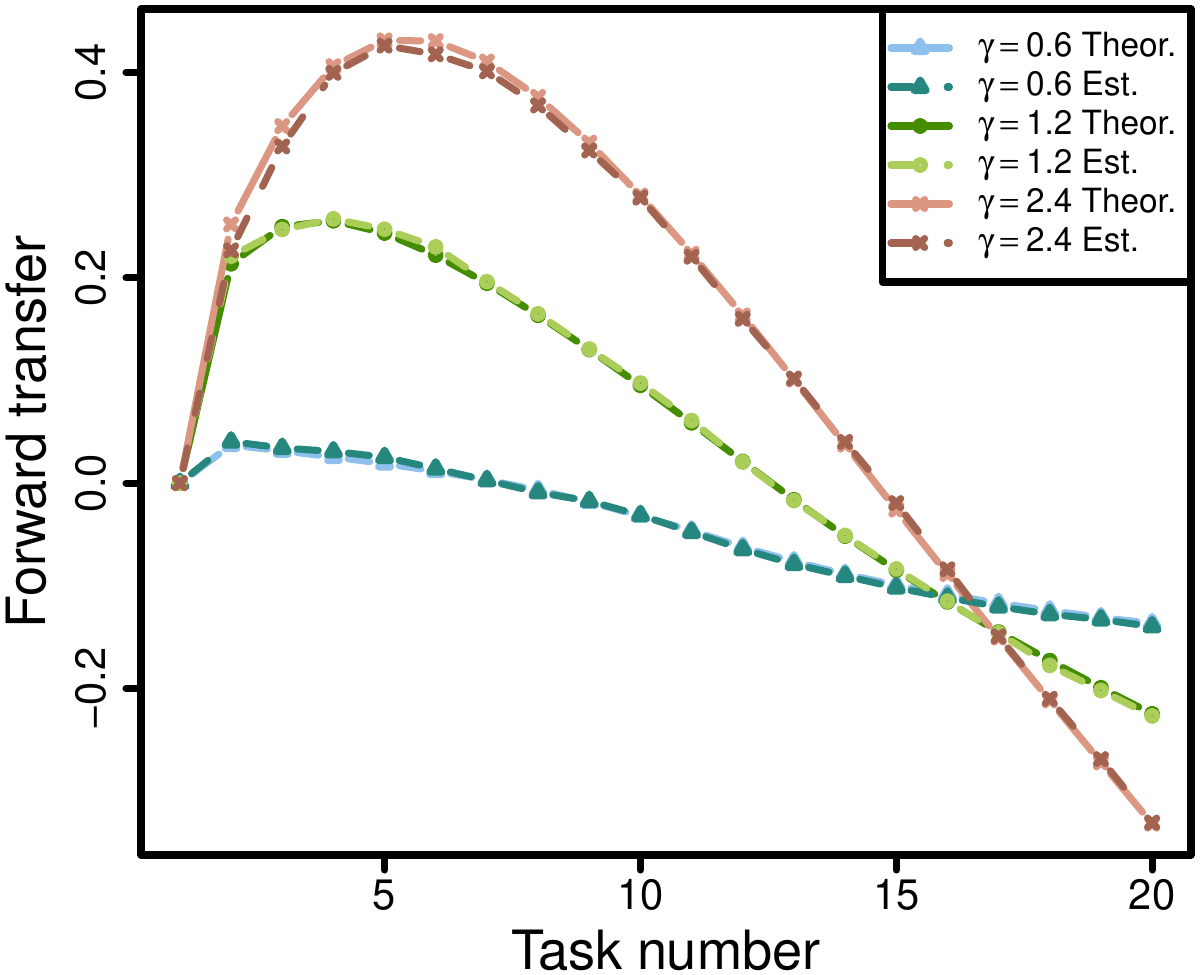}
    	\vspace{-0.13\textwidth}
        \caption{Forward transfer}
    \end{subfigure} 
    \caption{Risk curves with respect to the task number: In the first and second row, the covariance matrices satisfy random shift mechanism with $\bm{\lambda}=\bm{\lambda}_{st}$ and $\bm{\lambda}=\bm{\lambda}_{st}/20$ respectively. In the third and fourth row, the covariance matrices satisfy increasing shift mechanism with $\bm{\lambda}=\bm{\lambda}_{st}$ and $\bm{\lambda}=\bm{\lambda}_{st}/20$ respectively.}
    \label{101}
\end{figure*}


\begin{figure*}[!htbp]
    \centering
    \begin{subfigure}{0.31\textwidth}
    	\includegraphics[width=\textwidth]{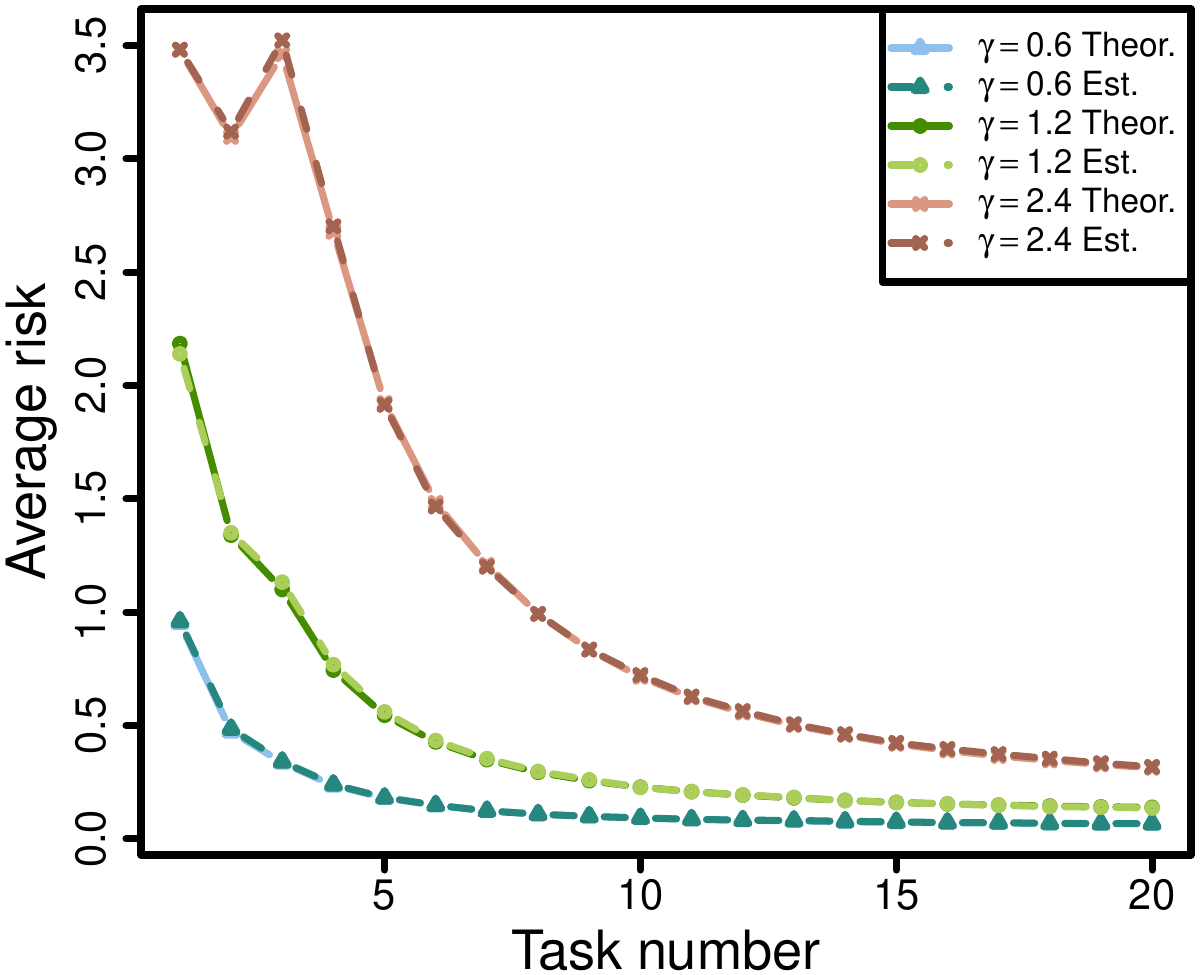}
    	\vspace{-0.13\textwidth}
        \caption{Average risk}
    \end{subfigure}
    \hspace{0.01\textwidth}
    \begin{subfigure}{0.31\textwidth}
        \includegraphics[width=\textwidth]{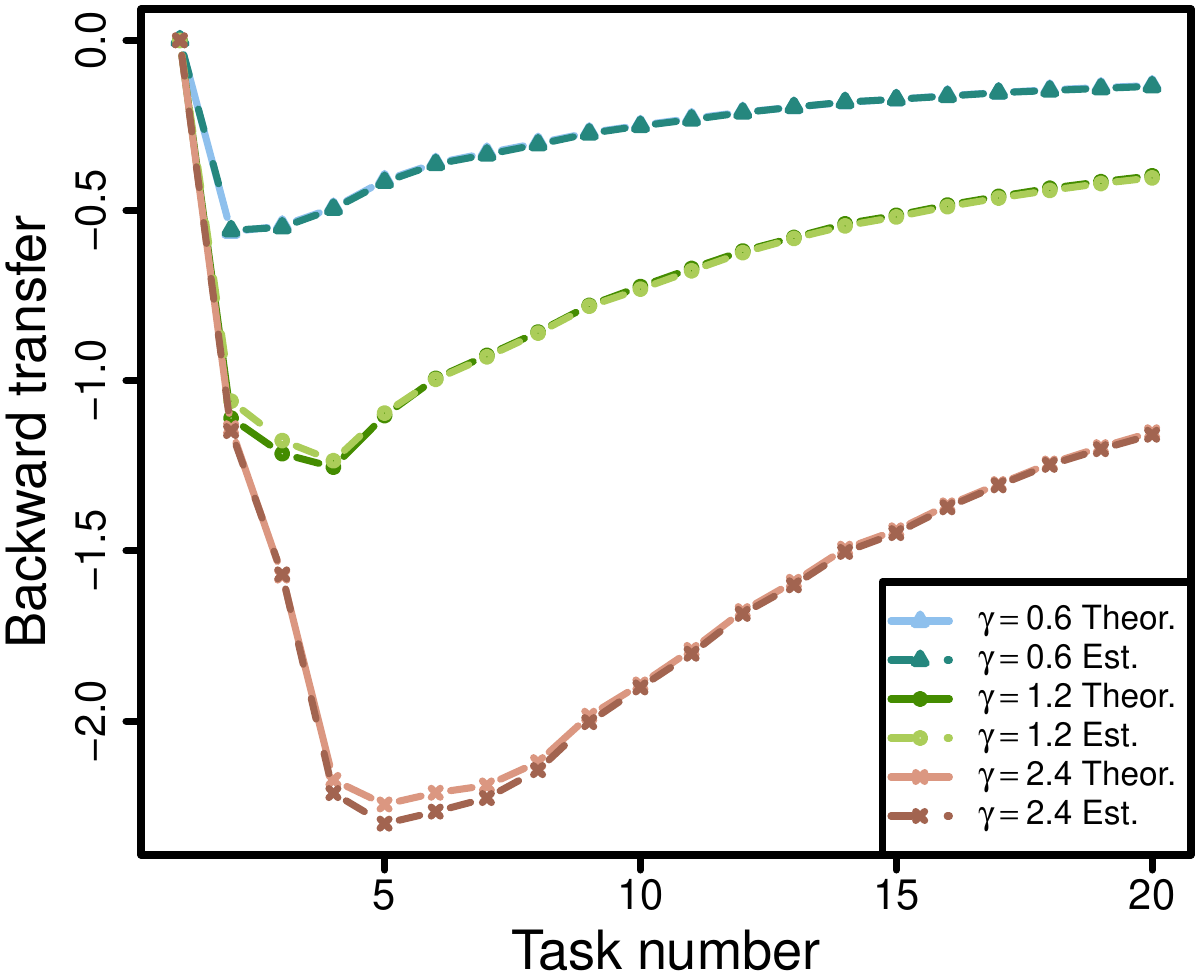}
    	\vspace{-0.13\textwidth}
        \caption{Backward transfer}
    \end{subfigure}
    \hspace{0.01\textwidth}
    \begin{subfigure}{0.31\textwidth}
        \includegraphics[width=\textwidth]{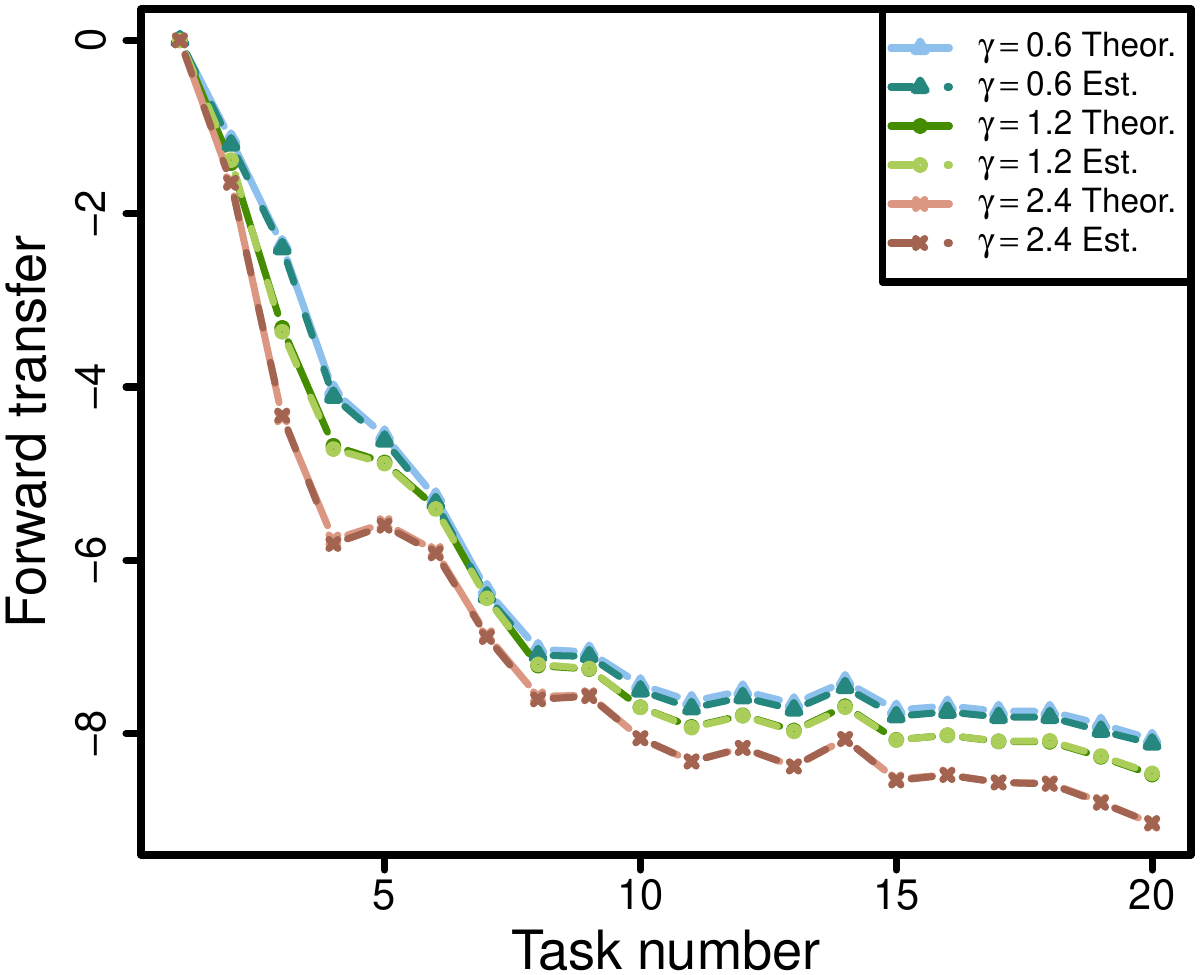}
    	\vspace{-0.13\textwidth}
        \caption{Forward transfer}
    \end{subfigure} \\
    \vspace{0.02\textwidth}
    \begin{subfigure}{0.31\textwidth}
    	\includegraphics[width=\textwidth]{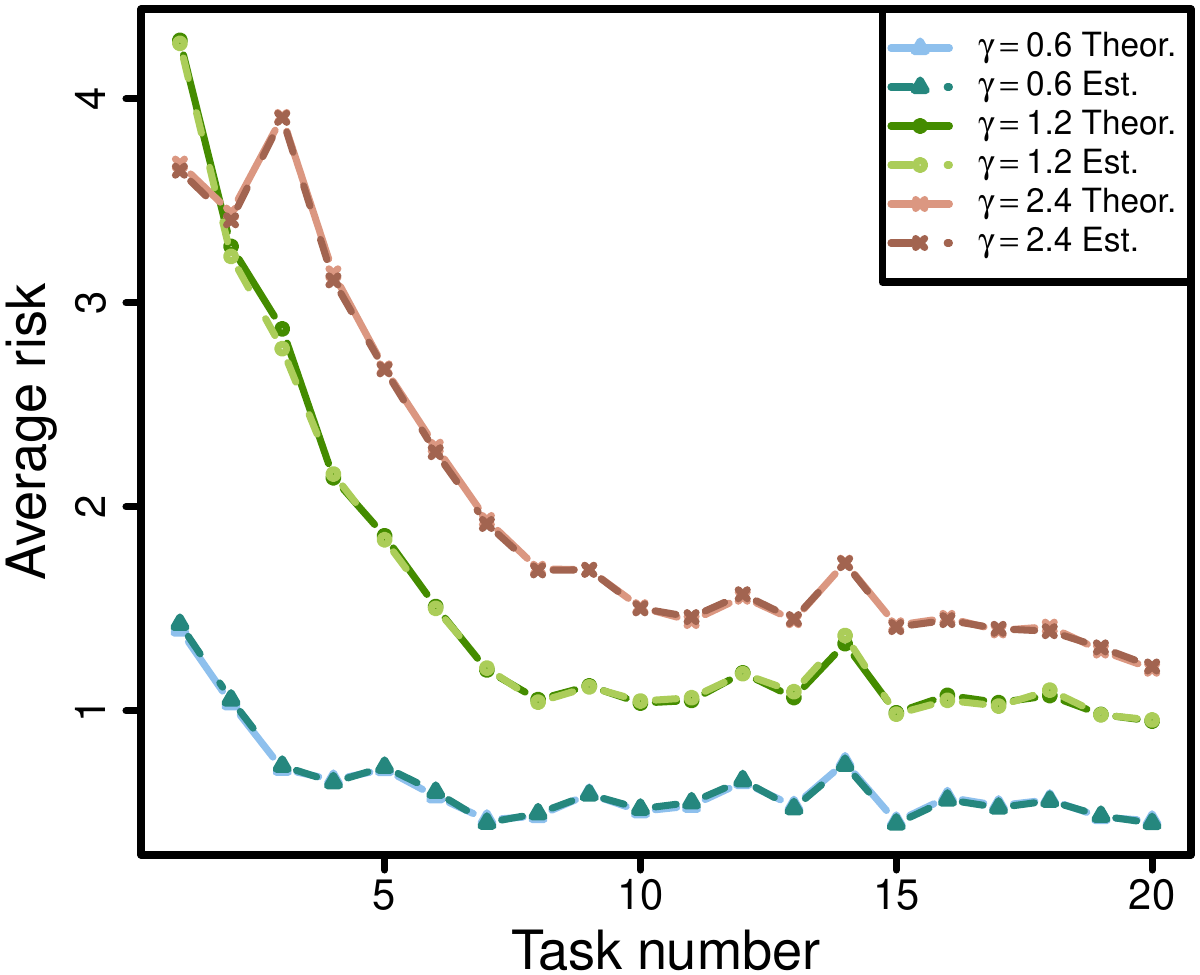}
    	\vspace{-0.13\textwidth}
        \caption{Average risk}
    \end{subfigure}
    \hspace{0.01\textwidth}
    \begin{subfigure}{0.31\textwidth}
        \includegraphics[width=\textwidth]{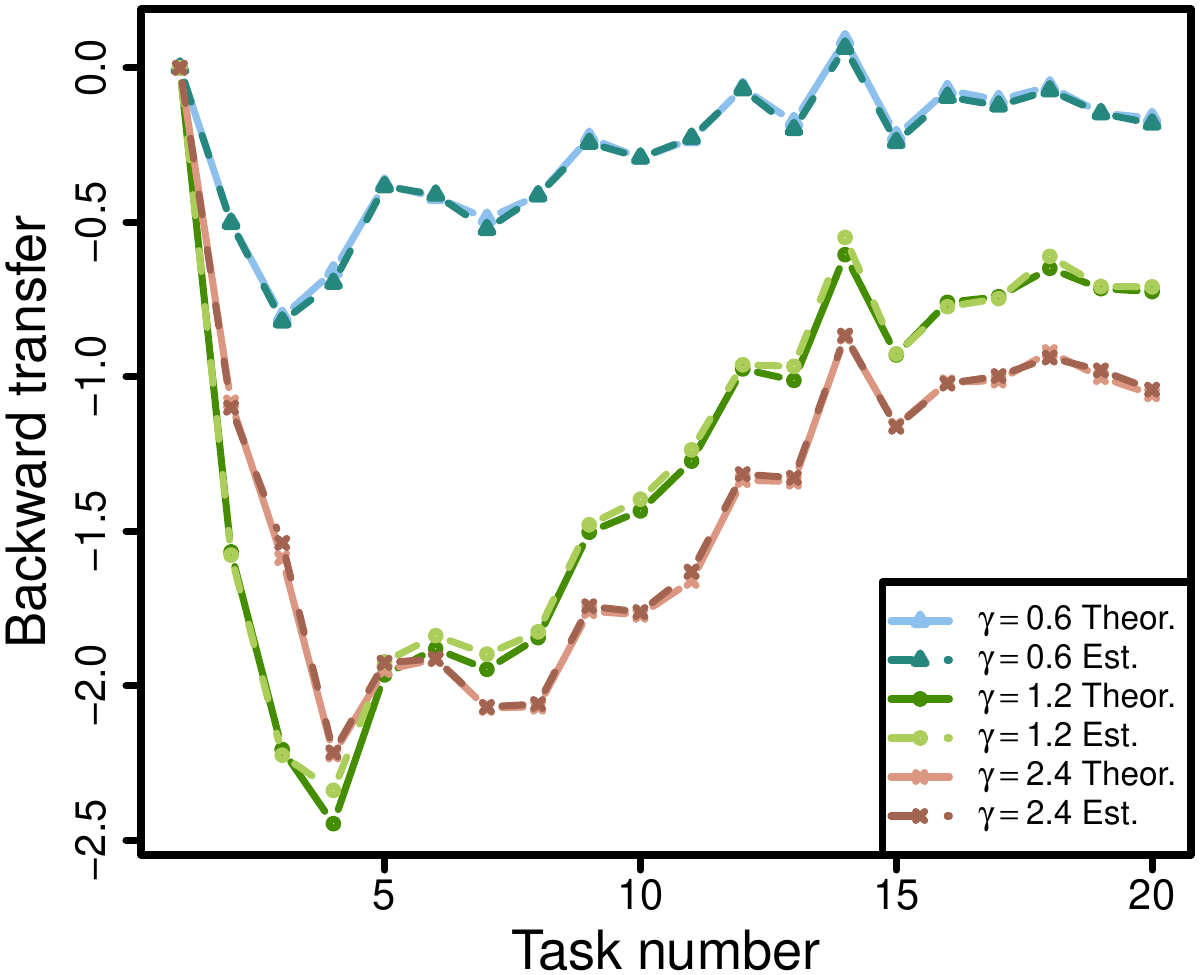}
    	\vspace{-0.13\textwidth}
        \caption{Backward transfer}
    \end{subfigure}
    \hspace{0.01\textwidth}
    \begin{subfigure}{0.31\textwidth}
        \includegraphics[width=\textwidth]{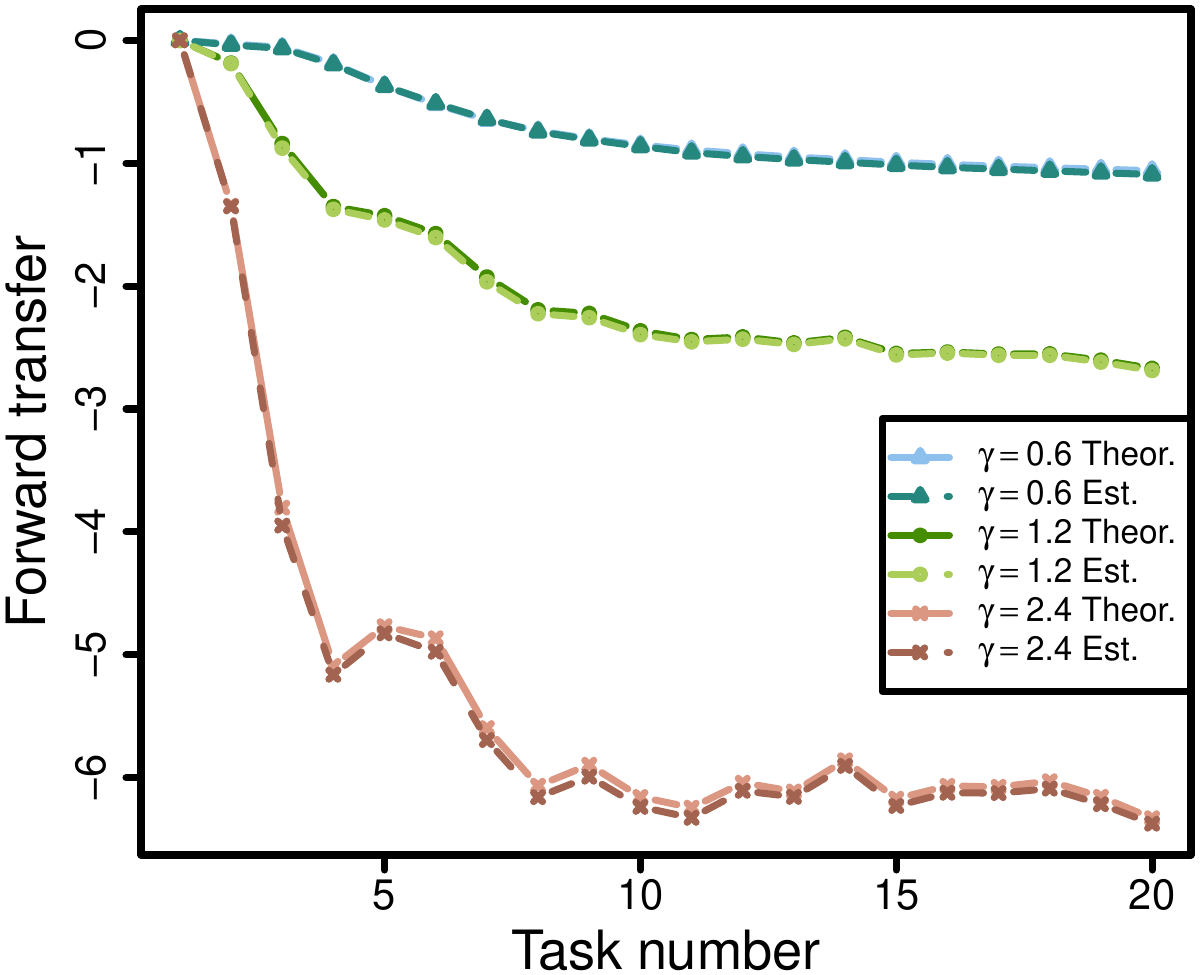}
    	\vspace{-0.13\textwidth}
        \caption{Forward transfer}
    \end{subfigure}\\
    \vspace{0.02\textwidth}
    \begin{subfigure}{0.31\textwidth}
    	\includegraphics[width=\textwidth]{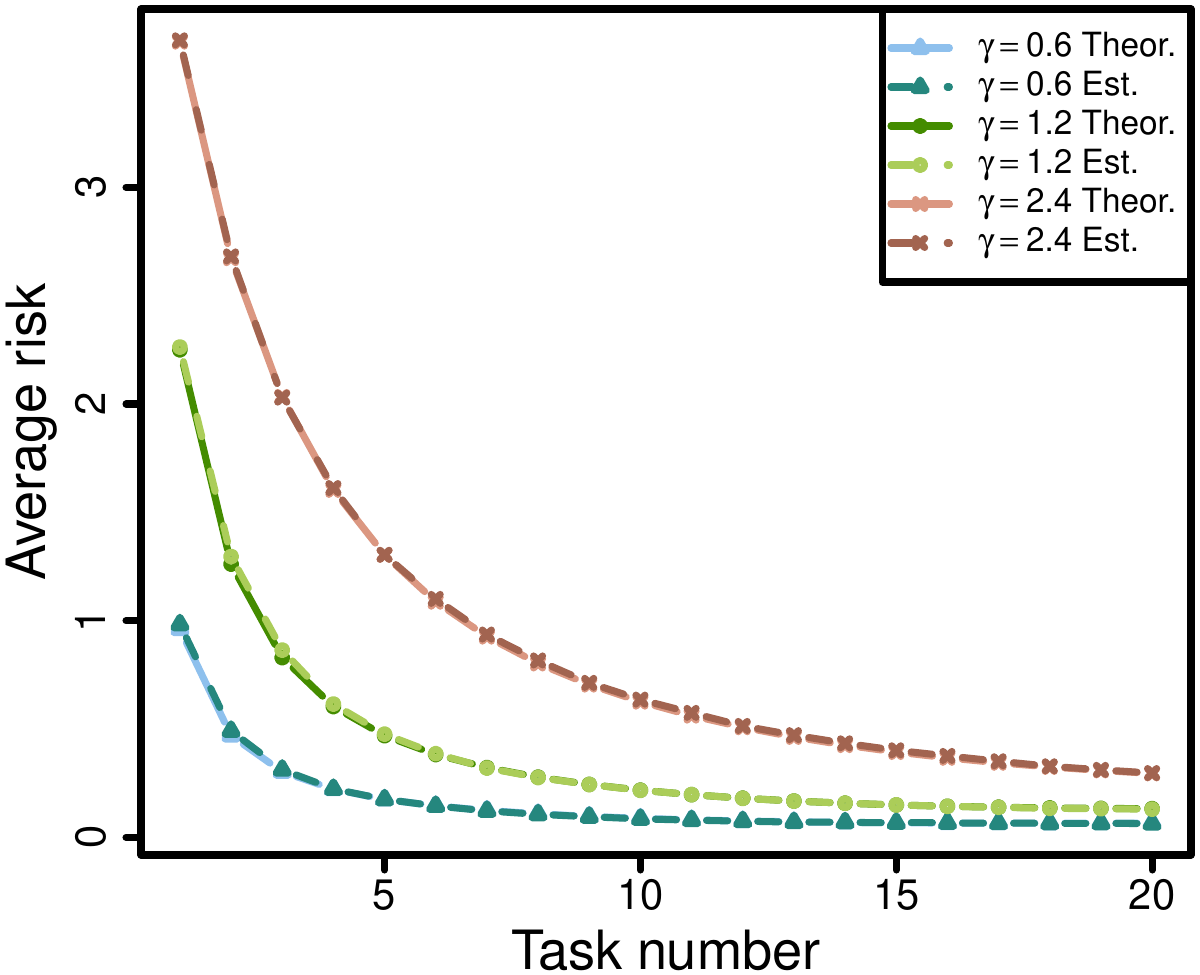}
    	\vspace{-0.13\textwidth}
        \caption{Average risk}
    \end{subfigure}
    \hspace{0.01\textwidth}
    \begin{subfigure}{0.31\textwidth}
        \includegraphics[width=\textwidth]{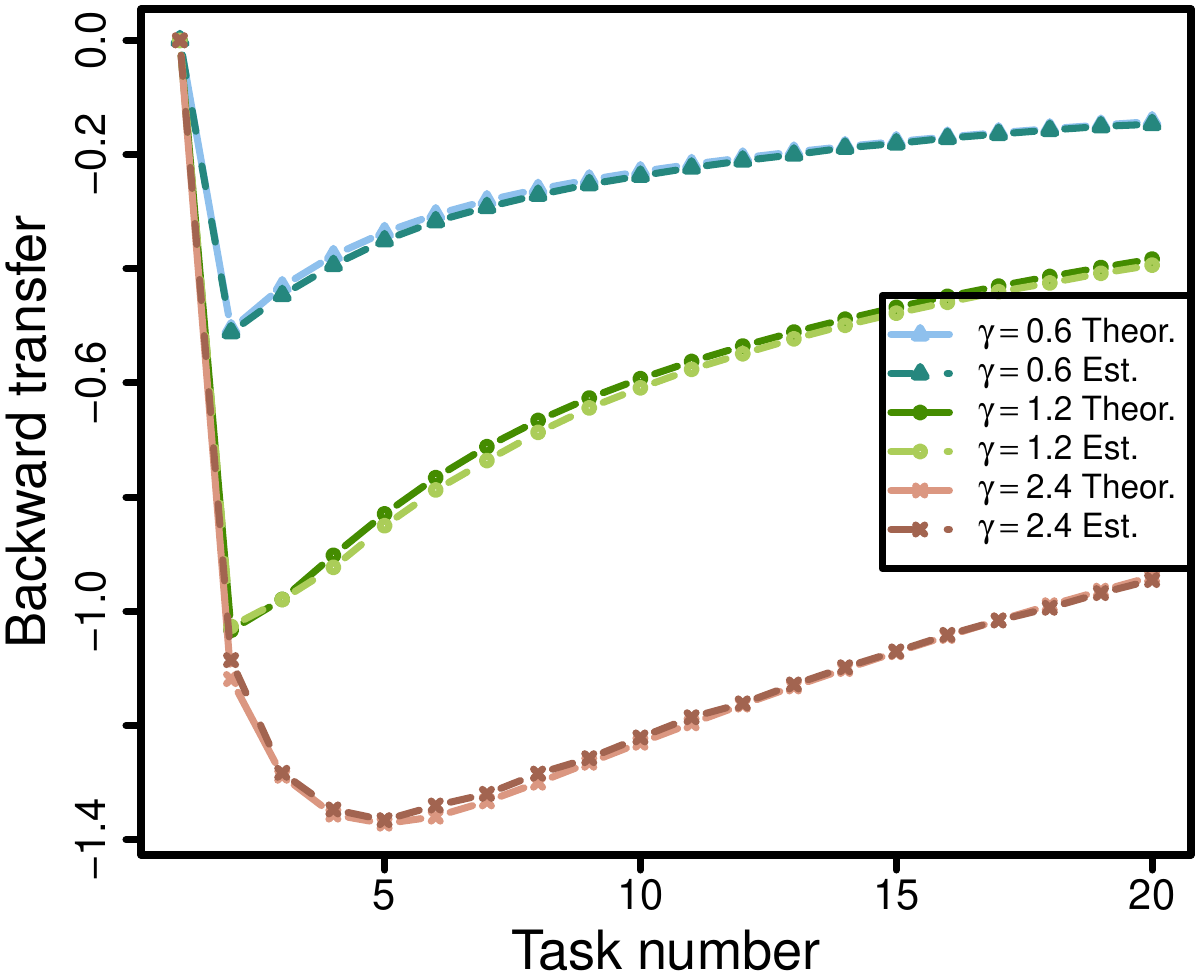}
    	\vspace{-0.13\textwidth}
        \caption{Backward transfer}
    \end{subfigure}
    \hspace{0.01\textwidth}
    \begin{subfigure}{0.31\textwidth}
        \includegraphics[width=\textwidth]{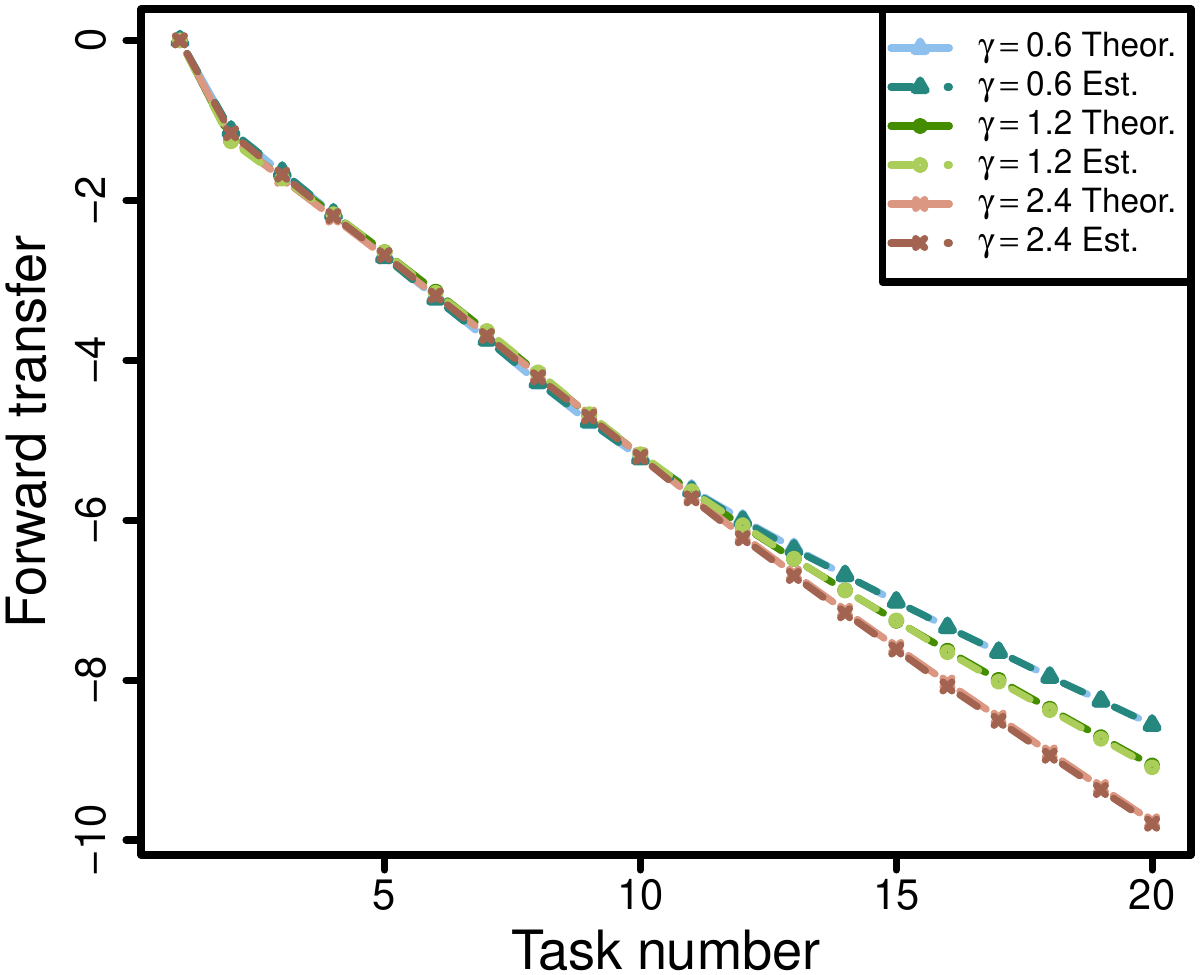}
    	\vspace{-0.13\textwidth}
        \caption{Forward transfer}
    \end{subfigure} \\
    \vspace{0.02\textwidth}
    \begin{subfigure}{0.31\textwidth}
    	\includegraphics[width=\textwidth]{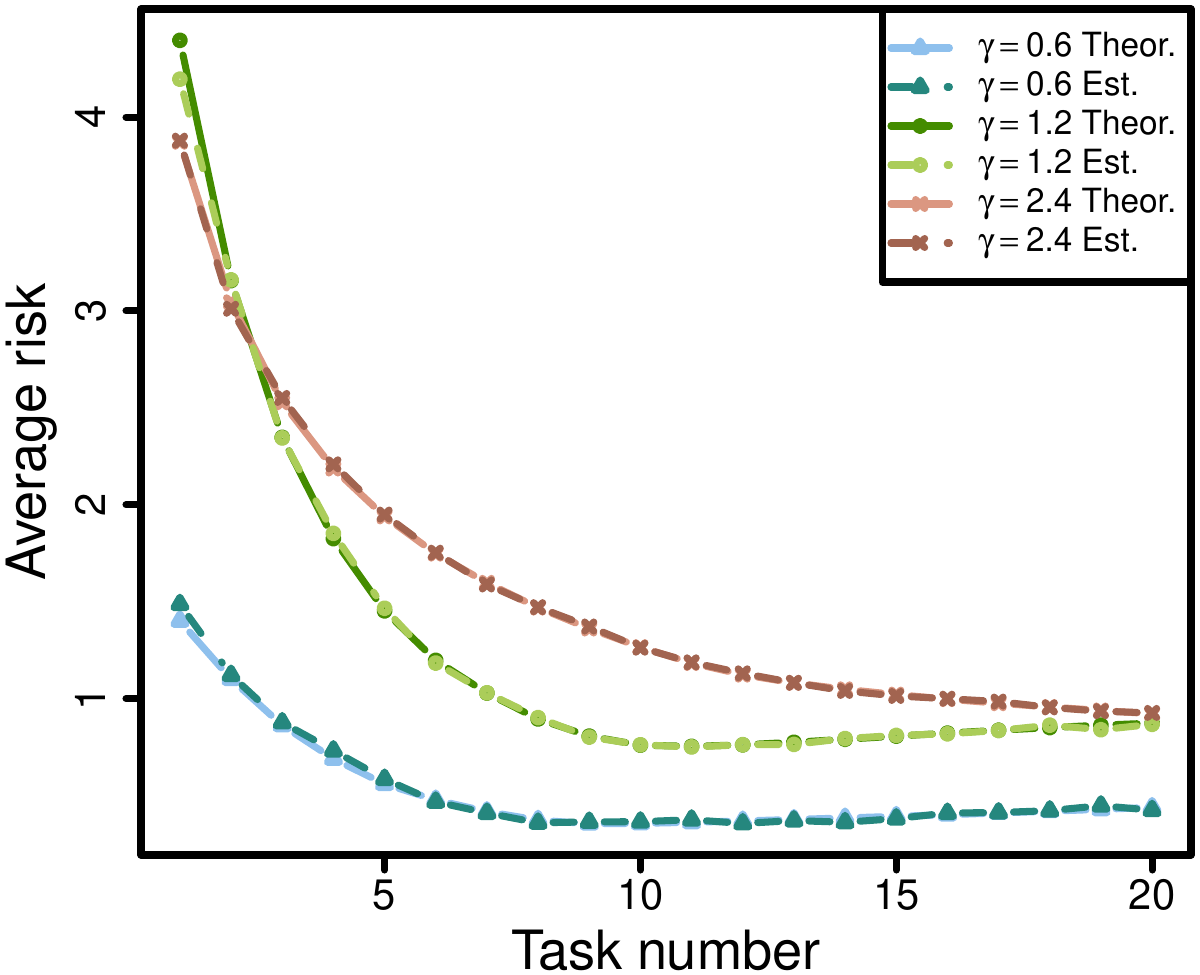}
    	\vspace{-0.13\textwidth}
        \caption{Average risk}
    \end{subfigure}
    \hspace{0.01\textwidth}
    \begin{subfigure}{0.31\textwidth}
        \includegraphics[width=\textwidth]{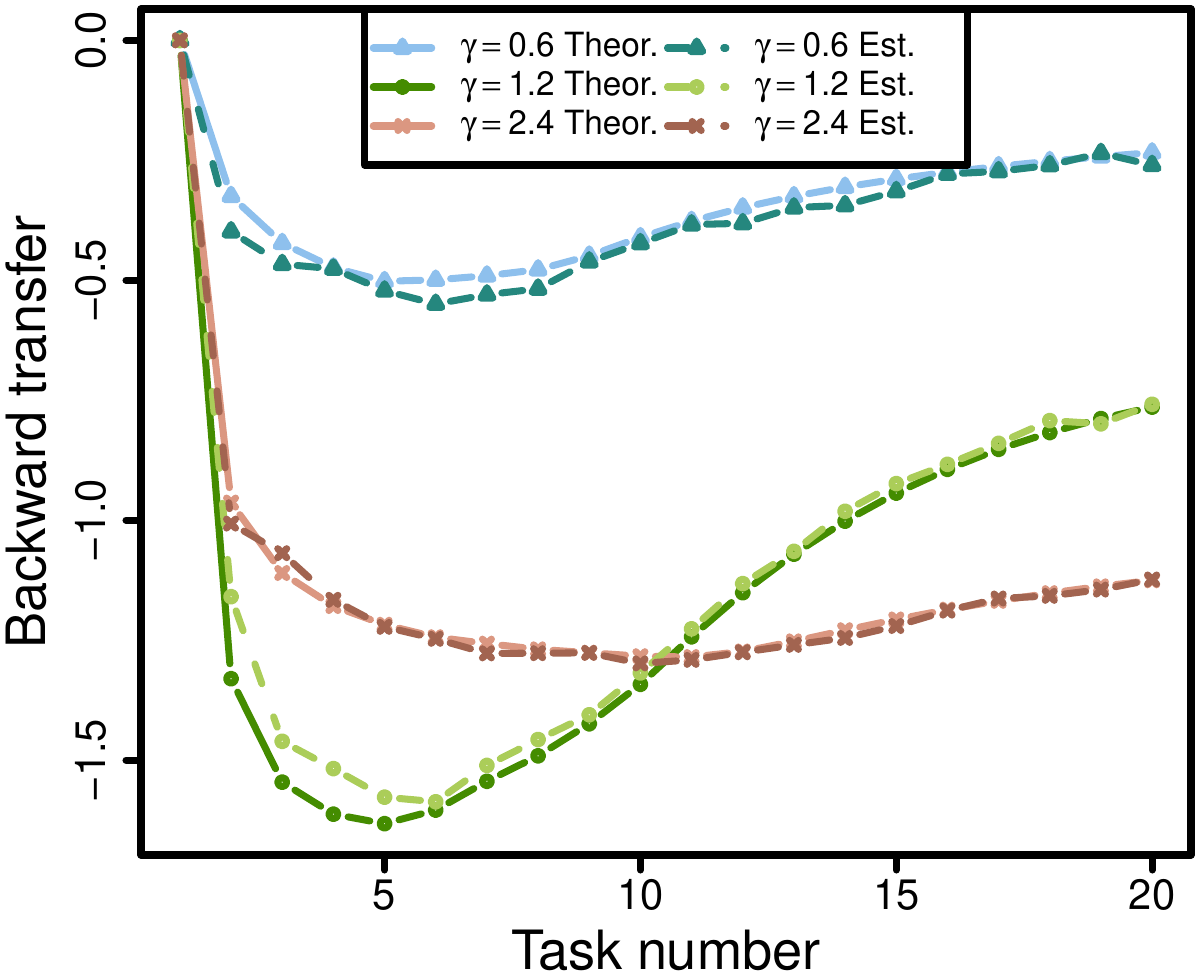}
    	\vspace{-0.13\textwidth}
        \caption{Backward transfer}
    \end{subfigure}
    \hspace{0.01\textwidth}
    \begin{subfigure}{0.31\textwidth}
        \includegraphics[width=\textwidth]{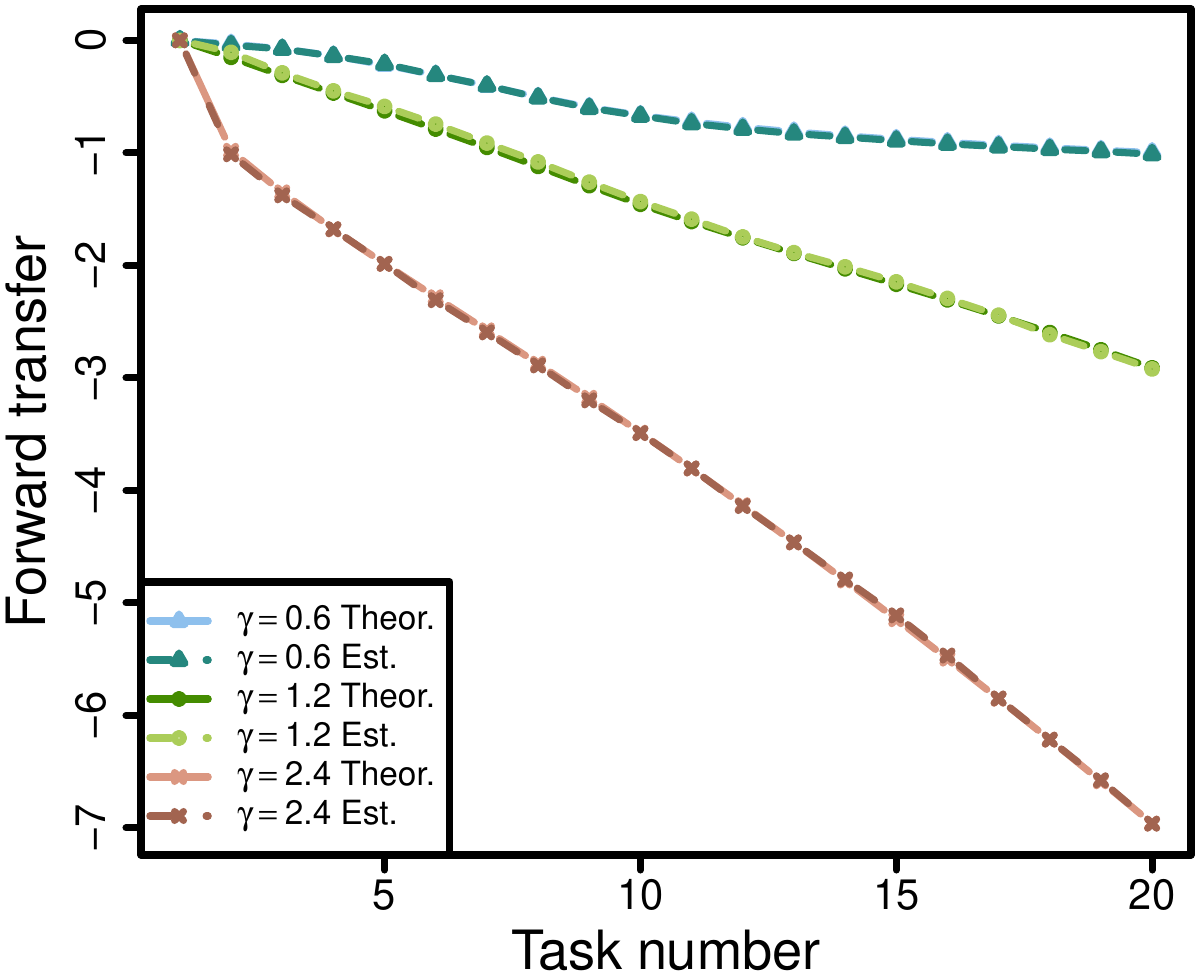}
    	\vspace{-0.13\textwidth}
        \caption{Forward transfer}
    \end{subfigure} 
    \caption{Risk curves with respect to the task number: In the first and second row, the covariance matrices have random block size with $\bm{\lambda}=\bm{\lambda}_{st}$ and $\bm{\lambda}=\bm{\lambda}_{st}/20$ respectively. In the third and fourth row, the covariance matrices have increasing block size with $\bm{\lambda}=\bm{\lambda}_{st}$ and $\bm{\lambda}=\bm{\lambda}_{st}/20$ respectively.}
    \label{102}
\end{figure*}

Our experimental results show that the performance of continual risk regression is highly sensitive to regularization parameters. With properly chosen regularization parameters, the average risk decreases nearly monotonically when the number of tasks is large enough, and the forward transfer becomes increasingly pronounced over time, providing that the task-specific covariance matrices exhibit reasonable structures. Conversely, inappropriate regularization parameters lead to performance reduction, resulting in an unstable relationship between the number of tasks and both average risk and transfer capacities.

\section{Conclusions}
\label{63}

In this paper, we establish a rigorous theoretical framework for analyzing continual ridge regression in high-dimensional linear models with random designs. More specifically, we derive exact asymptotic expressions for prediction risk, explicitly characterizing its dependence on model complexity and task similarity. Average risk, backward transfer, and forward transfer are formalized to evaluate stepwise generalization performance, with their asymptotic behaviors related to asymptotic prediction risk. Through three representative examples, we demonstrate how the performance of risk curves vary with structures of task-specific covariance and choices of regularization parameters, supported by numerical simulations. Our results enrich theoretical analysis for continual learning methods, offering interpretable insights into the interplay of dimensionality, task relationships, and regularization in mitigating catastrophic forgetting. 

An important direction for future work is to establish a data-driven choice for regularization parameters as well as its theoretical guarantees. Our present result can help evaluate the data-driven tuning strategies by providing the interpretable results for oracle regularization parameters. Additionally, the techniques in this work may be helpful for theoretical analysis on continual learning in nonlinear model, which is also considered in our future research.


\acks{This research was partially supported by the National Natural Science Foundation of China (Grant No. 12271286). Part of Wenqing Su’s work was conducted at Tsinghua University during the author’s doctoral studies.}



\appendix

\section{Technical Results for Main Proofs}
\subsection{Deterministic Equivalent}
Before proving the main theorem, we introduce a technique of deterministic equivalence in random matrix theory \citep{10.1214/105051606000000925, MR2817033}. 

\begin{definition}
 Let $A_n, B_n \in \mathbb{R}^{n \times n}$ be sequences of random or deterministic symmetric random matrices, we say $A_n, B_n$ are equivalent, if for any sequence of deterministic matrices $C_n \in \mathbb{R}^{n \times n}$  such that
 \begin{align*}
 \limsup_{n}\Vert C_n\Vert_{op} < \infty,
\end{align*}
and for any sequence of deterministic vectors $v_n \in \mathbb{R}^{n}$  such that
 \begin{align*}
 \limsup_{n}\Vert v_n\Vert < \infty,
\end{align*}
we have, almost surely,
\begin{align*}
 \frac{1}{n}{\rm Tr}[C_n(A_n-B_n)] \rightarrow 0, \quad v_n^\top(A_n-B_n)v_n \rightarrow 0.
\end{align*}
Here the operator norm is defined by  $\Vert A \Vert_{op}:=\sqrt{\lambda_{max}(A^\top A)}$ for any real matrix $A \in \mathbb{R}^{p \times q}$.
\end{definition}

We write $A_n \asymp B_n$ if $A_n, B_n$ are equivalent. For random matrix $A_n$, if there exists a (sequence of) deterministic matrix $\bar{A}_n$ such that $A_n \asymp \bar{A}_n$, then $\bar{A}_n$ is said to be a deterministic equivalent for $A_n$.

There are some useful rules of calculus for matrix equivalents. 
\begin{proposition}
~
\\
1. If $A_n \asymp B_n, C_n \asymp D_n$, then $A_n+C_n \asymp B_n+D_n$. \\
2. If $A_n \asymp B_n$, $E_n$ is a sequence of deterministic matrix such that $\Vert E_n \Vert_{op} < \infty$, then $E_nA_nE_n^\top \asymp E_nB_nE_n^\top$. \\
3. If $A_n \asymp B_n$, then $\frac{1}{n} {\rm Tr}[A_n -B_n]\rightarrow 0$, almost surely.
\label{16}
\end{proposition}
\begin{remark}
The notion of matrix equivalent varies in different literature. For instance, \cite{10.1214/20-AOS1984} relaxed the restriction that $A_n,B_n$ are symmetric, and instead required that
\begin{align*}
 {\rm Tr}[\Theta_n(A_n-B_n)] \rightarrow 0
\end{align*} 
almost surely for any $\Theta_n$ satisfying
\begin{align}
 \limsup_{n}\Vert \Theta_n\Vert_{tr} < \infty,
 \label{23}
\end{align}
where the trace norm is defined by $\Vert A \Vert_{tr}={\rm Tr}((A^\top A)^{1/2})$ for any matrix $A$. This condition is stronger than ours since $\Theta_n=\frac{1}{n}C_n$ and $\Theta_n=v_nv_n^\top$ are all special cases of (\ref{23}). In our work, we are mainly interested in the traces and quadratic forms of symmetric random matrices, therefore, our requirement in the definition of matrix equivalent is enough for analysis. Yet, the rules of calculus in Proposition \ref{16} are commonly true for any reasonable definition of matrix equivalent.
\end{remark}

After introducing the concept of matrix equivalents, the Mar\v{c}enko-Pastur Theorem about the sample covariance matrix $\hat{\Sigma}$ can be modified as the form of deterministic equivalent for resolvent matrix $Q(z)=(\hat{\Sigma}-zI_p)^{-1}$.

\begin{theorem}[Theorem 2.6, \cite{Couillet_Liao_2022}]
Assume $X=Z\Sigma^{1/2}$, where $Z \in \mathbb{R}^{n \times p}$ has i.i.d. entries with zero mean, unit variance and finite eighth-order moment. Suppose $\Sigma \in \mathbb{R}^{p \times p}$ is a deterministic positive semi-definite matrix with bounded operator norm, then as $n,p \rightarrow \infty$ with $p/n \rightarrow \gamma \in (0, \infty)$, we have
\begin{align}
 Q(z) \asymp \bar{Q}(z):=-\frac{1}{z}(I_p+\tilde{m}_{\gamma,p}(z)\Sigma)^{-1},
 \label{44}
\end{align}
where $\tilde{m}_{\gamma,p}(z)$ is the unique solution of
\begin{align}
\tilde{m}_{\gamma,p}(z)=\Big( -z+{\rm Tr}\big[\Sigma(I_p+\tilde{m}_{\gamma,p}(z)\Sigma)^{-1}\big]\Big), \quad(z,\tilde{m}_{\gamma,p}(z))\in \mathbb{C}^+\times\mathbb{C}^+,
\label{26}
\end{align}
Define $\{\tilde{m}_{\gamma,p}(z),z \in (-\infty,0)\}$ by (\ref{26}) and the continuity of $\tilde{m}_{\gamma,p}(z)$, then (\ref{44}) is true for $z \in (-\infty,0)$. Besides, if ESD $F_\Sigma$ converges weakly to some limit distribution $H$, then as $p \rightarrow \infty$, we have $\tilde{m}_{\gamma,p}(z) \rightarrow \tilde{m}_{H,\gamma}(z)$, 
where $\tilde{m}_{H,\gamma}(z)$ is defined by the unique solution of 
\begin{align}
\tilde{m}_{H,\gamma}(z)=\Big( -z+\gamma\int \frac{t}{1+\tilde{m}_{H,\gamma(z)}t}dH(t) \Big)^{-1}, \quad(z,\tilde{m}_{\gamma,p}(z))\in \mathbb{C}^+\times\mathbb{C}^+.
\end{align}
\label{8}
\end{theorem}
Here we note that, if $\Sigma=I_p$, the results of Theorem \ref{8} reduce to
\begin{align}
 Q(z) \asymp \bar{Q}(z):=m_\gamma(z)I_p,
 \label{12}
\end{align}
where $m_{\gamma}(z)$ is the Stieltjes transform of limiting ESD. It is the unique solution of equation 
\begin{align}
z\gamma m_\gamma(z)-(1-\gamma-z)m_\gamma(z)+1=0, \quad(z,\tilde{m}_{\gamma,p}(z))\in \mathbb{C}^+\times\mathbb{C}^+,
\label{45}
\end{align}
which coincides with (\ref{6}). 

In the rest of this subsection, we provide a deterministic equivalent for the quadratic form of resolvent matrix. This result will be applied repeatedly to deal with the products of resolvent matrix, which frequently appear in the expressions of bias and variance terms. For proof convenience, we consider $z\in (-\infty,0)$ for the resolvent matrix $Q(z)$.

\begin{theorem}
Assume $X=Z\Sigma^{1/2}$, where $Z \in \mathbb{R}^{n \times p}$ has i.i.d. entries with zero mean, unit variance and finite 16th-order moment. Let $\lambda>0$ be a fixed real number, $Q(-\lambda)=(\frac{1}{n}X^\top X+\lambda I_p)^{-1}$ be the resolvent matrix, $A \in \mathbb{R}^{p \times p}$ is a symmetric deterministic matrix or symmetric random matrix independent of $Q(-\lambda)$, and having bounded operator norm. Suppose $\Sigma \in \mathbb{R}^{p \times p}$ is a deterministic positive semi-definite matrix with bounded operator norm, then as $n,p \rightarrow \infty$ with $p/n \rightarrow \gamma \in (0, \infty)$, we have
\begin{align}
 Q(-\lambda)AQ(-\lambda) \asymp \bar{Q}A\bar{Q}+\frac{1}{n}{\rm Tr}\big[\Sigma \bar{Q} A\bar{Q}\big] \cdot \Big\{\big[1+\frac{1}{n}{\rm Tr}(\Sigma \bar{Q})\big]^2-\frac{1}{n}{\rm Tr}(\Sigma \bar{Q})^2\Big\}^{-1}\cdot \bar{Q}\Sigma\bar{Q},
\end{align}
\label{25}
where $\bar{Q}:=\bar{Q}(-\lambda)$ is defined in Theorem \ref{8}.
\end{theorem}
If we take $\Sigma = I_p$, the deterministic equivalent of resolvent matrix becomes (\ref{12}), thus the following result is immediately derived from Theorem \ref{25}.

\begin{corollary}
Assume $X \in \mathbb{R}^{n \times p}$ has i.i.d. entries with zero mean, unit variance and finite 16th-order moment. Let $\lambda>0$ be a fixed real number, $Q(-\lambda)=(\frac{1}{n}X^\top X+\lambda I_p)^{-1}$ be the resolvent matrix, $A \in \mathbb{R}^{p \times p}$ is a symmetric deterministic matrix or a symmetric random matrix independent of $Q(-\lambda)$, and $\Vert A \Vert_{op}$ is bounded by a constant. Then as $n,p \rightarrow \infty$ with $p/n \rightarrow \gamma \in (0, \infty)$, we have
\begin{align}
 Q(-\lambda)AQ(-\lambda) \asymp m_\gamma^2(-\lambda)A+\frac{1}{n}{\rm Tr}A \cdot \frac{m_\gamma'(-\lambda)m_\gamma^2(-\lambda)}{(1+\gamma m_\gamma(-\lambda))^2}I_p.
\end{align}
Taking $A=I_p$, the equation above becomes
\begin{align}
 Q^2(-\lambda) \asymp  m_\gamma'(-\lambda)I_p.
 \label{13}
\end{align}
\label{9}
\end{corollary}

\subsection{Proof of Theorem \ref{25}}
For the convenience of notation, we write $Q:=Q(-\lambda)$, $\bar{Q}:=\bar{Q}(-\lambda)$ and $\tilde{m}:=\tilde{m}_{\gamma,p}(-\lambda)$. For a sequence of deterministic matrices $A_n \in \mathbb{R}^{n \times n}$, we write $A_n=o_{\Vert \cdot \Vert}(1)$ if $\Vert A_n \Vert_{op} \rightarrow 0$ as $n \rightarrow \infty$. The proof of this result can be divided into two steps: \\
i) Show that $QAQ$ is concentrated in mean, that is, $QAQ \asymp \mathbb{E}(QAQ)$.\\
ii) Find a deterministic matrix $R$ such that $\mathbb{E}(QAQ)=R+o_{\Vert \cdot \Vert}(1)$, therefore $\mathbb{E}(QAQ)\asymp R$. 

\subsection*{Step 1: Concentration of $QAQ$}
To prove that $QAQ\asymp \mathbb{E}(QAQ)$, it suffices to show that
\begin{align*}
\frac{1}{p}{\rm Tr}\big[C\big(QAQ-\mathbb{E}(QAQ)\big)\big]\rightarrow 0, a.s.
\end{align*}
for $C \in \mathbb{R}^{p\times p}$ having bounded operator norm, and
\begin{align*}
v^\top\big(QAQ-\mathbb{E}(QAQ)\big)v\rightarrow 0, a.s.
\end{align*}
for $v\in \mathbb{R}^p$ having bounded norm.

Similar to the proof in Lemma \ref{34}, we write
\begin{align*}
\frac{1}{p}{\rm Tr}\big[C\big(QAQ-\mathbb{E}(QAQ)\big)\big] =\frac{1}{p}\sum_{i=1}^n(\mathbb{E}_{\le i}-\mathbb{E}_{\le i-1}){\rm Tr}\big[C\big(QAQ-Q_{-i}AQ_{-i})\big)\big].
\end{align*}
By (\ref{29}), we have
\begin{align*}
\frac{1}{p}{\rm Tr}\big[C\big(QAQ-Q_{-i}AQ_{-i})\big)\big] = & - \frac{1}{pn}{\rm Tr}\big[CQ_{-i}A\frac{Q_{-i}x_ix_i^\top Q_{-i}}{1+\frac{1}{n}x_i^\top Q_{-i}x_i}+C\frac{Q_{-i}x_ix_i^\top Q_{-i}}{1+\frac{1}{n}x_i^\top Q_{-i}x_i}AQ_{-i}\big]\\
& + \frac{1}{pn^2}{\rm Tr}\big[C\frac{Q_{-i}x_ix_i^\top Q_{-i}AQ_{-i}x_ix_i^\top Q_{-i}}{(1+\frac{1}{n}x_i^\top Q_{-i}x_i)^2}\big].
\end{align*}
Let $Y_i=Y_{i,1}+Y_{i,2}$, where
\begin{align*}
Y_{i,1}&=-(\mathbb{E}_{\le i}-\mathbb{E}_{\le i-1})\frac{1}{pn}{\rm Tr}\Big[C\Big(Q_{-i}A\frac{Q_{-i}x_ix_i^\top Q_{-i}}{1+\frac{1}{n}x_i^\top Q_{-i}x_i}+\frac{Q_{-i}x_ix_i^\top Q_{-i}}{1+\frac{1}{n}x_i^\top Q_{-i}x_i}AQ_{-i}\Big)\Big],\\
Y_{i,2}&=(\mathbb{E}_{\le i}-\mathbb{E}_{\le i-1})\frac{1}{pn^2}{\rm Tr}\Big[C\frac{Q_{-i}x_ix_i^\top Q_{-i}AQ_{-i}x_ix_i^\top Q_{-i}}{(1+\frac{1}{n}x_i^\top Q_{-i}x_i)^2}\Big],
\end{align*}
Using the matrix inequality $A^\top B+B^\top A \preceq A^\top A+B^\top B$ and note that $\Vert Q_{-i}\Vert_{op}\le \lambda^{-1}$, we have
\begin{align*}
\Big|\frac{1}{pn}{\rm Tr} & \Big[C\Big(Q_{-i}A\frac{Q_{-i}x_ix_i^\top Q_{-i}}{1+\frac{1}{n}x_i^\top Q_{-i}x_i}+\frac{Q_{-i}x_ix_i^\top Q_{-i}}{1+\frac{1}{n}x_i^\top Q_{-i}x_i}AQ_{-i}\Big)\Big]\Big| \\
\le &\frac{\lambda^{-2}}{pn}\Vert C\Vert_{op}\Big\Vert A\frac{Q_{-i}x_ix_i^\top}{1+\frac{1}{n}x_i^\top Q_{-i}x_i}+\frac{x_ix_i^\top Q_{-i}}{1+\frac{1}{n}x_i^\top Q_{-i}x_i}A\Big\Vert_{tr} \\
\le &\frac{\lambda^{-2}}{pn}\Vert C\Vert_{op}\Big\Vert x_ix_i^\top+\frac{A Q_{-i}x_ix_i^\top Q_{-i}A}{1+\frac{1}{n}x_i^\top Q_{-i}x_i}\Big\Vert_{tr} \\
\le &\frac{\lambda^{-2}}{pn}\Vert C\Vert_{op}(\Vert x_i\Vert^2+ x_i^\top Q_{-i}A^2Q_{-i}x_i)\\
\le &\frac{\lambda^{-2}}{pn}\Vert C\Vert_{op}(1+\lambda^{-2}\Vert A\Vert_{op}^2)\Vert \Sigma\Vert_{op}\Vert z_i\Vert^2,
\end{align*}
where $x_i=\Sigma^{1/2}z_i$. For $k \ge 2$, if the entries of $Z$ has $k$-th moment, we have $\mathbb{E}\Vert z_i\Vert^{2k}=O(n^k)$, therefore,
\begin{equation*}
    \mathbb{E}_{\le i-1}\Big|\frac{1}{pn}{\rm Tr}\Big[C\Big(Q_{-i}A\frac{Q_{-i}x_ix_i^\top Q_{-i}}{1+\frac{1}{n}x_i^\top Q_{-i}x_i}+\frac{Q_{-i}x_ix_i^\top Q_{-i}}{1+\frac{1}{n}x_i^\top Q_{-i}x_i}AQ_{-i}\Big)\Big]\Big|^k= O(n^{-2k})\Vert z_i\Vert^{2k},
\end{equation*}
\begin{equation*}
    \mathbb{E}_{\le i}\Big|\frac{1}{pn}{\rm Tr}\Big[C\Big(Q_{-i}A\frac{Q_{-i}x_ix_i^\top Q_{-i}}{1+\frac{1}{n}x_i^\top Q_{-i}x_i}+\frac{Q_{-i}x_ix_i^\top Q_{-i}}{1+\frac{1}{n}x_i^\top Q_{-i}x_i}AQ_{-i}\Big)\Big]\Big|^k = O(n^{-2k})\mathbb{E}\Vert z_i\Vert^{2k}=O(n^{-k}),
\end{equation*}
and
\begin{align*}
|Y_{i,1}|^k= 2^{k-1}\big(O(n^{-2k})\Vert z_i\Vert^{2k}+O(n^{-k})\big),\quad \mathbb{E}|Y_{i,1}|^k=O(n^{-k}).
\end{align*}
For $Y_{i,2}$, note that
\begin{align*}
\Big|\frac{1}{pn^2}{\rm Tr}\Big[C\frac{Q_{-i}x_ix_i^\top Q_{-i}AQ_{-i}x_ix_i^\top Q_{-i}}{(1+\frac{1}{n}x_i^\top Q_{-i}x_i)^2}\Big]\Big|
\le &\frac{1}{pn^2}\Vert C\Vert_{op}\Vert A\Vert_{op}\Big\Vert \frac{Q_{-i}x_ix_i^\top Q_{-i}Q_{-i}x_ix_i^\top Q_{-i}}{(1+\frac{1}{n}x_i^\top Q_{-i}x_i)^2}\Big\Vert_{tr}\\
\le &\frac{\lambda^{-2}}{pn^2}\Vert C\Vert_{op}\Vert A\Vert_{op} \frac{(x_i^\top Q_{-i}x_i)^2}{(1+\frac{1}{n}x_i^\top Q_{-i}x_i)^2}\\
\le &\frac{\lambda^{-2}}{p}\Vert C\Vert_{op}\Vert A\Vert_{op}=O(n^{-1}),
\end{align*}
thus $|Y_{i,2}|=O(n^{-1})$. By Lemma \ref{33}, we have
\begin{align*}
\mathbb{E}\Big|\frac{1}{p}{\rm Tr} & \big[C\big(QAQ-\mathbb{E}(QAQ)\big)\big]\Big|^k= \mathbb{E}\Big|\sum_{i=1}^nY_i\Big|^k \le  C_k \Big(\mathbb{E}\big(\sum_{i=1}^n \mathbb{E}_{\le i-1}|Y_i|^2\big)^{k/2}+\sum_{i=1}^n\mathbb{E}|Y_i|^k\Big)\\
\le &C_k \Big(\mathbb{E}\big[\sum_{i=1}^n 2\big(\mathbb{E}_{\le i-1}|Y_{i,1}|^2+\mathbb{E}_{\le i-1}|Y_{i,2}|^2\big)\big]^{k/2}+2^{k-1}\sum_{i=1}^n\big(\mathbb{E}|Y_{i,1}|^k+\mathbb{E}|Y_{i,2}|^k\big)\Big)\\
\le &2^{k-1}C_k \Big(\mathbb{E}\big[\sum_{i=1}^n \mathbb{E}_{\le i-1}|Y_{i,1}|^2\big]^{k/2}+\mathbb{E}\big[\sum_{i=1}^n \mathbb{E}_{\le i-1}|Y_{i,2}|^2\big]^{k/2}+\sum_{i=1}^n\big(\mathbb{E}|Y_{i,1}|^k+\mathbb{E}|Y_{i,2}|^k\big)\Big)\\
=&O(n^{-k/2}).
\end{align*}
Taking $k>2$, by Markov's inequality and Borel-Cantelli lemma, we get
\begin{align*}
\frac{1}{p}{\rm Tr}\big[C\big(QAQ-\mathbb{E}(QAQ)\big)\big]\rightarrow 0, \quad a.s.
\end{align*}

To prove that $v^\top\big(QAQ-\mathbb{E}(QAQ)\big)v\rightarrow 0, a.s.$, we write
\begin{align*}
v^\top\big(QAQ-\mathbb{E}(QAQ)\big)v =\sum_{i=1}^n(\mathbb{E}_{\le i}-\mathbb{E}_{\le i-1})v^\top\big(QAQ-Q_{-i}AQ_{-i})\big)v.
\end{align*}
By (\ref{29}), we have
\begin{align*}
v^\top\big(QAQ-Q_{-i}AQ_{-i})\big)v=&-\frac{1}{n}v^\top Q_{-i}A\frac{Q_{-i}x_ix_i^\top Q_{-i}}{1+\frac{1}{n}x_i^\top Q_{-i}x_i}v-\frac{1}{n}v^\top\frac{Q_{-i}x_ix_i^\top Q_{-i}}{1+\frac{1}{n}x_i^\top Q_{-i}x_i}AQ_{-i}v\\
& + \frac{1}{n^2}v^\top\frac{Q_{-i}x_ix_i^\top Q_{-i}AQ_{-i}x_ix_i^\top Q_{-i}}{(1+\frac{1}{n}x_i^\top Q_{-i}x_i)^2}v.
\end{align*}
Let $\tilde{Y}_i=Y_{i,3}+Y_{i,4}$, where
\begin{align*}
Y_{i,3}=-(\mathbb{E}_{\le i}-\mathbb{E}_{\le i-1})\frac{1}{n}v^\top \Big(Q_{-i}A\frac{Q_{-i}x_ix_i^\top Q_{-i}}{1+\frac{1}{n}x_i^\top Q_{-i}x_i}+\frac{Q_{-i}x_ix_i^\top Q_{-i}}{1+\frac{1}{n}x_i^\top Q_{-i}x_i}AQ_{-i}\Big)v,
\end{align*}
\begin{align*}
Y_{i,4}=(\mathbb{E}_{\le i}-\mathbb{E}_{\le i-1})\frac{1}{n^2}v^\top \frac{Q_{-i}x_ix_i^\top Q_{-i}AQ_{-i}x_ix_i^\top Q_{-i}}{(1+\frac{1}{n}x_i^\top Q_{-i}x_i)^2}v,
\end{align*}
Using the matrix inequality $A^\top B+B^\top A \preceq A^\top A+B^\top B$ and note that $\Vert Q_{-i}\Vert_{op}\le \lambda^{-1}$, we have
\begin{align*}
\Big|\frac{1}{n}v^\top &\Big(Q_{-i}A\frac{Q_{-i}x_ix_i^\top Q_{-i}}{1+\frac{1}{n}x_i^\top Q_{-i}x_i}+\frac{Q_{-i}x_ix_i^\top Q_{-i}}{1+\frac{1}{n}x_i^\top Q_{-i}x_i}AQ_{-i}\Big)v\Big| \\
\le &\frac{\lambda^{-2}}{n}\Big|v^\top \Big( A\frac{Q_{-i}x_ix_i^\top}{1+\frac{1}{n}x_i^\top Q_{-i}x_i}+\frac{x_ix_i^\top Q_{-i}}{1+\frac{1}{n}x_i^\top Q_{-i}x_i}A\Big)v\Big| \\
\le &\frac{\lambda^{-2}}{n}\Big| v^\top \Big(x_ix_i^\top+\frac{A Q_{-i}x_ix_i^\top Q_{-i}A}{1+\frac{1}{n}x_i^\top Q_{-i}x_i}\Big)v\Big| \\
\le &\frac{\lambda^{-2}}{n}(1+\lambda^{-2}\Vert A\Vert_{op}^2)\Vert \Sigma\Vert_{op}|v^\top z_i|^2.
\end{align*}
Let $\mathbb{E}z_{11}^4=m_4,\mathbb{E}z_{11}^{2k}=m_{2k}$, by Lemma \ref{32}, we have
\begin{align*}
\mathbb{E}\Big|z_i^\top vv^\top z_i-\Vert v\Vert^2\Big|^k\le C_k\Big[(m_4\Vert v\Vert^2)^{k/2}+m_{2k}(\Vert v\Vert^2)^{k/2}\Big]=O(1),
\end{align*}
\begin{align*}
\mathbb{E}|v^\top z_i|^{2k}\le2^k\Big(\mathbb{E}\Big|z_i^\top vv^\top z_i-\Vert v\Vert^2\Big|^k+\Vert v\Vert^{2k}\Big)=O(1),
\end{align*}
thus for $k \ge 2$, if the entries of $Z$ has $2k$-th moment, we have
\begin{align*}
\mathbb{E}_{\le i-1}\Big|\frac{1}{n}v^\top \Big(Q_{-i}A\frac{Q_{-i}x_ix_i^\top Q_{-i}}{1+\frac{1}{n}x_i^\top Q_{-i}x_i}+\frac{Q_{-i}x_ix_i^\top Q_{-i}}{1+\frac{1}{n}x_i^\top Q_{-i}x_i}AQ_{-i}\Big)v\Big|^k= O(n^{-k})|v^\top z_i|^{2k},
\end{align*}
\begin{align*}
\mathbb{E}_{\le i}\Big|\frac{1}{n}v^\top \Big(Q_{-i}A\frac{Q_{-i}x_ix_i^\top Q_{-i}}{1+\frac{1}{n}x_i^\top Q_{-i}x_i}+\frac{Q_{-i}x_ix_i^\top Q_{-i}}{1+\frac{1}{n}x_i^\top Q_{-i}x_i}AQ_{-i}\Big)v\Big|^k= O(n^{-k})\mathbb{E}|v^\top z_i|^{2k}=O(n^{-k}),
\end{align*}
and
\begin{align*}
|Y_{i,3}|^k= 2^{k-1}\big(O(n^{-k})|v^\top z_i|^{2k}+O(n^{-k})\big),\quad \mathbb{E}|Y_{i,3}|^k=O(n^{-k}).
\end{align*}
For $Y_{i,4}$, note that
\begin{align*}
\Big|\frac{1}{n^2}v^\top\frac{Q_{-i}x_ix_i^\top Q_{-i}AQ_{-i}x_ix_i^\top Q_{-i}}{(1+\frac{1}{n}x_i^\top Q_{-i}x_i)^2}v\Big|
\le &\frac{\lambda^{-2}}{n^2}\Vert A\Vert_{op}\Big|\frac{x_i^\top Q_{-i}^2x_i}{(1+\frac{1}{n}x_i^\top Q_{-i}x_i)^2}v^\top x_ix_i^\top v\Big|\\
\le &\frac{\lambda^{-3}}{n}\Vert A\Vert_{op}\Vert \Sigma\Vert_{op}|v^\top z_i|^2,
\end{align*}
thus
\begin{align*}
\mathbb{E}_{\le i-1}\Big|\frac{1}{n^2}v^\top\frac{Q_{-i}x_ix_i^\top Q_{-i}AQ_{-i}x_ix_i^\top Q_{-i}}{(1+\frac{1}{n}x_i^\top Q_{-i}x_i)^2}v\Big|^k= O(n^{-k})|v^\top z_i|^{2k},
\end{align*}
\begin{align*}
\mathbb{E}_{\le i}\Big|\frac{1}{n^2}v^\top\frac{Q_{-i}x_ix_i^\top Q_{-i}AQ_{-i}x_ix_i^\top Q_{-i}}{(1+\frac{1}{n}x_i^\top Q_{-i}x_i)^2}v\Big|^k= O(n^{-k}),
\end{align*}
and
\begin{align*}
|Y_{i,4}|^k= 2^{k-1}\big(O(n^{-k})|v^\top z_i|^{2k}+O(n^{-k})\big),\quad \mathbb{E}|Y_{i,4}|^k=O(n^{-k}).
\end{align*}
By Lemma \ref{33}, we have
\begin{align*}
\mathbb{E}\Big|v^\top &\big(QAQ-\mathbb{E}(QAQ)\big)v\Big|^k 
= \mathbb{E}\Big|\sum_{i=1}^n\tilde{Y}_i\Big|^k \le  C_k \Big(\mathbb{E}\big(\sum_{i=1}^n \mathbb{E}_{\le i-1}|\tilde{Y}_i|^2\big)^{k/2}+\sum_{i=1}^n\mathbb{E}|\tilde{Y}_i|^k\Big)\\
\le &2^{k-1}C_k \Big(\mathbb{E}\big[\sum_{i=1}^n \mathbb{E}_{\le i-1}|Y_{i,3}|^2\big]^{k/2}+\mathbb{E}\big[\sum_{i=1}^n \mathbb{E}_{\le i-1}|Y_{i,4}|^2\big]^{k/2}+\sum_{i=1}^n\big(\mathbb{E}|Y_{i,3}|^k+\mathbb{E}|Y_{i,4}|^k\big)\Big)\\
=&O(n^{-k/2}).
\end{align*}
Taking $k>2$, by Markov's inequality and Borel-Cantelli lemma, we get 
\begin{align*}
v^\top\big(QAQ-\mathbb{E}(QAQ)\big)v\rightarrow 0, \quad a.s.
\end{align*}

\subsection*{Step 2: Finding Deterministic Matrix}
Note that
\begin{align}
\mathbb{E}[QAQ]=& \mathbb{E}[QA\bar{Q}]+\mathbb{E}[QA(Q-\bar{Q})] \nonumber\\
= & \bar{Q}A\bar{Q}+\mathbb{E}[QAQ(I_p-Q^{-1}\bar{Q})] \nonumber\\
= & \bar{Q}A\bar{Q}+\mathbb{E}\big[QAQ\big(I_p-(\frac{1}{n}X^\top X+\lambda I_p)\bar{Q}\big)\big] \nonumber\\
= & \bar{Q}A\bar{Q}+\mathbb{E}\big[QAQ(\bar{Q}^{-1}-\lambda I_p-\frac{1}{n}X^\top X)\bar{Q}\big] \nonumber\\
= & \bar{Q}A\bar{Q}+\mathbb{E}\big[QAQ(\lambda\tilde{m}\Sigma-\frac{1}{n}X^\top X)\bar{Q}\big] \nonumber\\
= & \bar{Q}A\bar{Q}+\mathbb{E}[QAQ]\frac{\Sigma\bar{Q}}{1+\frac{1}{n}{\rm Tr}(\Sigma\bar{Q})}-\mathbb{E}[QAQ\cdot\frac{1}{n}X^\top X]\bar{Q},
\label{21}
\end{align}
the last equality holds since $\frac{1}{n}{\rm Tr}(\Sigma\bar{Q})=-(\frac{1}{-\lambda\tilde{m}}+1)$ from (\ref{26}). Similarly, 
\begin{align}
\mathbb{E}[QAQ]=& \mathbb{E}[\bar{Q}AQ]+\mathbb{E}[(Q-\bar{Q})AQ] \nonumber\\
= & \bar{Q}A\bar{Q}+\frac{\Sigma\bar{Q}}{1+\frac{1}{n}{\rm Tr}(\Sigma\bar{Q})}\mathbb{E}[QAQ]-\bar{Q}\mathbb{E}[\frac{1}{n}X^\top X\cdot QAQ].
\label{37}
\end{align}

Let $Q_{-i}=(\frac{1}{n}\sum_{j\not= i}x_j x_j^\top-zI_p)^{-1}$ be the resolvent matrix without sample $x_i$. By Sherman-Morrison formula, we have
\begin{align*}
Q=Q_{-i}-\frac{\frac{1}{n}Q_{-i}x_ix_i^\top Q_{-i}}{1+\frac{1}{n}x_i^\top Q_{-i}x_i}
\end{align*}
and
\begin{align}
Qx_i=\frac{Q_{-i}x_i}{1+\frac{1}{n}x_i^\top Q_{-i}x_i},
\label{30}
\end{align}
thus
\begin{align}
\mathbb{E}[QAQ\cdot\frac{1}{n}X^\top X]\bar{Q}=\frac{1}{n}\sum_{i=1}^n\mathbb{E}\Big[\frac{QAQ_{-i}x_ix_i^\top}{1+\frac{1}{n}x_i^\top Q_{-i}x_i}\Big]\bar{Q}.
\end{align}
Let $\alpha=\frac{1}{n}{\rm Tr}[\mathbb{E}(\Sigma Q_{-i})]$, which is a constant for all $i$. Then we have
\begin{align}
\frac{QAQ_{-i}x_ix_i^\top}{1+\frac{1}{n}x_i^\top Q_{-i}x_i}=&\frac{1}{1+\alpha}\Big[Q_{-i}x_ix_i^\top-\frac{Q_{-i}x_ix_i^\top(\frac{1}{n}x_i^\top Q_{-i}x_i-\alpha)}{1+\frac{1}{n}x_i^\top Q_{-i}x_i}\Big] \nonumber\\
=&\frac{1}{1+\alpha}\Big[Q_{-i}x_ix_i^\top-Qx_ix_i^\top(\frac{1}{n}x_i^\top Q_{-i}x_i-\alpha)\Big].
\label{27}
\end{align}
Let $d_i=\frac{1}{n}x_i^\top Q_{-i}x_i-\alpha$, $D=\text{diag}\{d_i\}_{i=1}^n$, by (\ref{27}), we obtain
\begin{align}
\mathbb{E}[QAQ\cdot\frac{1}{n}X^\top X]\bar{Q}= &\frac{1}{(1+\alpha)n}\sum_{i=1}^n\mathbb{E}[QAQ_{-i}x_ix_i^\top-QAQx_id_ix_i^\top ]\bar{Q} \nonumber\\
=&\frac{1}{(1+\alpha)n}\sum_{i=1}^n\mathbb{E}[QAQ_{-i}x_ix_i^\top]\bar{Q}-\frac{1}{(1+\alpha)n}\mathbb{E}[QAQX^\top DX]\bar{Q},
\label{28}
\end{align}
and similarly,
\begin{align}
\bar{Q}\mathbb{E}[\frac{1}{n}X^\top X\cdot QAQ]= \frac{1}{(1+\alpha)n}\sum_{i=1}^n\bar{Q}\mathbb{E}[x_ix_i^\top Q_{-i}AQ]-\frac{1}{(1+\alpha)n}\bar{Q}\mathbb{E}[X^\top DXQAQ],
\label{35}
\end{align}

We wish to control the operator norm of $\Delta:=\frac{1}{n}\mathbb{E}[QAQX^\top DX]\bar{Q}$ plus its transport. Using the matrix inequality $A^\top B+B^\top A \preceq A^\top A+B^\top B$, we have
\begin{align*}
\Big\Vert\Delta+\Delta^\top\Big\Vert &= \frac{1}{n}\Big\Vert\mathbb{E}[n^{\epsilon/2}QAQX^\top D\cdot n^{-\epsilon/2}X\bar{Q}+n^{-\epsilon/2}\bar{Q}X^\top\cdot n^{\epsilon/2}(QAQX^\top D)^\top]\Big\Vert \\
&\le \frac{1}{n}\Big\Vert\mathbb{E}[n^{\epsilon}QAQX^\top D^2XQAQ+n^{-\epsilon}\bar{Q}X^\top X\bar{Q}]\Big\Vert \\
&\le n^{-1+\epsilon}\Vert\mathbb{E}[QAQX^\top D^2XQAQ]\Vert + n^{-1-\epsilon}\Vert\mathbb{E}[\bar{Q}X^\top X\bar{Q}]\Vert \\
&\le n^{-1+\epsilon}\Vert\mathbb{E}[\Vert D \Vert^2 QAQX^\top XQAQ]\Vert + n^{-1-\epsilon}\Vert\mathbb{E}[\bar{Q}X^\top X\bar{Q}]\Vert \\
& \le C_1n^\epsilon \mathbb{E}\Vert D \Vert^2 +C_2 n^{-\epsilon} \Vert \Sigma \Vert,
\end{align*}
the last inequality holds since $\Vert A \Vert$ is bounded by a constant, $\Vert\frac{1}{n}QX^\top X\Vert=\Vert I_p-\lambda Q\Vert$ is bounded by 2, $\Vert Q \Vert \le \lambda^{-1}$, and$\Vert \bar{Q} \Vert \le \Vert EQ\Vert +\Vert \bar{Q}-EQ\Vert \le  \lambda^{-1}+o(1)$.
 
To control $n^\epsilon \mathbb{E}\Vert D \Vert^2$, we notice that
\begin{align*}
n^\epsilon \mathbb{E}\Vert D \Vert^2 =& n^\epsilon\mathbb{E} \max_{i}d_i^2= n^\epsilon\int_{0}^{\infty} \textbf{P}(\max_{i}d_i^2>t)dt\\
\le & n^\epsilon \int_{0}^{n^{-\theta-\epsilon}} \textbf{P}(\max_{i}d_i^2>t)dt+n^{\epsilon}\int_{n^{-\theta-\epsilon}}^\infty n\textbf{P}(d_1^2>t)dt\\
\le & n^{-\theta} +n^{1+\epsilon}\int_{n^{-\theta-\epsilon}}^\infty \textbf{P}(d_1^2>t)dt.
\end{align*}
Since
\begin{align*}
\textbf{P}(d_1^2>t)=& \textbf{P}\Big(\Big|\frac{1}{n}z_1^\top \Sigma^{1/2}Q_{-1}\Sigma^{1/2}z_1-\frac{1}{n}{\rm Tr}[\mathbb{E}(\Sigma Q_{-1})]\Big|^2>t\Big)\\
\le & \frac{1}{t^4}\mathbb{E}\Big|\frac{1}{n}z_1^\top \Sigma^{1/2}Q_{-1}\Sigma^{1/2}z_1-\frac{1}{n}{\rm Tr}[\mathbb{E}(\Sigma Q_{-1})]\Big|^8 \\
\le & \frac{2^7}{t^4n^8}\Big\{\mathbb{E}\Big|z_1^\top \Sigma^{1/2}Q_{-1}\Sigma^{1/2}z_1-{\rm Tr}[\Sigma Q_{-1}]\Big|^4+\mathbb{E}\Big|{\rm Tr}[\Sigma (Q_{-1}-\mathbb{E}Q_{-1})]\Big|^8\Big\},
\end{align*}
taking $k=8$ in Lemma \ref{32}, we have
\begin{align*}
\mathbb{E}\Big|z_1^\top \Sigma^{1/2}Q_{-1}\Sigma^{1/2}z_1-{\rm Tr}[\Sigma Q_{-1}]\Big|^8=\mathbb{E}_{-1}\mathbb{E}_1\Big|z_1^\top \Sigma^{1/2}Q_{-1}\Sigma^{1/2}z_1-{\rm Tr}[\Sigma^{1/2}Q_{-1}\Sigma^{1/2}]\Big|^8 \le Cp^4
\end{align*}
for some $C>0$, and taking $k=8$ in Lemma \ref{34}, we have
\begin{align*}
 \mathbb{E}\Big|\frac{1}{p}{\rm Tr}[\Sigma (Q_{-1}-\mathbb{E}Q_{-1})]\Big|^8=O(n^{-4}),
\end{align*}
thus $\textbf{P}(d_1^2>t) \le Ct^{-4}n^{-4}$ for some $C>0$, and
\begin{align*}
n^\epsilon \mathbb{E}\Vert D \Vert^2 
\le  n^{-\theta} +Cn^{4\epsilon+3\theta-3}.
\end{align*}
Choose $\epsilon=\theta=3/8$, we have $n^\epsilon \mathbb{E}\Vert D \Vert^2=O(n^{-3/8})$, and
\begin{align*}
\Big\Vert \frac{1}{n}\mathbb{E}[QAQX^\top DX]\bar{Q} \Big\Vert = O(n^{-3/8}),
\end{align*}
therefore, (\ref{28}) and (\ref{35}) implies
\begin{align*}
\mathbb{E}[QA&Q\cdot\frac{1}{n}X^\top X]\bar{Q}+\bar{Q}\mathbb{E}[\frac{1}{n}X^\top X\cdot QAQ]\\
=&\frac{1}{(1+\alpha)n}\sum_{i=1}^n\Big\{\mathbb{E}[QAQ_{-i}x_ix_i^\top]\bar{Q}+\bar{Q}\mathbb{E}[x_ix_i^\top Q_{-i}AQ]\Big\}+o_{\Vert\cdot\Vert}(1).
\end{align*}

Applying Sherman-Morrison formula (\ref{29}) and using the same technique as (\ref{27}), we have
\begin{align*}
\frac{1}{(1+\alpha)n}&\sum_{i=1}^n\mathbb{E}[QAQ_{-i}x_ix_i^\top]\bar{Q}\\
= & \frac{1}{(1+\alpha)n} \Big\{\sum_{i=1}^n \mathbb{E}[Q_{-i}AQ_{-i}x_ix_i^T\bar{Q}]-\sum_{i=1}^n\frac{Q_{-i}x_in^{-1}x_i^\top Q_{-i}AQ_{-i}x_ix_i^\top \bar{Q}}{1+\frac{1}{n}x_i^\top Q_{-i} x_i}\Big\} \\ = & \Delta_1 -\Delta_2,
\end{align*}
and
\begin{align*}
\mathbb{E}[QAQ\cdot\frac{1}{n}X^\top X]\bar{Q}+\bar{Q}\mathbb{E}[\frac{1}{n}X^\top X\cdot QAQ]=(\Delta_1 +\Delta_1^\top)-(\Delta_2 +\Delta_2^\top),
\end{align*}
where
\begin{align*}
\Delta_1=\frac{1}{(1+\alpha)n}\sum_{i=1}^n \mathbb{E}[Q_{-i}AQ_{-i}x_ix_i^T]\bar{Q},
\end{align*}
\begin{align*}
\Delta_2=&\frac{1}{(1+\alpha)^2n}\sum_{i=1}^n \mathbb{E}[Q_{-i}x_in^{-1}x_i^\top Q_{-i}AQ_{-i}x_ix_i^\top] \bar{Q}\\
-&\frac{1}{(1+\alpha)^2n}\sum_{i=1}^n \mathbb{E}[Q_{-i}x_in^{-1}x_i^\top Q_{-i}AQx_id_ix_i^\top]\bar{Q},
\end{align*}
since $Q_{-i}$ and $x_i$ are independent, we have
\begin{align*}
\Delta_1=&\frac{1}{(1+\alpha)}\sum_{i=1}^n \mathbb{E}_{-i}[Q_{-i}AQ_{-i}]\mathbb{E}_{i}[n^{-1}x_ix_i^T]\bar{Q} \\
= & \frac{1}{(1+\alpha)}\sum_{i=1}^n \mathbb{E}_{-i}[Q_{-i}AQ_{-i}]\Sigma\bar{Q}\\
=& \frac{1}{1+\alpha} \mathbb{E}[QAQ]\Sigma\bar{Q}+o_{\Vert\cdot\Vert}(1).
\end{align*}
To analyze $\Delta_2+\Delta_2^\top$, we first control the norm of $\tilde{\Delta}:=\frac{1}{n}\sum_{i=1}^n \mathbb{E}[Q_{-i}x_in^{-1}x_i^\top Q_{-i}AQx_id_ix_i^\top]\bar{Q}$ plus its transport. Substituting $Q_{-i}x_i$ by (\ref{30}), we have
\begin{align*}
Q_{-i}x_in^{-1}x_i^\top Q_{-i}AQx_id_ix_i^\top=Qx_in^{-1}x_i^\top QAQx_id_ix_i^\top\times(1+\frac{1}{n}x_i^\top Q_{-i}x_i)^2.
\end{align*}
Let $\tilde{d}_i=d_i\cdot (1+\frac{1}{n}x_i^\top Q_{-i}x_i)^2=d_i(1+\alpha+d_i)^2$, $\tilde{D}=\text{diag}\{|\tilde{d}_i|\}_{i=1}^n$, using similar matrix inequality techniques, we have
\begin{align*}
\Big\Vert \tilde{\Delta}+\tilde{\Delta}^\top \Big\Vert =& \Big\Vert \frac{1}{n}\sum_{i=1}^n \mathbb{E}[Qx_in^{-1}x_i^\top QAQx_i\tilde{d}_ix_i^\top]\bar{Q} +\bar{Q}[Qx_in^{-1}x_i^\top QAQx_i\tilde{d}_ix_i^\top]^\top\Big\Vert \\
\le & \Big\Vert \frac{1}{n} \mathbb{E}[n^{\epsilon/2}\frac{1}{n}QX^\top X  QA QX^\top\tilde{D}\cdot n^{-\epsilon/2}X\bar{Q}]+\mathbb{E}[n^{\epsilon/2}\frac{1}{n}QX^\top X  QA QX^\top\tilde{D}\cdot n^{-\epsilon/2}X\bar{Q}]^\top \Big\Vert \\
\le & C_1n^\epsilon \mathbb{E}\Vert \tilde{D} \Vert^2 +C_2 n^{-\epsilon} \Vert \Sigma \Vert,
\end{align*}
and
\begin{align*}
n^\epsilon \mathbb{E}\Vert \tilde{D} \Vert^2 
\le n^{-\theta} +n^{1+\epsilon}\int_{n^{-\theta-\epsilon}}^\infty \textbf{P}(\tilde{d}_1^2>t)dt.
\end{align*}
Since 
\begin{align*}
\textbf{P}(\tilde{d}_1^2>t) \le &\textbf{P}\big(8d_1^2[(1+|\alpha|)^4+d_1^4]>t\big) \\
\le & \textbf{P}\big(8d_1^2[(1+|\alpha|)^4+d_1^4]>t,d_1^2<(1+|\alpha|)^2\big) +\textbf{P}\big(8d_1^2[(1+|\alpha|)^4+d_1^4]>t,d_1^2\ge(1+|\alpha|)^2\big)\\
\le & \textbf{P}\big(16d_1^2(1+|\alpha|)^4>t\big)+\textbf{P}\big(16d_1^6>t\big) \\
\le & C_1t^{-4}n^{-4}+C_2t^{-4/3}n^{-4},
\end{align*}
we have
\begin{align*}
n^\epsilon \mathbb{E}\Vert D \Vert^2 
\le  n^{-\theta} +C_1n^{4\epsilon+3\theta-3}+C_2n^{\frac{4}{3}\epsilon+\frac{1}{3}\theta-3}.
\end{align*}
Choose $\theta=\epsilon=3/8$, we have $n^\epsilon \mathbb{E}\Vert D \Vert^2=O(n^{-3/8})$, and
\begin{align*}
\Big\Vert \frac{1}{n}\mathbb{E}[QAQX^\top DX]\bar{Q} \Big\Vert = O(n^{-3/8}),
\end{align*}
therefore, 
\begin{align*}
\Delta_2+\Delta_2^\top=&\frac{1}{(1+\alpha)^2n}\sum_{i=1}^n \Big\{\mathbb{E}[Q_{-i}x_in^{-1}x_i^\top Q_{-i}AQ_{-i}x_ix_i^\top] \bar{Q}\\
&+\bar{Q}\mathbb{E}[Q_{-i}x_in^{-1}x_i^\top Q_{-i}AQ_{-i}x_ix_i^\top]^\top\Big\}+o_{\Vert\cdot\Vert}(1).
\end{align*}
For the first term of $\Delta_2+\Delta_2^\top$, note that
\begin{align*}
\mathbb{E}[Q_{-i}&x_in^{-1}x_i^\top Q_{-i}AQ_{-i}x_ix_i^\top]\\
=& \mathbb{E}_{-i}\big[Q_{-i}\mathbb{E}_{i}[x_in^{-1}x_i^\top Q_{-i}AQ_{-i}x_ix_i^\top]\big] \\
= & \mathbb{E}_{-i}\big[Q_{-i}\mathbb{E}_{i}[\frac{1}{n}\Sigma^{1/2}z_iz_i^\top \Sigma^{1/2}Q_{-i}AQ_{-i}\Sigma^{1/2}z_iz_i^\top\Sigma^{1/2}]\big] \\
= & \mathbb{E}_{-i}\big[Q_{-i}\Sigma^{1/2}[\frac{2}{n} \Sigma^{1/2}Q_{-i}AQ_{-i}\Sigma^{1/2}+\frac{1}{n}{\rm Tr}(\Sigma Q_{-i}AQ_{-i})I_p]\Sigma^{1/2}\big] \\
= & \mathbb{E}_{-i}\big[Q_{-i}\Sigma\cdot\frac{1}{n}{\rm Tr}(\Sigma Q_{-i}AQ_{-i})\big]+o_{\Vert\cdot\Vert}(1),
\end{align*}
thus
\begin{align*}
\Delta_2+\Delta_2^\top=& \frac{1}{(1+\alpha)^2}\Big\{\mathbb{E}\big[Q\Sigma\cdot\frac{1}{n}{\rm Tr}(\Sigma QAQ)\big]\bar{Q}+\bar{Q}\big[Q\Sigma\cdot\frac{1}{n}{\rm Tr}(\Sigma QAQ)\big]^\top\Big\}+o_{\Vert\cdot\Vert}(1) \\
= & \frac{2}{(1+\alpha)^2}\cdot\frac{1}{n}{\rm Tr}(\mathbb{E}[\Sigma QAQ])\bar{Q}\Sigma\bar{Q}+o_{\Vert\cdot\Vert}(1).
\end{align*}
Combining the results above and noticing that 
\begin{align*}
\alpha=\frac{1}{n}{\rm Tr}[\mathbb{E}(\Sigma Q_{-i})]=\frac{1}{n}{\rm Tr}[\mathbb{E}(\Sigma Q)]+o(1)=\frac{1}{n}{\rm Tr}[\Sigma \bar{Q}]+o(1),
\end{align*}
 we have
\begin{align*}
\mathbb{E}[Q & AQ\cdot\frac{1}{n}X^\top X]\bar{Q}+\bar{Q}\mathbb{E}[\frac{1}{n}X^\top X\cdot QAQ] \\= 
& \frac{\mathbb{E}[QAQ]\Sigma\bar{Q}}{1+\frac{1}{n}{\rm Tr}[\Sigma \bar{Q}]}+\frac{\Sigma\bar{Q}\mathbb{E}[QAQ]}{1+\frac{1}{n}{\rm Tr}[\Sigma \bar{Q}]}
-\frac{2}{(1+\frac{1}{n}{\rm Tr}[\Sigma \bar{Q}])^2}\cdot\frac{1}{n}{\rm Tr}(\mathbb{E}[\Sigma QAQ])\bar{Q}\Sigma\bar{Q}+o_{\Vert\cdot\Vert}(1).
\end{align*}
Therefore, the sum of equations (\ref{21}) and (\ref{37}) becomes
\begin{align*}
\mathbb{E}[QAQ]=\bar{Q}A\bar{Q}+\frac{1}{(1+\frac{1}{n}{\rm Tr}[\Sigma \bar{Q}])^2}\cdot\frac{1}{n}{\rm Tr}(\mathbb{E}[\Sigma QAQ])\bar{Q}\Sigma\bar{Q}+o_{\Vert\cdot\Vert}(1).
\end{align*}
Taking normalized trace in both sides of this equation, we get
\begin{align*}
\frac{1}{n}{\rm Tr}(\mathbb{E}[\Sigma QAQ])=\frac{1}{n}{\rm Tr}(\Sigma\bar{Q}A\bar{Q})+\frac{\frac{1}{n}{\rm Tr}[(\Sigma\bar{Q})^2]}{(1+\frac{1}{n}{\rm Tr}[\Sigma \bar{Q}])^2}\cdot\frac{1}{n}{\rm Tr}(\mathbb{E}[\Sigma QAQ])+o(1).
\end{align*}
and obtain the solution
\begin{align*}
\frac{1}{n}{\rm Tr}(\mathbb{E}[\Sigma QAQ])=\Big\{1-\frac{\frac{1}{n}{\rm Tr}[(\Sigma\bar{Q})^2]}{(1+\frac{1}{n}{\rm Tr}[\Sigma\bar{Q}])^2} \Big\}^{-1}\cdot \frac{1}{n}{\rm Tr}(\Sigma\bar{Q}A\bar{Q})+o(1).
\end{align*}
Thus
\begin{align*}
\mathbb{E}[QAQ]=\bar{Q}A\bar{Q}+\frac{\frac{1}{n}{\rm Tr}[\Sigma \bar{Q}A\bar{Q}]}{(1+\frac{1}{n}{\rm Tr}[\Sigma \bar{Q}])^2-\frac{1}{n}{\rm Tr}[(\Sigma \bar{Q})^2]}\cdot\bar{Q}\Sigma\bar{Q}+o_{\Vert\cdot\Vert}(1).
\end{align*}

\section{Main Proofs}
\subsection{Proof of Lemma \ref{46}}
\label{58}
 For bias term, note that 
\begin{align*}
    \mathbb{E}(\hat{\beta}_T|X)=  \mathbb{E}(\hat{\beta}^{\rm ridge}_{T}|X) + A_T\mathbb{E}(\hat{\beta}^{\rm ridge}_{T-1}|X)+  \cdots + A_TA_{T-1}\cdots A_2\mathbb{E}(\hat{\beta}^{\rm ridge}_{1}|X).
\end{align*}
Since $\mathbb{E}(\hat{\beta}^{\rm ridge}_{t}|X)=  (\hat{\Sigma}_t+\lambda_tI_p)^{-1}\hat{\Sigma}_t\beta=(I_p-A_t)\beta$, we have
\begin{align*}
    \mathbb{E}(\hat{\beta}_T|X)=& (I_p-A_T)\beta  + A_T(I_p-A_{T-1})\beta+  \cdots + A_TA_{T-1}\cdots A_2(I_p-A_1)\beta\\
    =& \beta-A_TA_{T-1}\cdots A_1\beta.
\end{align*}
by (\ref{3}), $
B_X(\hat{\beta}; \beta,\Sigma_0)=\beta^\top A_1A_2\cdots A_T\Sigma_0 A_T\cdots A_2A_1\beta$.

For variance term, note that ridge estimators using different datasets are independent, we have
\begin{align*}
    \text{Cov}(\hat{\beta}_T|X)=  \sum_{t=1}^TA_T\cdots A_{t+1}\text{Cov}(\hat{\beta}^{\rm ridge}_{1}|X)A_{t+1}\cdots A_T.
\end{align*}
Since 
\begin{align*}
    \text{Cov}(\hat{\beta}^{\rm ridge}_{t}|X)&= \text{Cov}[(\hat{\Sigma}_t+\lambda_t I_p)^{-1}(\frac{1}{n_t}X_t^\top \epsilon_t)|X]\\
    &=\frac{\sigma^2}{n_t}(\hat{\Sigma}_t+\lambda_t I_p)^{-1}\hat{\Sigma}_t(\hat{\Sigma}_t+\lambda_t I_p)^{-1}\\
    &=\frac{\sigma^2}{n_t}[(\hat{\Sigma}_t+\lambda_t I_p)^{-1}-\lambda_t(\hat{\Sigma}_t+\lambda_t I_p)^{-2}]\\
    &=\frac{\sigma^2}{\lambda_tn_t}(A_t-A_t^2),
\end{align*}
we have 
\begin{align*}
    \text{Cov}(\hat{\beta}_T|X)=  \sigma^2\sum_{t=1}^{T} \frac{1}{\lambda_tn_t}A_T\cdots A_{t+1}(A_t-A_t^2)A_{t+1}\cdots A_T.
\end{align*}
Thus by (\ref{4}),
\begin{align*}
  V_X(\hat{\beta}; \beta,\Sigma_0)=\sigma^2\sum_{t=1}^{T} \frac{1}{\lambda_tn_t}{\rm Tr}\big[A_T\cdots A_{t+1}(A_t-A_t^2)A_{t+1}\cdots A_T\Sigma_0\big].
\end{align*}

\subsection{Proof of Theorem \ref{5}}
\label{24}
\subsubsection{Bias Term}
Let $Q_t=A_t/\lambda_t$, then $Q_t$ is a resolvent matrix of sample covariance matrix $\hat{\Sigma}_t=\frac{1}{n_t}X_t^\top X_t$. Note that 
\begin{align*}
B_X(\hat{\beta}; \beta)=\big(\prod_{t=1}^T \lambda_t^2\big) \beta^\top Q_1Q_2\cdots Q_T Q_T\cdots Q_2Q_1\beta,
\end{align*}
Let $B_t=Q_t\cdots Q_T Q_T\cdots Q_t$. For $t<T$, since $B_t=Q_tB_{t+1}Q_t$ and $B_{t+1}$ is independent of $Q_t$, by Corollary \ref{9},
\begin{align}
 B_t \asymp m_{\gamma_t}^2(-\lambda_t) B_{t+1} + \frac{1}{n_t}{\rm Tr}B_{t+1}\cdot m_{\gamma_t}^2(-\lambda_t)\mu_tI_p, 
  \label{17}
\end{align}
where
\begin{align*}
 \mu_t=m_{\gamma_t}'(-\lambda_t)/[1+\gamma_tm_{\gamma_t}(-\lambda_t)]^2,
\end{align*}
by the third rule of Proposition \ref{16}, 
\begin{align}
 \frac{1}{p}{\rm Tr}B_t - m_{\gamma_t}^2(-\lambda_t) (1+ \gamma_t \mu_t) \cdot\frac{1}{p}{\rm Tr}B_{t+1}\rightarrow 0, \quad a.s.
  \label{18}
\end{align}
By differentiating (\ref{45}), we obtain
\begin{align*}
 m'_\gamma(z)=m_\gamma^2(z)\Big/\Big[1-\frac{\gamma m_\gamma^2(z)}{(1+\gamma m_\gamma(z))^2}\Big],
\end{align*}
therefore, 
\begin{align}
 m_{\gamma_t}^2(-\lambda_t)(1+\gamma_t\mu_t)=m_{\gamma_t}'(-\lambda_t).
 \label{19}
\end{align}
Applying (\ref{18}) iteratively and using (\ref{19}), we obtain
\begin{align}
 \frac{1}{p}{\rm Tr}B_t - \prod_{s=t}^{T} m_{\gamma_s}'(-\lambda_s)\rightarrow 0, \quad a.s., \quad 1 \le t \le T,
  \label{20}
\end{align}
substituting (\ref{20}) in (\ref{17}), we get
\begin{align*}
 B_t \asymp m_{\gamma_t}^2(-\lambda_t) B_{t+1} + \gamma_t m_{\gamma_t}^2(-\lambda_t)\mu_t \cdot\prod_{s=t+1}^{T} m_{\gamma_s}'(-\lambda_s)I_p,\quad 1\le t \le T-1. 
\end{align*}
To simplify our formula, we denote $w_t=\gamma_t \mu_t \cdot\prod_{s=t+1}^{T} m_{\gamma_s}'(-\lambda_s)$, then by iteration, 
\begin{align}
 B_1 \asymp & m_{\gamma_1}^2(-\lambda_1) (B_{2} + w_1 I_p)\nonumber\\
 \asymp & m_{\gamma_1}^2(-\lambda_1) \big[m_{\gamma_2}^2(-\lambda_2) (B_{3} + w_2 I_p) + w_1 I_p\big]\nonumber\\
  \asymp & \cdots \asymp \prod_{t=1}^{T-1}m_{\gamma_t}^2(-\lambda_t)B_T+\sum_{t=1}^{T-1} \big(\prod_{s=1}^t m_{\gamma_s}^2(-\lambda_s)\big)w_t \cdot I_p \nonumber\\
  \asymp & \Big\{m_{\gamma_T}'(-\lambda_T)\prod_{t=1}^{T-1}m_{\gamma_t}^2(-\lambda_t)+\sum_{t=1}^{T-1} \Big[\gamma_tu_t \prod_{s=1}^t m_{\gamma_s}^2(-\lambda_s)\prod_{s=t+1}^{T}m_{\gamma_s}'(-\lambda_s) \Big]\Big\} \cdot I_p \nonumber\\
  =&\Big\{m_{\gamma_T}'(-\lambda_T)(1+\gamma_{T-1}\mu_{T-1})\prod_{t=1}^{T-1}m_{\gamma_t}^2(-\lambda_t)+\sum_{t=1}^{T-2} \Big[\gamma_tu_t \prod_{s=1}^t m_{\gamma_s}^2(-\lambda_s)\prod_{s=t+1}^{T}m_{\gamma_s}'(-\lambda_s) \Big]\Big\} \cdot I_p \nonumber\\
  =&\Big\{\prod_{t=T-1}^{T}m_{\gamma_t}'(-\lambda_t)\prod_{t=1}^{T-2}m_{\gamma_t}^2(-\lambda_t)+\sum_{t=1}^{T-2} \Big[\gamma_tu_t \prod_{s=1}^t m_{\gamma_s}^2(-\lambda_s)\prod_{s=t+1}^{T}m_{\gamma_s}'(-\lambda_s) \Big]\Big\} \cdot I_p \nonumber\\
  =&\cdots=\prod_{t=1}^{T}m_{\gamma_t}'(-\lambda_t) \cdot I_p 
  \label{22}
\end{align}
By the definition of deterministic equivalent, we have 
\begin{align*}
 \beta^\top B_1 \beta \rightarrow r^2   \prod_{t=1}^{T}m_{\gamma_t}'(-\lambda_t), \quad a.s.
\end{align*}
thus
\begin{align*}
 B_X(\hat{\beta}_T; \beta) \rightarrow r^2 \prod_{t=1}^T\big(\lambda_t^2m_{\gamma_t}'(-\lambda_t)\big), \quad a.s.
\end{align*}

\subsubsection{Variance Term}
For $1 \le s < t \le T$, Let $C_{s,t}= Q_t\cdots Q_{s+1}Q_s Q_{s+1} \cdots Q_t, D_{s,t}= Q_t\cdots Q_{s+1}Q_s^2Q_{s+1} \cdots Q_t$, then
\begin{align*}
V_X(\hat{\beta}_T;\beta)=\sigma^2\sum_{t=1}^{T} \big(\prod_{s=t+1}^T\lambda_s^2\big)\cdot\frac{1}{n_t}{\rm Tr}\big[C_{t,T}-\lambda_t D_{t,T}\big].
\end{align*}
For $s<t$, since $C_{s,t}=Q_tC_{s,t-1}Q_t$, and $C_{s,t-1}$ is independent of $Q_t$ with bounded operator norm, by Corollary \ref{9},
\begin{align*}
  C_{s,t}\asymp m_{\gamma_t}^2(-\lambda_t)C_{s,t-1}+\frac{1}{n_t}{\rm Tr}C_{s,t-1} \cdot m_{\gamma_t}^2(-\lambda_t)\mu_tI_p,
\end{align*}
By rules of calculus for deterministic equivalents, we have
\begin{align}
  \frac{1}{p}{\rm Tr}C_{s,t}- m_{\gamma_t}^2(-\lambda_t)(1+\gamma_t\mu_t)\cdot\frac{1}{p}{\rm Tr}C_{s,t-1} \rightarrow 0, \quad a.s.
  \label{11}
\end{align}
Applying (\ref{11}) iteratively and notice (\ref{19}), we obtain
\begin{align*}
  \frac{1}{p}{\rm Tr}C_{t,T}- \prod_{s=t+1}^{T}m_{\gamma_s}'(-\lambda_s)\cdot\frac{1}{p}{\rm Tr}Q_t \rightarrow 0, \quad a.s.
\end{align*}
By (\ref{12}), we have
\begin{align*}
  \frac{1}{p}{\rm Tr}Q_t \rightarrow m_{\gamma_t}(-\lambda_t), \quad a.s.
\end{align*}
thus
\begin{align}
  \frac{1}{n_t}{\rm Tr}C_{t,T}\rightarrow \prod_{s=t+1}^{T}m_{\gamma_s}'(-\lambda_s)\cdot\gamma_tm_{\gamma_t}(-\lambda_t), \quad a.s.
    \label{14}
\end{align}
Using similar approach to $D_{s,t}=Q_tD_{s,t-1}Q_t$, we obtain
\begin{align*}
  \frac{1}{p}{\rm Tr}D_{t,T}- \prod_{s=t}^{T}m_{\gamma_s}'(-\lambda_s) \rightarrow 0, \quad a.s.
\end{align*}
thus
\begin{align}
  \frac{1}{n_t}{\rm Tr}D_{t,T}\rightarrow \prod_{s=t}^{T}m_{\gamma_s}'(-\lambda_s)\cdot\gamma_t, \quad a.s.
  \label{15}
\end{align}
Substituting (\ref{14}) and (\ref{15}) in expressions of $V_X$, we get
\begin{align*}
  V_X(\hat{\beta}_T;\beta) \rightarrow \sigma^2   \sum_{t=1}^{T}\gamma_{t}\upsilon_{t}\prod_{s=t+1}^T\big[\lambda_s^2m_{\gamma_s}'(-\lambda_s)\big], \quad a.s.
\end{align*}

\subsection{Proof of Theorem \ref{10}}
\label{43}
\subsubsection{Bias Term}
Let $B_t=Q_t\cdots Q_T \Sigma_0Q_T\cdots Q_t, 1\le t \le T$ and $B_{T+1}=\Sigma_0$. For $s\le t$, define $\bar{Q}_t=\lambda_t^{-1}(I_p+\tilde{m}_t\Sigma_t)^{-1}$, $H_{s,t}=\bar{Q}_s\cdots\bar{Q}_t\Sigma_t\bar{Q}_t\cdots \bar{Q}_s$, and
\begin{align*}
\tilde{\mu}_t&=\Big\{\big[1+\frac{1}{n_t}{\rm Tr}(\Sigma_t \bar{Q}_t)\big]^2-\frac{1}{n_t}{\rm Tr}(\Sigma_t \bar{Q}_t)^2\Big\}^{-1},\\
\tilde{\rho}_t&=\frac{1}{p}{\rm Tr}[\Sigma_t\bar{Q}_tB_{t+1}\bar{Q}_t]=\frac{1}{p}{\rm Tr}[H_{t,t}B_{t+1}],\\
\bar{B}_t&=\bar{Q}_t\cdots\bar{Q}_T\Sigma_0\bar{Q}_T\cdots \bar{Q}_t.
\end{align*}
Notice that $B_t=Q_t B_{t+1}Q_t$ and $B_{t+1}$ has bounded operator norm, by Theorem \ref{25}, we have
\begin{align*}
B_t \asymp &\bar{Q}_tB_{t+1}\bar{Q}_t+\gamma_t\tilde{\rho}_t\tilde{\mu}_tH_{t,t}\\
 \asymp &\bar{Q}_t(\bar{Q}_{t+1}B_{t+2}\bar{Q}_{t+1}+\gamma_{t+1}\tilde{\rho}_{t+1}\tilde{\mu}_{t+1}H_{t+1,t+1})\bar{Q}_t+\gamma_t\tilde{\rho}_t\tilde{\mu}_tH_{t,t}\\
=&\bar{Q}_t\bar{Q}_{t+1}B_{t+2}\bar{Q}_{t+1}\bar{Q}_t+\gamma_{t+1}\tilde{\rho}_{t+1}\tilde{\mu}_{t+1}H_{t,t+1}+\gamma_t\tilde{\rho}_t\tilde{\mu}_tH_{t,t}\\
\asymp& \cdots \asymp \bar{Q}_t\cdots\bar{Q}_{T}B_{T+1}\bar{Q}_{T}\cdots\bar{Q}_t+\sum_{s=t}^T\gamma_s\tilde{\rho}_s\tilde{\mu}_sH_{t,s}\\
=&\bar{B}_t+\sum_{s=t}^T\gamma_s\tilde{\rho}_s\tilde{\mu}_sH_{t,s}.
\end{align*}
For $1\le t\le T-1$, since $H_{t,t}$ has bounded operator norm, we have
\begin{align*}
H_{t,t}^{1/2}B_{t+1}H_{t,t}^{1/2} \asymp H_{t,t}^{1/2}\bar{B}_{t+1}H_{t,t}^{1/2}+\sum_{s=t+1}^T\gamma_s\tilde{\rho}_s\tilde{\mu}_sH_{t,t}^{1/2}H_{t+1,s}H_{t,t}^{1/2},
\end{align*}
by the property of deterministic equivalence,
\begin{align*}
\tilde{\rho}_t=\frac{1}{p}{\rm Tr}[H_{t,t}B_{t+1}] \rightarrow\frac{1}{p}{\rm Tr}[H_{t,t}\bar{B}_{t+1}]+\sum_{s=t+1}^T\gamma_s\tilde{\rho}_s\tilde{\mu}_s\frac{1}{p}{\rm Tr}[H_{t,t}H_{t+1,s}], \quad a.s.
\end{align*}
By Assumption \ref{51} and Assumption \ref{52}, we have $\frac{1}{p}{\rm Tr}[H_{t,t}\bar{B}_{t+1}]\rightarrow a_t, \frac{1}{p}{\rm Tr}[H_{t,t}H_{t+1,s}]\rightarrow a_{t,s},\tilde{\mu}_t \rightarrow \mu_t$,
thus $\tilde{\rho}_t \rightarrow \rho_t,a.s.$, where $\rho_t$ is defined recursively by $\rho_T=a_T$ and
\begin{align*}
\rho_t=a_t+\sum_{s=t+1}^T\gamma_s\mu_sa_{t,s}\rho_s,\quad 1\le t \le T-1.
\end{align*}
Therefore, we have
\begin{align*}
B_{1} \asymp \bar{B}_{1}+\sum_{t=1}^T\gamma_t\rho_t\mu_tH_{1,t},
\end{align*}
similar to the proof of Theorem \ref{5}, the deterministic equivalence above yields
\begin{align*}
\beta^\top B_{1}\beta -\beta^\top\Big[\bar{B}_{1}+\sum_{t=1}^T\gamma_t\rho_t\mu_tH_{1,t}\Big]\beta \rightarrow 0,\quad a.s.
\end{align*}
and by definition of $G_n$,we have almost surely
\begin{align*}
\beta^\top\Big[\bar{B}_{1}+\sum_{t=1}^T\gamma_t\rho_t\mu_tH_{1,t}\Big]\beta=& \Vert\beta \Vert^2 \Big\{\int\frac{s_0}{\prod_{t=1}^T\lambda_t^2(1+\tilde{m}_ts_t)^2}dG_n(\textbf{s})+\\ &\sum_{t=1}^T\gamma_t\mu_t\rho_t\int\frac{s_t}{\prod_{j=1}^t\lambda_j^2(1+\tilde{m}_js_j)^2}dG_n(\textbf{s})\Big\}\\
\rightarrow & r^2 \Big(g_0+ 
\sum_{t=1}^T\gamma_t\mu_t\rho_tg_{1,t}\Big),
\end{align*}
therefore,
\begin{align*}
   B_X(\hat{\beta}_T; \beta) \rightarrow  \tilde{B}_T(r, \bm{\gamma}, \bm{\lambda},G,H):=r^2 \Big(\prod_{j=1}^T\lambda_j^2\Big)\Big(g_1+ 
   \sum_{t=1}^T\gamma_t\mu_t\rho_tg_{1,t}\Big),\quad a.s.
\end{align*}

\subsubsection{Variance Term}
For $1 \le s < t \le T$, Let $C_{s,t}= Q_t\cdots Q_{s+1}Q_s Q_{s+1} \cdots Q_t, D_{s,t}= Q_t\cdots Q_{s+1}Q_s^2Q_{s+1} \cdots Q_t$, then
\begin{align}
V_X(\hat{\beta}_T;\beta)=\sigma^2\sum_{t=1}^{T} \big(\prod_{s=t+1}^T\lambda_s^2\big)\cdot\frac{1}{n_t}{\rm Tr}\big[(C_{t,T}-\lambda_t D_{t,T})\Sigma_0\big].
\label{42}
\end{align}

For $s<t$, let $G_{s,t}=\bar{Q}_t\cdots\bar{Q}_{s}\Sigma_s\bar{Q}_{s}\cdots\bar{Q}_t$, $\bar{C}_{s,t}=\bar{Q}_t\cdots\bar{Q}_{s+1}\bar{Q}_s\bar{Q}_{s+1}\cdots\bar{Q}_t$, $\tilde{\rho}_{s,t}^{(1)}=p^{-1}{\rm Tr}(G_{t,t}C_{s,t-1})$. Notice that $C_{s,t}=Q_t C_{s,t-1}Q_t$ and $C_{s,t-1}$ has bounded operator norm, by Theorem \ref{25}, we have
\begin{align*}
C_{s,t} \asymp &\bar{Q}_tC_{s,t-1}\bar{Q}_t+\gamma_t\tilde{\rho}_{s,t}^{(1)}\tilde{\mu}_tG_{t,t}\\
 \asymp &\bar{Q}_t(\bar{Q}_{t-1}C_{s,t-2}\bar{Q}_{t-1}+\gamma_{t-1}\tilde{\rho}_{s,t-1}^{(1)}\tilde{\mu}_{t-1}G_{t-1,t-1})\bar{Q}_t+\gamma_t\tilde{\rho}_{s,t}^{(1)}\tilde{\mu}_tG_{t,t}\\
\asymp&\bar{Q}_t\bar{Q}_{t-1}C_{s,t-2}\bar{Q}_{t-1}\bar{Q}_t+\gamma_{t-1}\tilde{\rho}_{s,t-1}^{(1)}\tilde{\mu}_{t-1}G_{t-1,t}+\gamma_t\tilde{\rho}_{s,t}^{(1)}\tilde{\mu}_tG_{t,t}\\
\asymp& \cdots \asymp \bar{Q}_t\cdots\bar{Q}_{s+1}Q_{s}\bar{Q}_{s+1}\cdots\bar{Q}_t+\sum_{j=s+1}^t\gamma_j\tilde{\rho}_{s,j}^{(1)}\tilde{\mu}_jG_{j,t}\\
\asymp&\bar{C}_{s,t}+\sum_{j=s+1}^t\gamma_j\tilde{\rho}_{s,j}^{(1)}\tilde{\mu}_jG_{j,t}.
\end{align*}
Since $G_{t,t}$ has bounded operator norm, we have
\begin{align*}
G_{t,t}^{1/2}C_{s,t-1}G_{t,t}^{1/2} \asymp G_{t,t}^{1/2}\bar{C}_{s,t-1}G_{t,t}^{1/2}+\sum_{j=s+1}^t\gamma_j\tilde{\rho}_{s,j}^{(1)}\tilde{\mu}_jG_{t,t}^{1/2}G_{j,t-1}G_{t,t}^{1/2}.
\end{align*}
If we define $\tilde{\rho}_{s,s}^{(1)}=0$, we have
\begin{align*}
\tilde{\rho}_{s,t}^{(1)}=\frac{1}{p}{\rm Tr}[G_{t,t}C_{s,t-1}] \rightarrow\frac{1}{p}{\rm Tr}[G_{t,t}\bar{C}_{s,t-1}]+\sum_{j=s}^{t-1}\gamma_j\tilde{\rho}_{s,j}^{(1)}\tilde{\mu}_j\frac{1}{p}{\rm Tr}[G_{t,t}G_{j,t-1}], \quad a.s.,
\end{align*}
By Assumption 4 and Assumption 5, we have $\frac{1}{p}{\rm Tr}[G_{t,t}\bar{C}_{s,t-1}]\rightarrow b_{s,t},\frac{1}{p}{\rm Tr}[G_{t,t}G_{j,t-1}]\rightarrow a_{j,t}$, thus $\tilde{\rho}_{s,t}^{(1)} \rightarrow \rho_{s,t}^{(1)},a.s.$, where $\rho_{s,t}^{(1)}$ is defined recursively by $\rho_{s,s}^{(1)}=0$ and
\begin{align*}
\rho_{s,t}^{(1)}=b_{s,t}+\sum_{j=s}^{t-1}\gamma_j\mu_ja_{j,t}\rho_{s,j}^{(1)},\quad 1\le s< t \le T.
\end{align*}
Since $\frac{1}{p}{\rm Tr}[\bar{C}_{t,T}\Sigma_0]\rightarrow b_{t}, \frac{1}{p}{\rm Tr}[G_{j,T}\Sigma_0]\rightarrow a_{j}$, and
\begin{align*}
\frac{1}{p}{\rm Tr}[C_{t,T}\Sigma_0]\rightarrow \frac{1}{p}{\rm Tr}[\bar{C}_{t,T}\Sigma_0]+\sum_{j=t}^T \gamma_j\mu_j\rho_{t,j}^{(1)}\frac{1}{p}{\rm Tr}[G_{j,T}\Sigma_0],\quad a.s.
\end{align*}
we have
\begin{align}
\frac{1}{p}{\rm Tr}[C_{t,T}\Sigma_0]\rightarrow b_t+\sum_{j=t}^T\gamma_j\mu_ja_{j}\rho_{t,j}^{(1)}=L_{1,t},\quad a.s.
\label{40}
\end{align}
Similarly, let $\bar{D}_{s,t}=\bar{Q}_t\cdots\bar{Q}_{s+1}\bar{Q}_s^2\bar{Q}_{s+1}\cdots\bar{Q}_t$, $\tilde{\rho}_{s,t}^{(2)}=p^{-1}{\rm Tr}(G_{t,t}D_{s,t-1})$, we have
\begin{align*}
D_{s,t} \asymp &\bar{Q}_t\cdots\bar{Q}_{s+1}Q_{s}^2\bar{Q}_{s+1}\cdots\bar{Q}_t+\sum_{j=s+1}^t\gamma_j\tilde{\rho}_{s,j}^{(2)}\tilde{\mu}_jG_{j,t}\\
\asymp&\bar{D}_{s,t}+\frac{1}{n_s}{\rm Tr}[G_{s,s}]\tilde{\mu}_sG_{s,t}+\sum_{j=s+1}^t\gamma_j\tilde{\rho}_{s,j}^{(2)}\tilde{\mu}_jG_{j,t}.
\end{align*}
If we define $\tilde{\rho}_{s,s}^{(2)}=\frac{1}{p}{\rm Tr}[G_{s,s}]$, we have
\begin{align*}
\tilde{\rho}_{s,t}^{(2)}=\frac{1}{p}{\rm Tr}[G_{t,t}D_{s,t-1}] \rightarrow\frac{1}{p}{\rm Tr}[G_{t,t}\bar{D}_{s,t-1}]+\sum_{j=s}^{t-1}\gamma_j\tilde{\rho}_{s,j}^{(2)}\tilde{\mu}_j\frac{1}{p}{\rm Tr}[G_{t,t}G_{j,t-1}], \quad a.s.,
\end{align*}
By Assumption \ref{51} and Assumption \ref{52}, we have $\frac{1}{p}{\rm Tr}[G_{t,t}\bar{D}_{s,t-1}]\rightarrow c_{s,t},\frac{1}{p}{\rm Tr}[G_{t,t}G_{j,t-1}]\rightarrow a_{j,t}$, and $\frac{1}{p}{\rm Tr}[G_{s,s}]=c_{s,s}$, thus $\tilde{\rho}_{s,t}^{(2)} \rightarrow \rho_{s,t}^{(2)},a.s.$, where $\rho_{s,t}^{(2)}$ is defined recursively by $\rho_{s,s}^{(2)}=c_{s,s}$ and
\begin{align*}
\rho_{s,t}^{(2)}=c_{s,t}+\sum_{j=s}^{t-1}\gamma_j\mu_ja_{j,t}\rho_{s,j}^{(2)},\quad 1\le s< t \le T.
\end{align*}
Since $\frac{1}{p}{\rm Tr}[\bar{D}_{t,T}\Sigma_0]\rightarrow c_{t}$, we obtain 
\begin{align}
\frac{1}{p}{\rm Tr}[D_{t,T}\Sigma_0]\rightarrow c_t+\sum_{j=t}^T\gamma_j\mu_ja_{j}\rho_{t,j}^{(2)}=L_{2,t},\quad a.s.
\label{41}
\end{align}
Substituting (\ref{40}) and (\ref{41}) in (\ref{42}), we get
\begin{align*}
   V_X(\hat{\beta}_T; \beta) \rightarrow \tilde{V}_T(\bm{\gamma}, \bm{\lambda},H):=\sigma^2
\sum_{t=1}^{T}\gamma_{t}\big(\prod_{s=t+1}^T\lambda_s^2\big)(L_{1,t}-\lambda_tL_{2,t}).
\end{align*}

\section{Auxiliary Lemmas}
\begin{lemma}[Lemma B.26, \cite{book1}]
Let $A \in \mathbb{R}^{n \times n}$ be nonrandom matrix, $x=(x_1, \cdots, x_n)^\top\in\mathbb{R}^{n}$ be random vector with independent entries. Assume that $\mathbb{E}x_i=0,\mathbb{E}|x_i|^2=1$, and $\mathbb{E}|x_i|^l\le v_l$. Then, for any $k \ge 1$,
\begin{align*}
\mathbb{E}|x^\top A x-{\rm Tr}A|^k \le C_k\Big[\big(v_4{\rm Tr}(AA^\top)\big)^{k/2}+v_{2k}{\rm Tr}(AA^\top)^{k/2}\Big],
\end{align*}
for some $C_k>0$.
\label{32}
\end{lemma}

\begin{lemma}[Lemma 2.13, \cite{book1}]
Let $x_i$ be a martingale difference sequence with
respect to the increasing $\sigma$-field $\{\mathcal{F}_i\}$. Then for $k\ge 2$,
\begin{align*}
\mathbb{E}\Big|\sum_{i=1}^n x_i\Big|^k \le C_k \Big(\mathbb{E}\big(\sum_{i=1}^n \mathbb{E}_{\le i-1}|x_i|^2\big)^{k/2}+\sum_{i=1}^n\mathbb{E}|x_i|^k\Big).
\end{align*}
for some $C_k>0$, where $\mathbb{E}_{\le i}$ denotes the expectation taken over $\mathcal{F}_i$.
\label{33}
\end{lemma}

\begin{lemma}
Let $A$ be a deterministic matrix with bounded operator norm, $Q$ be the resolvent matrix defined in Theorem \ref{25}. Then for $k\ge 2$,
\begin{align*}
\mathbb{E}\Big|\frac{1}{p}{\rm Tr}[A(Q-\mathbb{E}Q)]\Big|^k =O(n^{-k/2}).
\end{align*}
for some $C_k>0$.
\label{34}
\end{lemma}
\begin{proof}
Let $Q_{-i}=(\frac{1}{n}\sum_{j\not= i}x_j x_j^\top-zI_p)^{-1}$ be the resolvent matrix without sample $x_i$. Note that
\begin{align*}
\frac{1}{p}{\rm Tr}[A(Q-\mathbb{E}Q)] =& \sum_{i=1}^n\frac{1}{p}{\rm Tr}[A\mathbb{E}_{\le i}(Q)]-\sum_{i=1}^n\frac{1}{p}{\rm Tr}[A\mathbb{E}_{\le i-1}(Q)] \\
= & \frac{1}{p}\sum_{i=1}^n(\mathbb{E}_{\le i}-\mathbb{E}_{\le i-1}){\rm Tr}[A(Q-Q_{-i})].
\end{align*}
where $\mathbb{E}_{\le i}$ is the expectation taken over $\mathcal{F}_i$ generated by $z_1 \cdots, z_i$, and $\mathbb{E}_{\le 0}Q=Q$. By Sherman-Morrison formula, we have
\begin{align}
Q=Q_{-i}-\frac{\frac{1}{n}Q_{-i}x_ix_i^\top Q_{-i}}{1+\frac{1}{n}x_i^\top Q_{-i}x_i},
\label{29}
\end{align}
thus
\begin{align*}
\frac{1}{p}{\rm Tr}[A(Q-Q_{-i})] =\frac{1}{pn}\cdot\frac{x_i^\top Q_{-i}AQ_{-i}x_i}{1+\frac{1}{n}x_i^\top Q_{-i}x_i}.
\end{align*}
Let $Y_i=(\mathbb{E}_{\le i}-\mathbb{E}_{\le i-1})\frac{1}{p}{\rm Tr}[A(Q-Q_{-i})]$, then $\{Y_i\}$ is a martingale difference sequence. Since
\begin{align*}
\Big|\frac{1}{pn}\cdot\frac{x_i^\top Q_{-i}AQ_{-i}x_i}{1+\frac{1}{n}x_i^\top Q_{-i}x_i}\Big|\le \Big|\frac{1}{p}\cdot\frac{x_i^\top Q_{-i}AQ_{-i}x_i}{x_i^\top Q_{-i}x_i}\Big|
\le \frac{1}{p}\cdot\Vert Q_{-i}^{1/2}AQ_{-i}^{1/2}\Vert_{op}
\le p^{-1}\lambda^{-1}\Vert A\Vert_{op},
\end{align*}
we have $|Y_i| \le 2p^{-1}\lambda^{-1}\Vert A\Vert_{op}=O(p^{-1})$. By Lemma \ref{33},
\begin{align*}
\mathbb{E}\Big|\frac{1}{p}{\rm Tr}[A(Q-\mathbb{E}Q)]\Big|^k\le C_k \Big(\mathbb{E}\big(\sum_{i=1}^n \mathbb{E}_{\le i-1}|Y_i|^2\big)^{k/2}+\sum_{i=1}^n\mathbb{E}|Y_i|^k\Big)=O(n^{-k/2}).
\end{align*}
\end{proof}

\vskip 0.2in
\bibliography{sample}

\end{document}